\documentclass[a4paper, 10pt]{amsart}
\usepackage{fullpage}
\usepackage[utf8]{inputenc}

\usepackage{amsmath}
\usepackage{amsthm}
\usepackage{ mathrsfs }
\usepackage{ amssymb }
\usepackage{mathtools}
\usepackage{enumitem}  

\usepackage{tikz}
\usetikzlibrary{arrows}

\usepackage{geometry}

\usepackage{bigstrut}
\usepackage{multirow}
\usepackage{listings} 
\usepackage{tikz-cd}		
\usepackage{graphicx}
\newtheorem{thm}{Theorem}[section]
\newtheorem{lem}[thm]{Lemma}

\newtheorem{cor}[thm]{Corollary}
\newtheorem{prop}[thm]{Proposition}

\newtheorem*{thm*}{Theorem}
\newtheorem*{conj}{Conjecture}
\newtheorem*{cor*}{Corollary}

\theoremstyle{definition}

\newtheorem{defi}[thm]{Definition}

\theoremstyle{remark}
\newtheorem{rmk}[thm]{Remark}

\newtheorem*{Ack}{Acknowledgment}

\newcommand{\Hom}{\operatorname{Hom}}
\newcommand{\Map}{\operatorname{Map}}

\def\C{\mathbb C}

\def\N{\mathbb N}
\def\Z{\mathbb Z}
\def\Q{\mathbb Q}

\title{Homotopy transfer theorem and KZB connections}
\author{Claudio Sibilia}
\begin{document}

\maketitle
\begin{abstract}
 We show that the KZB connection on the punctured torus and on the configuration space of points of the punctured torus can be constructed via the homotopy transfer theorem.  
\end{abstract}




\section*{Introduction}
Let $M$ be a smooth complex manifold. We denote by $A_{DR}^{\bullet}(M)$ the differential graded algebra of complex differential forms on $M$. For any pronilpotent Lie algebra $\mathfrak{g}$, each 
\[
C\in A_{DR}^{1}(M)\widehat{\otimes}\mathfrak{g}
\] 
defines a connection $d-C$ on the trivial bundle $M\times\mathfrak{g}$, where the latter is considered to be equipped with the adjoint action. We consider $A_{DR}(M)\widehat{\otimes}\mathfrak{g}$ to be equipped with the differential (graded) Lie algebra structure
\[
d(v\otimes X):=dv\otimes X,\quad\left[w\otimes X, v\otimes Y \right]:=wv\otimes [X,Y].
\]
If $C$ is a Maurer-Cartan element, i.e. $dC+\frac{[C,C]}{2}=0$, then the resulting connection $d-C$ is flat.
The configuration space of points of a topological space is defined as 
\[
\operatorname{Conf}_{n}\left(X \right):=\left\lbrace \left(x_{1}, \dots, x_{n} \right)\in X^{n}\: :\:  x_{i}\neq x_{j}\text{ for }i\neq j \right\rbrace. 
\]
Let $A_{KZ,n}$ be the unital differential graded subalgebra of $A_{DR}\left( \operatorname{Conf}_{n}\left(\C-\left\lbrace0,1 \right\rbrace  \right)\right)  $ generated by $\omega_{i,-1}$, $\omega_{i,0}$, and $\omega_{i,j}$ where
\[
\omega_{ij}:=d \log(z_{i}-z_{j})=\frac{dz_{i}-dz_{j}}{z_{i}-z_{j}}
\] 
for $1\leq i\neq j\leq n$ such that $z_{-1}:=1$ and $z_{0}:=0$. In \cite{Arnold}, it was shown that the $A_{KZ,n}$ is a model for $A_{DR}\left( \operatorname{Conf}_{n}\left(\C-\left\lbrace0,1 \right\rbrace  \right)\right)  $.
We define the rational Lie algebra $\mathfrak{t}_{n}$ with generators $T_{i,j}$ for $-1\leq i\neq j\leq n$ with $j>0$ or $i>0$ such that
\begin{eqnarray}
T_{ij}=T_{ji}, \quad \left[ T_{ij},T_{ik}+T_{jk}\right]=0, \quad  \left[ T_{ij},T_{kl}\right]=0
\end{eqnarray}
for $i,j,k,l$ distinct. We call $\mathfrak{t}_{n}$ the \emph{Kohno-Drinfeld Lie algebra}. The KZ connection 
is the flat connection $d-\omega_{KZ,n}$ on $\operatorname{Conf}_{n}\left(\C-\left\lbrace 0,1\right\rbrace\right) \times \mathfrak{t}_{n} $ where
\[
\omega_{KZ,n}:=\sum_{-1\leq i\neq j\leq n, i>0\text{ or }j>0} \omega_{ij}T_{ij}.
\]
This connection can be constructed via Chen's theory of formal power series connections (see \cite{extensionChen}). The KZ connection is the \emph{universal} version of the Knizhnik–Zamolodchikov equation (see \cite{KZ}), i.e.  a differential equation used in quantum field theory. The KZ equation has an elliptic version called \emph{Knizhnik–Zamolodchikov-Bernard equation} (KZB equation) given in \cite{Be1}. The \emph{universal KZB connection} is constructed in \cite{Damien} (compare with \cite{Levrac}). It is a holomorphic connection on the configuration space of points of the punctured torus. Its fiber corresponds to the Malcev completion of the  fundamental group of the base space (i.e. the pure elliptic braid group).\\
Let $G$ be a group acting properly and discontinuously on a complex manifold $M$. The action groupoid of $(M,G)$ is a simplicial manifold. Getzler and Cheng (see \cite{Getz}) have shown that the space of differential forms on such a simplicial manifold carries a natural $C_{\infty}$-algebra structure. We denote this $C_{\infty}$-algebra by $\left( \operatorname{Tot}_{N}\left(A_{DR}(M_{\bullet}G)\right),m_{\bullet}\right)$. There is a strict natural $C_{\infty}$-map 
\[
r\: : \: \left( \operatorname{Tot}_{N}\left(A_{DR}(M_{\bullet}G)\right),m_{\bullet}\right)\to \left(A_{DR}(M) , d \wedge\right) 
\]
induced by the canonical map $M\to M_{\bullet}G$. We extend the $G$-action on $A_{DR}(M)\widehat{\otimes}\mathfrak{g}$ by assuming that $\mathfrak{g}$ is $G$-invariant. In \cite{Sibilia1}, we have shown the following facts. 
\begin{enumerate}
\item Let $C$ be a Maurer Cartan element which is gauge equivalent to a $G$-invariant Maurer-Cartan element $C'$. There exists a vector bundle $E$ on $M/G$ such that $d-C$ is a well-defined flat connection on $M/G$ (\cite[Theorem 4.9]{Sibilia1}). 
\item  In \cite{Sibilia1} we introduce the category of $1-C_{\infty}-$ algebras and $1-C_{\infty}-$morphism, these are a truncated version of $C_{\infty}$-algebras. In particular there is a forgetful functor $\mathcal{F}$ between $C_{\infty}$-algebras and $1-C_{\infty}$-algebras. We introduce the notion of $1$-isomorphism, $1$-quasi-isomorphism of $1-C_{\infty}$-algebras and $1$-minimal model. Let $(W, m_{\bullet}^{W})$ be a $1$-minimal $1-C_{\infty}$-algebra quasi-isomorphic to $\mathcal{F}\left( \operatorname{Tot}_{N}\left(A_{DR}(M_{\bullet}G)\right),m_{\bullet}\right)$. For any $1-C_{\infty}$-morphism $g_{\bullet}\: : \: (W, m_{\bullet}^{W})\to \mathcal{F}\left( \operatorname{Tot}_{N}\left(A_{DR}(M_{\bullet}G)\right),m_{\bullet}\right) $, the map $rg_{\bullet}$ gives an explicit Maurer-Cartan element in $A_{DR}(M)\widehat{\otimes}\mathfrak{u}$, where $\mathfrak{u}$ corresponds to the Malcev completion of the fundamental group of $M/G$ (\cite[Theorem 3.16]{Sibilia1}).
\end{enumerate}
Examples of morphisms $g_{\bullet}\: : \: (W, m_{\bullet}^{W})\to \mathcal{F}\left( \operatorname{Tot}_{N}\left(A_{DR}(M_{\bullet}G)\right),m_{\bullet}\right) $ can be constructed via the homotopy transfer theorem (see  \cite{Prelie},\cite{Getz}, \cite{Kadesh}, \cite{kontsoibel} and \cite{Markl}). The homotopy transfer theorem is a general version of the of Chen's formal power series connections (see \cite{Huebsch}). In particular the KZ connection can be constructed via the homotopy transfer theorem.
\begin{itemize}
	\item The punctured torus and the configuration space of points on the punctured torus can be written as $M/G$ for some appropriate $M$ and $G$. We use the homotopy transfer theorem to construct two $1$-morphism $g_{\bullet}$ and ${g'}_{\bullet}$. By point (2), the maps $rg_{\bullet}$ and ${{g'}_{\bullet}}$ correspond to two gauge equivalent Maurer-Cartan elements $C$ and $C'$. By point (1), we get a flat connection on $M/G$. In particular, the KZB connection on the punctured torus and on the configuration space of points on the punctured torus can be computed in this way (see Theorem \ref{main1} and Theorem \ref{final} respectively). In order to compute the connection we have to choose a $1$-model ($1$-extension) and a particular vector space decomposition. However, we show that different choices produce isomorphic connections (Corollary \ref{invariance}).
	\item  We construct a Lie algebra morphism $Q^{*}$ between the Kohno-Drinfeld Lie algebra and the Malcev Lie algebra of the fundamental group of the configuration space of points on the punctured torus. Furthermore we show that if $\tau\to i\infty$, then the KZB connection becomes equal to the ``KZ connection composed with $Q^{*}$'' (see Theorem \ref{thmfinale}). This is a generalization of the Lie algebra morphism defined by Hain in \cite[Section 12]{Hain}.
	\end{itemize}
For $g> 1$, the higher genus version of the KZB equation is constructed by Bernard in \cite{Be2}, where it was used the fact that higher genus Riemann surfaces can be written as a quotient $M/G$, for an appropriate $M$ and a Schottky group $G$. This suggest the following conjecture.
\begin{conj}
Let $g>1$. The universal version of the KZB equation constructed in \cite{Be2} is a holomorphic flat connection which can be constructed via the homotopy transfer theorem, (1) and (2).
\end{conj}

\subsection*{Plan of the paper}
The paper is divided in 4 sections. The first one presents some standard results about the homotopy transfer theorem.
In the second section we prove the above statement in the case of the punctured torus. More precisely, we fix a  $\tau\in\mathbb{H}:=\left\lbrace z\in \C \: : \: \Im(z)>0\right\rbrace $ and let $\mathbb{Z}+\tau\mathbb{Z}$ be the lattice spanned by $1,\tau$. We define
\[
\mathcal{E}_{\tau}^{\times}=\left(\C- \left\lbrace \mathbb{Z}+\tau\mathbb{Z}\right\rbrace \right)/ \Z^{2}.
\]
A smooth model $A$ for $A_{DR}\left(\mathcal{E}_{\tau}^{\times}\right) $ is constructed in \cite{LevBrown}. We apply the homotopy transfer theorem on $A$ and we get a $C_{\infty}$-morphism ${g'}_{\bullet}$. $\mathcal{E}_{\tau}^{\times}$ is the geometric realization of the simplicial manifold $\left(\C- \left\lbrace \mathbb{Z}+\tau\mathbb{Z}\right\rbrace \right)_{\bullet} \Z^{2}$. In this subsection we construct a \emph{holomorphic} $1$-model $B$ for the $C_{\infty}$-algebra of differential forms on $\left(\C- \left\lbrace \mathbb{Z}+\tau\mathbb{Z}\right\rbrace \right)_{\bullet} \Z^{2}$ (see Corollary \ref{cormodel}).  We apply the homotopy transfer theorem on $B$ and we get a $1-C_{\infty}$-morphism ${g}_{\bullet}$. The main results of this section is the following: we prove that the KZB connection can be constructed by using (1) where the $C$ and $C'$ corresponds to $rg_{\bullet}$ and ${g'}_{\bullet}$ respectively, in particular  the KZB connection on the configuration space of points of the punctured torus can be computed via the homotopy transfer theorem on $B$ (see Theorem \ref{main2}). We compare the KZ and the KZB connection for $\tau\to i \infty$ and we show that the Lie algebra morphism $Q^*$ constructed by in \cite[Section 12]{Hain}) can be constructed via $C_{\infty}$-algebra methods (see Subsection \ref{ComparisonHain}). In Subsection \ref{Arationalconnection}, we present another holomorphic connection on the punctured torus (with singularities of order $2$) which is isomorphic to the KZB connection.\\
In the third section, we extend the above results for the configuration spaces of points of the punctured elliptic curve. Let $(\xi_{1}, \dots \xi_{n})$ be the coordinates on $\C^{n}$. We define $\mathcal{D}  \subset \C^{n}$ as 
\[
\mathcal{D}:=\left\lbrace (\xi_{1}, \dots \xi_{n})\: : \: \xi_{i}-\xi_{j}\in \mathbb{Z}+\tau\mathbb{Z} \text{ for some distinct }i,j=0, \dots n  \right\rbrace.
\] 
We define a $\Z^{2n}$-action on $\C^{n}$ via translation, i.e.
\[
\left( \left({l}_{1},m_{1} \right), \dots, \left({l}_{n},m_{n} \right) \right)(\xi_{1}, \dots \xi_{n}):= \left( \xi_{1}+l_{1}+m_{1}\tau, \dots ,\xi_{n}+l_{n}+m_{n}\tau\right).
\]
Notice that $\mathcal{D}$ is preserved by the action of $\Z^{2n}$. There is a canonical isomorphism
\[
\left( \C^{n}-\mathcal{D}\right) /\left( \Z^{2n}\right) \cong \operatorname{Conf}_{n}\left({\mathcal{E}}_{\tau}^{\times} \right)
\]
since the action is free and properly discontinuous. We construct a $C_{\infty}$-algebra ${B'}_{n}$ which is conjectured to be a $1$-model of 
\[
\operatorname{Tot}_{N}\left(A_{DR}\left( \C^{n}-\mathcal{D}\right)_{\bullet}\left( \Z^{2n}\right) \right)\otimes \Omega(1)
\]
where $\Omega(1)$ is the differential graded algebra of polynomial forms on the unit interval. In particular, the evaluation at $1$ of ${B'}_{n}$ corresponds to the smooth model $A_{n}\subset A_{DR}\left( \operatorname{Conf}_{n}\left(\mathcal{E}_{\tau}^{\times}\right)\right)$ constructed in \cite{LevBrown} and the evaluation at $0$ is a holomorphic $C_{\infty}$-algebra. We introduce the notion of $1$-extension which is a weaker notion of $1$-model and we construct a $1$-extension $B_{n}$ of $A_{n}$ by taking the quotient of ${{B'}}_{n}$ by some $C_{\infty}$-ideal. We compute the homotopy transfer theorem on $B_{n}$ and we get a $1-C_{\infty}$-morphism $H_{\bullet}\: : \: (W, m_{\bullet}^{W})\to\mathcal{F}\left( \operatorname{Tot}_{N}\left(A_{DR}\left( \C^{n}-\mathcal{D}\right)_{\bullet}\left( \Z^{2n}\right) \right)\otimes \Omega(1)\right) $. The evaluation of $H$ at $1$ can be lifted to a minimal model $C_{\infty}$-morphism ${g'}_{\bullet}\: : \: (W, m_{\bullet}^{W})\to A_n$ and the evaluation at $0$ gives a holomorphic $1-C_{\infty}$-morphism ${g}_{\bullet}\: : \: (W, m_{\bullet}^{W})\to \mathcal{F}\left(\operatorname{Tot}_{N}\left(A_{DR}\left( \C^{n}-\mathcal{D}\right)_{\bullet}\left( \Z^{2n}\right) \right)\right) $. The main results of this section is the following: we prove that the KZB connection on the configuration space of points of the punctured torus can be constructed by using (1) where the $C$ and $C'$ corresponds to $rg_{\bullet}$ and ${g'}_{\bullet}$, in particular the KZB connection on the punctured torus can be computed via the homotopy transfer theorem on $B_{n}$ ( see Theorem \ref{final}). In Subsection \ref{KZKZB}, we show that the KZB connection for $\tau\to i \infty$ corresponds to the KZ connection modulo a Lie algebra morphism $Q^{*}$ which is explicitly computed. For $n=1$, $Q^*$ coincides with the automorphism constructed by Hain in \cite[Section 12]{Hain}). 
\begin{Ack}
This article is based on the second part of my PhD project. I wish to thank my two advisors Damien Calaque and Giovanni Felder for their guidance and constant support and Benjamin Enriquez for his interesting questions about my PhD thesis. I thanks the SNF for providing an essential financial support through the ProDoc module PDFMP2 137153 ``Gaudin subalgebras, Moduli Spaces and Integrable Systems". I also thank the support of the ANR SAT and of the Institut Universitaire de France.
\end{Ack}

\subsection*{Notation}
Let $\Bbbk$ be a field of characteristic zero.
 For a graded vector space $V^{\bullet}$, $V^{i}$ is called the homogeneous component of $V$, and for $v\in V^{i}$ we define its degree via $|v|:=i$. For a vector space $W:=\oplus_{i\in I} W_{i}$ we denote by ${pro}_{W_{i}}\: : \: W\to W_{i}$ the projection. For a graded vector space $V^{\bullet}$ we denote by $V[n]$ the $n$-shifted graded vector space, where $\left( V[n]\right) ^{i}=V^{n+i}$. For example $\Bbbk[n]$ is a graded vector space concentrated in degree $-n$ (its $-n$ homogeneous component is equal to $\Bbbk$, the other homogeneous component are all equal to zero). A (homogeneous) morphism of graded vector spaces $f\: : \: V^{\bullet}\to W^{\bullet}$ of degree $|f|:=r$ is a linear map such that $f(V^{i})\subseteq V^{i+r}$. We denote by $s\: : \: V\to V[1]$, $s^{-1}\: : \: V[1]\to V$  the shifting morphisms that send $V^{n}$ to $V[1]^{n-1}=\Bbbk \otimes V^{n}=V^{n}$ resp. $V[1]^{n}=\Bbbk \otimes V^{n+1}=V^{n+1}$ to $V^{n+1}$. Those maps can be extended to a map $s^{n}\: : \: V\to V[n]$, (the identity map shifted by $n$).
Note that $s^{n}\in \Hom^{-n}\left(V, V[n] \right)$. A graded vector space is said to be of finite type if each homogeneous component is a finite vector space. A graded vector space $V^{\bullet}$ is said to be bounded below at $k$ if there is a $k$ such that $V^{l}=0$ for $l<k$. Analogously it is said to be bounded above at $k$ if there is a $k$ such that $V^{l}=0$ for $l>k$. For a non-negatively graded vector space we define the positively graded vector space
$W^{\bullet}_{+}$ as
\[
W_{+}^{0}:=0,\quad W_{+}^{i}:=W^{i}\text{ if }i\neq0
\]
Let $(V,d_{V})$ be a differential graded vector space, then $V^{\otimes n}$ is again a differential graded vector space with differential
\[
d_{V^{\otimes n}}(v_{1}\otimes\cdots\otimes  v_{n}):=\sum_{i=1}^{n}\pm v_{1}\otimes\cdots\otimes  d_{V}v_{i}\cdots\otimes  v_{n},
\]
where the signs follow from the Koszul signs rule.
Let $(V,d_{V})$, $(W,d_{W})$ be differential graded vector spaces, then $\Hom^{\bullet}_{gVect}\left(V,W \right)$ is a differential graded vector space with differential
\[
\partial f:= d_{W}f-(-1)^{|f|}fd_{v}.
\]

\section{The homotopy transfer theorem}
We use the same notation of \cite{Sibilia1}. We give a short introduction about the homotopy transfer theorem for $C_{\infty}$-algebras (see \cite{Getz}, \cite{Kadesh}, \cite{kontsoibel}, \cite{Markl} and \cite{Prelie}). We work only with non-negatively graded $C_{\infty}$-algebras.
\subsection{Simplicial manifolds}
Let $G$ be a group acting properly and discontinuously on $M$. The nerve gives a simplicial manifold $M_{\bullet}G$ is called the action groupoid. In particular, $A_{DR}\left(M_{\bullet}G \right) $ is a simplicial commutative differential graded algebra. Let $\operatorname{Tot}_{N}^{\bullet}\left(A_{DR}(M_{\bullet}G)\right)$ be the normalized total complex. An element $w$ of degree $n$ can be written as $w=\sum_{p+q=n}w^{p,q}$, where $w^{p,q}$ is a set map
\[
w^{p,q}\: : \: G^{p}\to A_{DR}^{q}(M)
\]
such that $w^{p,q}(g_{1}, \dots, g_{n})=0$ if $g_{i}=e$ for some $i$. The elements $w^{p,q}$ are called elements of bidegree $(p,q)$, in particular $\operatorname{Tot}_{N}^{0,q}\left(A_{DR}(M_{\bullet}G)\right)=A_{DR}^{q}(M)$. An element $w$ is said to be holomorphic if each $w^{p,q}$ takes holomorphic values for any $(p,q)$. Since the action is properly discontinuous, we can visualize $A_{DR}^{\bullet}(M/G)\subset A_{DR}^{\bullet}(M)$ as the differential graded algebra of $G$-invariant differential forms. Getzler and Cheng (see \cite{Getz}) have shown that the normalized total complex $\operatorname{Tot}_{N}\left(A_{DR}\left( M_{\bullet}G\right)\right) $ carries a natural unital $C_{\infty}$-structure $m_{\bullet}$
such that 
\[
m_{1}=d,\quad m_{2}=\wedge,\quad m_{n}=0\text{ for }n\geq 2
\]
on $A_{DR}^{\bullet}(M/G)\subset\operatorname{Tot}_{N}^{0,\bullet}\left(A_{DR}(M_{\bullet}G)\right)$ and
\[
\left( A_{DR}(M/G),d, \wedge\right) \hookrightarrow\left( \operatorname{Tot}_{N}\left(A_{DR}(M_{\bullet}G)\right),m_{\bullet}\right) 
\]
is a quasi-isomorphism of $C_{\infty}$-algebras. Their cohomology corresponds to the singular cohomology of $M/G$. We denote $m_{1}$ by $D$. In particular for an element $a$ of bidegree $(p,q)$ we have
\[
D(a)=\partial_{G}a+(-1)^{p}da,
\]
where $\partial_{G}$ is differential obtained by the alternating sum of the pullback of the cofaces maps of the action groupoid. There is a strict $C_{\infty}$-map 
\[
r\: : \: \left( \operatorname{Tot}_{N}\left(A_{DR}(M_{\bullet}G)\right),m_{\bullet}\right)\to \left(A_{DR}(M) , d, \wedge\right) 
\]
induced by the evaluation of the forms at $e\in G$. Let $g_{\bullet}\: :\: \left(W, m_{\bullet}^{W} \right)\to \mathcal{F}\left( \operatorname{Tot}_{N}\left(A_{DR}(M_{\bullet}G)\right),m_{\bullet}\right)$ be a morphism of $1-C_{\infty}$-algebras. In \cite{Sibilia1}, we show that $g_{\bullet}$ corresponds to a Maurer-Cartan element $\overline{C}$ in the $L_{\infty}$-algebra
\[
\mathcal{F}\left( \operatorname{Tot}_{N}\left(A_{DR}(M_{\bullet}G)\right),m_{\bullet}\right)\widehat{ \otimes}\mathfrak{u}
\]
where $\mathfrak{u}$ is the Malcev Lie algebra of the fundamental group of $M/G$. We call $\overline{C}$ the degree zero geometric connection . In \cite[Theorem 3.16 ]{Sibilia1}, we show that
$r$ induces a map $r_{*}=(r\otimes\mathrm{Id})$ such that $r_*\overline{C}$ is a Maurer-Cartan element in 
\[
A_{DR}(M)\widehat{ \otimes}\mathfrak{u}.
\]
\subsection{Transferring $A_{\infty}$ and $C_{\infty}$-structures}
We recall the homotopy \emph{homotopy transfer theorem} (see \cite{lodayVallette}). This theorem appears originally in \cite{kontsoibel} (see \cite{Markl}) for $A_{\infty}$-algebras and in \cite{Getz} for $C_{\infty}$-algebras.
\begin{thm}\label{markglobal}
Let $(V^{\bullet},d_{V})$, $(W^{\bullet}, d_{W})$ be two cochain differential graded vector spaces. Assume that there are two cochain maps
\begin{equation}\label{homretract}
 \begin{tikzcd}
 f\: : \: (V^{\bullet},d_{V})\arrow[r, shift right, ""]&(W^{\bullet}, d_{W})\: : \: g\arrow[l, shift right, ""] 
 \end{tikzcd}
 \end{equation}
 such that $gf$ is homotopic to $Id_{V}$ via a cochain homotopy $h$. Assume that $(V^{\bullet},d_{V})$ is equipped with a $P_{\infty}$-algebra structures $m_{\bullet}^{V}$, such that $m_{1}^{V}=d_{V}$. There exist
\begin{enumerate}
\item a $(1-)P_{\infty}$-algebra structures $m_{\bullet}^{W}$ on $W^{\bullet}$, such that $m_{1}^{V}=d_{W}$;
\item a $(1-)P_{\infty}$-morphism ${{g}}_{\bullet}\: : \:\left(W, m_{\bullet}^{W}\right)\to \left(V, m_{\bullet}^{V}\right)$, such that ${{g}}_{1}=g$.
\end{enumerate}
\end{thm}
\begin{rmk}\label{transfnonuniq}
The maps $m_{\bullet}$, $ {g}_{\bullet}$ are not unique in general. There may be more solutions. 
\end{rmk}

\begin{rmk}\label{computability}
Consider Theorem \ref{markglobal}. In \cite{Markl} it is proved the following. If $P_{\infty}=A_{\infty}$, there exist
\begin{enumerate}
\item a $A_{\infty}$-morphism ${{f}}_{\bullet}\: : \:\left(V, m_{\bullet}^{V}\right)\to \left(W, m_{\bullet}^{W}\right)$, such that ${{f}}_{1}=f$;
\item a $A_{\infty}$-homotopy ${{h}}_{\bullet}$ between ${{g}}{{f}}$ and $Id_{V}$, such that ${{h}}_{1}=h$, where $h_{\bullet}$ is a chain homotopy in the sense of \cite{Markl}.
\end{enumerate}
Moreover there is an explicit formula for these maps. Point (1) is extended in \cite{Prelie} for more general operads (by assuming that the diagram \eqref{homretract}  satisfies $d_{W}=0$, $fg=1_{W}$, $fh=0$, $h\circ g=0$ and $h^2=0$). 
\end{rmk}

\begin{defi}
An oriented planar rooted tree $T$ is a connected oriented planar graph that contains no loops, such that the orientation goes toward one marked external vertex (the root). Let $V(T)$ be the set of vertices, $E(T)$ the set of (oriented) edges.\\ 
Given an edge $e$, between two vertices $v_{1},v_{2}$, if the orientation goes from $v_{1}$ to $v_{2}$ we call $v_{1}$ the \emph{source} and $v_{2}$ the \emph{target} of $e$, respectively. A edge $e$ is \emph{internal} if its source is the target of another edge. A non internal edge is called \emph{leaf}. The root is the only one vertex which is not the source of any other edges. The \emph{arity} of a vertex is the number of incoming edges.
\end{defi}
We denote by $\mathcal{P}$ the set of finite oriented planar rooted trees where the arity of each internal vertex is $\geq 2$. We denote by $\mathcal{P}_{n}$ the trees in $\mathcal{P}$ with exactly $n$ leaves and with $\mathcal{P}^{2}_{l}$ the set oriented planar rooted trees where the arity of each internal vertex is less then or equal to $l$. We fix a diagram of the type \eqref{homretract}. Each tree can be decorated as follows: we associate to each internal edges the map $h$ and to each internal vertex of arity $k$ the map $m^{V}_{k}$. Figure \ref{decorated} is an example of decorated tree $T'$ where the root is the lower vertex.
\\
\begin{figure}[ht]
\centering
  \includegraphics[scale=.1]{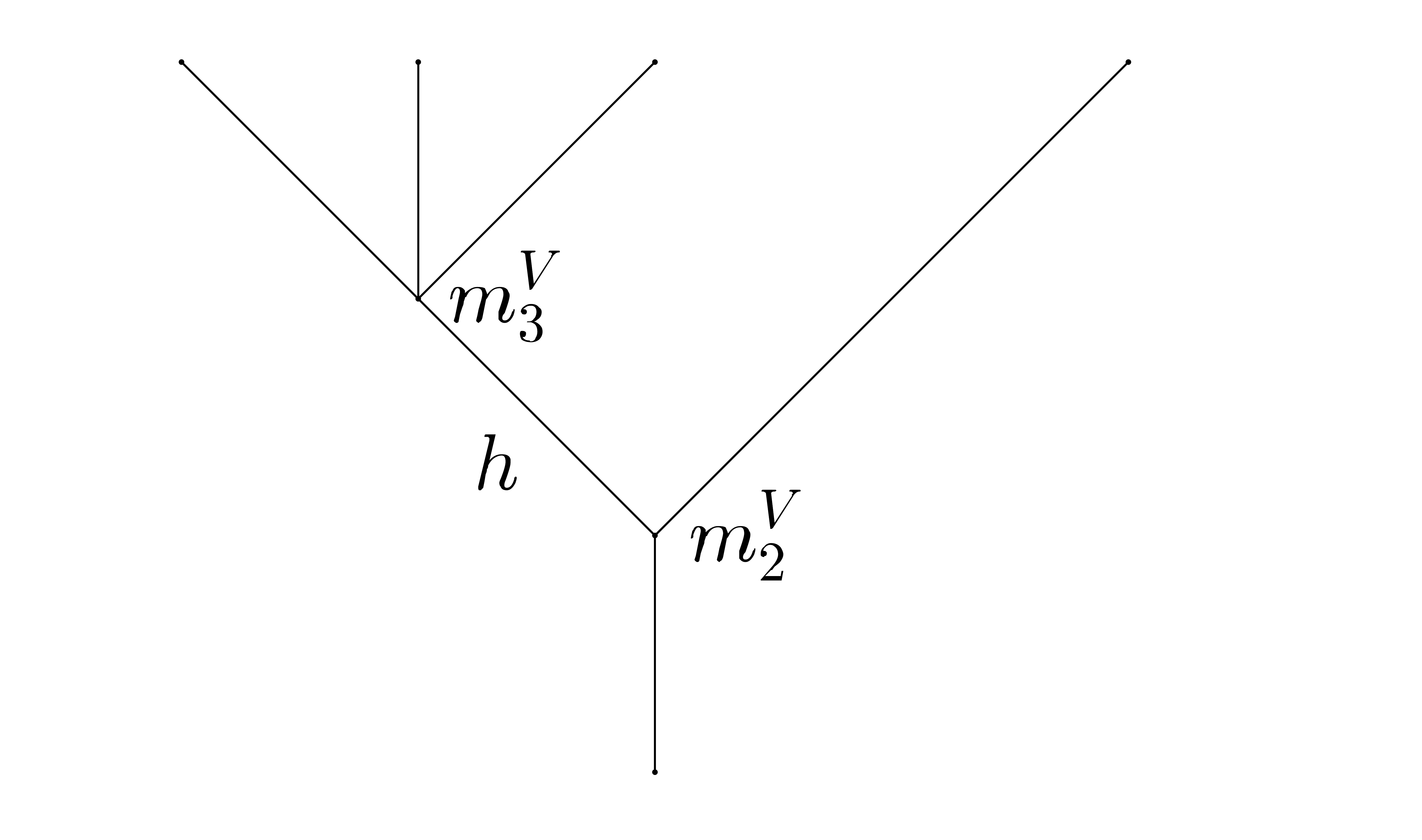}
	\caption{ $T'$ decorated}
	\label{decorated}
\end{figure}
\\
The three above gives a map $\mathcal{P}_{T'}\: : \: V^{\otimes 4}\to V$ via $\mathcal{P}_{T'}:=m_{2}^{V}\circ\left(m_{3}^{V}\otimes Id \right) $; hence to each tree $T\in \mathcal{P}_{n}$ we associate a linear map $\mathcal{P}_{T}\: : \: V^{\otimes n}\to V$ as above.
To each trees we can associated a values $\theta(T)$ as follows. For a vertex $v$ in $T$ with arity $k$, consider the edges $e_{1}, \dots, e_{k}$. For each $e_{i}$ let $n_{i}$ be the numbers of all the paths that connect the root of $T$ with a leaf, passing trough $e_{i}$. Define $\theta_{T}(v):=\theta(n_{1}, \dots, n_{k})$, and $\theta(T):=\sum_{v}\theta_{T}(v)$, where the sum is taken over the internal edges.

\begin{defi}\label{explicitpkernel}
For each $n\geq 2$, the \emph{p-kernels} are
\[
p_{n}:=\sum _{T\in \mathcal{P}_{n}}(-1)^{\theta(T)}P_{T}.
\]
\end{defi}
In \cite[Proposition 6]{Markl} it is proved that the maps $m_{\bullet}^{W}, g_{\bullet}$ defined via
\[
m_{n}^{W}:=f\circ p_{n}\circ g^{\otimes n}, \quad {{g}}_{n}:=h\circ p_{n}\circ g^{\otimes n} 
\]
give a proof of Theorem \ref{markglobal}. This was originally proved in \cite{kontsoibel} but without explicit signs.
\begin{lem}\label{generallemmadecomp}
Let $\left( V^{\bullet},d_{v}\right) $ be a cochain complex.
\begin{enumerate}
\item Assume that there is a decomposition
\begin{equation}\label{decomp}
V^{\bullet}=W\oplus d\mathcal{M} \oplus \mathcal{M},
\end{equation}
where $W=\oplus_{p\geq 0}{W}^{p}\cong H^{\bullet}\left(V\right)$ is a graded vector subspace of closed elements, and $\mathcal{M}$ is a graded vector subspace containing no exact elements except $0$. There exist maps $f,g,h$ and a diagram of the type \eqref{homretract} such that $f g=1_{W}$ and $d_{W}=0$. Moreover the maps satisfy the \emph{side conditions}: $fh=0$, $hg=0$ and $h^2=0$.
\item Let $\left( W^{\bullet},d_{W}\right) $ be a cochain complex. Assume that $d_{W}=0$ and that there is a diagram of the type \eqref{homretract} between $W$ and $V$ such that $f g=1_{W}$, $fh=0$, $hg=0$ and $h^2=0$. There is a Hodge type decomposition 
\begin{equation*}
V^{\bullet}=W\oplus d\mathcal{M} \oplus \mathcal{M}.
\end{equation*}
\item Let $\left( W^{\bullet},d_{W}\right) $ be a cochain complex. Assume that there is a diagram of the type \eqref{homretract} between $W$ and $V$ such that $f g=1_{W}$. There exists a cochain homotopy between $gf$ and $1_{V}$ such that the diagram satisfies the side conditions. 
\item There exist a cochain complexes $(W^{\bullet}, d_{w})$ and a diagram of the type \eqref{homretract} between $W$ and $V$ such that that $d_{W}=0$, $f g=1_{W}$, $f h=0$, $h g=0$ and $h^2=0$.
\end{enumerate}
\end{lem}
\begin{proof}
\begin{enumerate}
\item Let $W$ be a graded vector isomorphic to the cohomology of $V^{\bullet}$. We consider $W^{\bullet}$ as a differential graded vector space with zero differential. The inclusion $g\: : \:\left( W^{\bullet}, 0\right)\hookrightarrow \left(V, d_{V} \right)$ is a quasi-isomorphism. Let $f\: : \: W\oplus d_{V}\mathcal{M} \oplus \mathcal{M}\to W$ be the projection on the first coordinate. Let $h\: : \:  W\oplus d_{V}\mathcal{M} \oplus \mathcal{M}\to W\oplus d_{V}\mathcal{M} \oplus \mathcal{M}$ be defined by
\[
h(a_{1}, d_{V}a_{2}, a_{3}):=h(0, 0, a_{2}).
\]
It is a degree $-1$ map between graded vector spaces. A short calculation show that $h$ defines a homotopy between $f\circ g$ and $1_{W}$ and fulfills the conditions.
\item This follows by setting $\mathcal{M}:=\operatorname{Im}(h)$, $W:=\operatorname{Im}(g)$.
\item See \cite{lambestaheff}.
\item See \cite[Lemma 9.4.7]{lodayVallette}.
\end{enumerate}
\end{proof}
\begin{defi}
Given a diagram of the type \eqref{homretract} that satisfies $d_{w}=0$, $f\circ g=1_{W}$, $f\circ h=0$, $h\circ g=0$ and $h\circ h=0$. We call the Hodge type decomposition obtained by Lemma \ref{generallemmadecomp} \emph{the Hodge type decomposition associated to the diagram}.
\end{defi}
\begin{rmk}\label{rmmkvallette}
A proof of Theorem \ref{markglobal} assuming the side conditions is contained in \cite{KadeshHuebsch} (see the coalgebra perturbation lemma $2.1_{*}$).
\end{rmk}
\begin{cor}
Let $\left(A,m^{A}_{\bullet}\right) $ be a non-negatively graded $(1-)P_{\infty}$-algebra. Fix a diagram
\begin{equation}\label{diaglast}
\begin{tikzcd}
f\: : \: (A^{\bullet},m^{A}_{1})\arrow[r, shift right, ""]&(W^{\bullet}, d_{W})\: : \: g\arrow[l, shift right, ""] 
\end{tikzcd}
\end{equation}
 and a homotopy $h$. Assume that $d_{W}=0$, $f g=1_{W}$, $f h=0$, $hg=0$ and $h^2=0$. Let 
 \begin{equation}\label{decomp2}
 A=W\oplus d\mathcal{M} \oplus \mathcal{M},
 \end{equation}
 be the associated Hodge type decomposition where $d=m_{1}^{A}$. Then there exists an unique pair $(g_{\bullet},m^{W}_{\bullet})$ that satisfies Theorem \ref{markglobal} such that $g_{n}\: : \: W^{n}\to \mathcal{M}$ for $n>1$.
\end{cor}
\begin{proof}
 Let $\left(g_{\bullet}, m_{\bullet}^{W}, \right)$, $\left(\tilde{g}_{\bullet}, \tilde{m}_{\bullet} \right)$ be two distinct solutions that satisfies the above conditions. We have $\tilde{m}_{1}=0=m_{1}^{W}$ and $\tilde{g}_{1}={g}_{1}$. Assume that $\tilde{m}_{i}=m_{i}^{W}$, $g_{i}=\tilde{g}_{i}$ for $i<n$. Since they are both $C_{\infty}$-morphisms, we have
\begin{align*}
g_{1}m_{n}^{W}-dg_{n}=\tilde{g}_{1}\tilde{m}_{n}-d\tilde{g}_{n}\in A
\end{align*}
Let $w_{1}, \dots, w_{n}\in W$, then by \eqref{decomp2} we have
\begin{align*}
g_{1}m_{n}^{W}\left(w_{1}, \dots, w_{n} \right)& = \tilde{g}_{1}\tilde{m}_{n}\left(w_{1}, \dots, w_{n} \right)\\
dg_{n}\left(w_{1}, \dots, w_{n} \right)& = d\tilde{g}_{n}\left(w_{1}, \dots, w_{n} \right)
\end{align*}
We conclude $m_{n}^{W}\left(w_{1}, \dots, w_{n} \right)=\tilde{m}_{n}\left(w_{1}, \dots, w_{n} \right)$ and $g_{n}\left(w_{1}, \dots, w_{n} \right)=\tilde{g}_{n}\left(w_{1}, \dots, w_{n} \right)$.
\end{proof}
\subsection{Minimal models and homotopies}\label{Morphism of homological pairs}
Let $\left(B, m_{\bullet}^{B} \right) $ be a $P_{\infty}$-algebra and let ${g'}_{\bullet}\: : \:\left(W', m_{\bullet}^{W'}\right) \to \left( B, m_{\bullet}^{B}\right) $ be a minimal model. It has an inverse up to homotopy ${f'}_{\bullet}$. Let $p_{\bullet}\: \: \:\left(A, m_{\bullet}^{A}\right)\to\left( B, m_{\bullet}^{B}\right)  $ be a $P_{\infty}$-algebra morphism and let ${g}_{\bullet}\: : \:\left(W, m_{\bullet}^{W}\right) \to \mathcal{F}\left( A, m_{\bullet}^{A}\right) $ be a $1$-minimal model for $1-P_{\infty}$-algebras. We get the following diagram for $1-P_{\infty}$-algebras
\[
\begin{tikzcd}
\mathcal{F}\left(A, m_{\bullet}^{A}\right)\arrow[r,  "\mathcal{F}\left( p_{\bullet}\right) "]
 &  \mathcal{F}\left(B, m_{\bullet}^{B} \right)\arrow[d, , "\mathcal{F}\left( {g'}_{\bullet}\right) "]\\(W ,m_{\bullet}^{W})\arrow[u, , "g_{\bullet}"]& ( W' ,m_{\bullet}^{W'})\arrow[u, shift right, "\mathcal{F}\left({f'}_{\bullet}\right) "].
\end{tikzcd}
\]
We set $q_{\bullet}:=\mathcal{F}\left({f'}_{\bullet}\right) \mathcal{F}\left( {p}_{\bullet}\right) {g}_{\bullet}$.
\begin{lem}\label{connected maurerCartan}
$\mathcal{F}\left(g' \right)_{\bullet}q_{\bullet}$ is $1$-homotopic to $\mathcal{F}(p)_{\bullet}g_{\bullet} $. If ${g'}_{\bullet}{f'}_{\bullet}=\mathrm{Id}$, then $\mathcal{F}\left(g' \right)_{\bullet}q_{\bullet}=\mathcal{F}(p)_{\bullet}g_{\bullet} $.
\end{lem}
\begin{proof}
$\mathcal{F}\left(g' \right)_{\bullet}q_{\bullet}=\mathcal{F}\left({g'}_{\bullet}{f'}_{\bullet} {p}_{\bullet}\right)  {g}_{\bullet}$. Since ${g'}_{\bullet}{f'}_{\bullet}$ is homotopic to the identity, we have the statement.
\end{proof}

\section{The complex punctured elliptic curve and the KZB connection}\label{pt}
We construct a family of holomorphic $1$-models with logarithmic singularities for the $C_{\infty}$-algebra of smooth differential forms on punctured elliptic curves. The family is parametrized by the holomorphic structure $\tau$ in the complex upper half plane and is a holomorphic $C_{\infty}$-version of the smooth model given in \cite{LevBrown}. We apply the theory developed in \cite{Sibilia1} to build a flat connection $d-\alpha_\tau$ that corresponds to the universal KZB connection on the punctured elliptic curve (see \cite{Damien}, \cite{Hain} and \cite{Levrac}). A comparison between KZB and KZ connection is contained in \cite{Hain}, by sending $\tau$ to $i\infty$. We describe this comparison via the machinery of Subsection \ref{Morphism of homological pairs}. In the last subsection, we present another holomorphic connection gauge equivalent to the KZB connection on the punctured elliptic curve.

\subsection{A $1$-model for the torus}
Let $\tau$ be a fixed element of the upper half plane $\mathbb{H}:=\left\lbrace z\in \C \: : \: \Im(z)>0\right\rbrace $. Let $\mathbb{Z}+\tau\mathbb{Z}$ be the lattice spanned by $1,\tau$. Let $\xi$ be the coordinate on $\C$. We define a holomorphic action of $\mathbb{Z}^{2}$ on $\C$ by
\[
(m,n)(\xi):=\xi+m+\tau n,
\]
which is free and properly discontinuous. The lattice $\mathbb{Z}+\tau\mathbb{Z}$ is a normal crossing divisor of $\C$ and it is preserved by the action of $\Z^{2}$. The action groupoid $\C_{\bullet}\Z^{2}$ is a simplicial manifold equipped with a simplicial normal crossing divisor $\left( \mathbb{Z}+\tau\mathbb{Z}\right)_{\bullet}\Z^{2}$ (see \cite{Sibilia1}). By the de Rham Theorem for simplicial manifolds we have
\[
 H^{\bullet}\left(  \operatorname{Tot}_{N}\left(A_{DR}\left( \left( \C-\left\lbrace \mathbb{Z}+\tau\mathbb{Z}\right\rbrace\right)_{\bullet}\Z^{2} \right) \right),D\right) \cong H^{\bullet}\left({\mathcal{E}}^{\times},\C \right).
\]
Let $\gamma\: : \: \Z^{2}\to \C$ be the group homomorphism defined by $\gamma(m,n):=n2\pi i$. Then $d\xi$ and $\gamma$ are closed forms in $\operatorname{Tot}_{N}\left(A_{DR}\left( \left( \C-\left\lbrace \mathbb{Z}+\tau\mathbb{Z}\right\rbrace\right)_{\bullet}\Z^{2} \right)\right)$ of type $(0,1)$ and $(1,0)$  respectively. They generate the cohomology. We construct a $1$-model for $\operatorname{Tot}_{N}\left(A_{DR}\left( \left( \C-\left\lbrace \mathbb{Z}+\tau\mathbb{Z}\right\rbrace\right)_{\bullet}\Z^{2} \right)\right)$.\\
We fix a family of holomorphic functions $f^{(i)}\: : \: \C-\left\lbrace \mathbb{Z}+\tau\mathbb{Z}\right\rbrace \to \C$ indexed by $i\in \N$ such that
\begin{equation}\label{conditions}
f^{(0)}=1,\quad f^{(n)}(\xi+l)=f^{(n)}(\xi),\quad f^{(n)}(\xi+l\tau)= \sum_{j=0}^{n}\frac{f^{(n-j)}(\xi)(-2\pi il)^{j}}{j!}
\end{equation}
for all $l,n\in \Z$. 
\begin{lem}\label{linindfunct}
Let $f^{(i)}\: : \: \C-\left\lbrace \mathbb{Z}+\tau\mathbb{Z}\right\rbrace \to \C$ be a family of holomorphic functions indexed by $\N$ that satisfy $\eqref{conditions}$. Then they are linearly independent.
\end{lem}
\begin{proof}
Assume that there is a non trivial relation $\sum_{i\in I}\lambda_{i}f^{(i)}=0$. Let $p$ be the maximal integer of $I$ such that $\lambda_{p}\neq 0$. Notice that $f^{(i)}\: : \: \C-\left\lbrace \mathbb{Z}+\tau\mathbb{Z}\right\rbrace \to \C$ are elements of $\operatorname{Tot}_{N}\left(A_{DR}\left( \C-\left\lbrace \mathbb{Z}+\tau\mathbb{Z}\right)_{\bullet}\Z^{2}\right) \right) $ of bidegree $(0,0)$. Let $V$ be the complex vector space generated by $\left\lbrace f^{(i)}\right\rbrace_{i\in \N}$. Then each $f^{(i)}$ defines a map $\partial_{G}f^{(i)}\: : \: \Z^{2}\to V$ via $$((m,n))\mapsto \partial_{G}f{(m,n)}=\sum_{j=0}^{i}\frac{f^{(i-j)}(\xi)(-2\pi in)^{j}}{j!}.$$
Set $x:=n$, then $\sum_{i\in I}\lambda_{i}\partial_{G}f^{(i)}$ can be seen as a polynomial $p(x)$ in variable $x$ and coefficients in $V$ of the form
\[
p(x)=\sum h^{a}x^{a}
\]
where $h^{a}\in V$. In particular $h^{p}=\lambda_{p}f^{\left( 0\right)}(-2\pi i)^{p}=\lambda_{p}(-2\pi i)^{p}$.
Since $\sum_{i\in I}\lambda_{i}\partial_{G}f^{(i)}=0$ we get that $p(x)=0$ for any $x\in \N$. This implies that $h^{p}=0$ and then $\lambda_{p}=0$.
\end{proof}
For $i\in \N$, we set $\phi^{(i)}:=f^{(i)}d\xi$. Then $$\phi^{(i)}\in \operatorname{Tot}_{N}\left(A_{DR}\left( \left( \C-\left\lbrace \mathbb{Z}+\tau\mathbb{Z}\right\rbrace\right)_{\bullet}\Z^{2} \right)\right)$$ are linearly indipendent elements of bidegree $(0,1)$. Note that $d\phi^{(i)}=0$ since they are holomorphic $1$-forms.
We consider $\operatorname{Tot}_{N}\left(A_{DR}\left( \left( \C-\left\lbrace \mathbb{Z}+\tau\mathbb{Z}\right\rbrace\right)_{\bullet}\Z^{2} \right)\right)$ equipped with the unital $C_{\infty}$-structure ${m}_{\bullet}$ such that $m_{1}=D$ (see \cite[Section 2]{Sibilia1}). 
\begin{lem}\label{modelelliptic curve}
\begin{enumerate}
\item $D\left( -\phi^{(n)}\right) =\sum_{l=1}^{n}{m}_{l+1}(\gamma,\dots, \gamma, \phi^{(p-l)})$, for any $n$,
\item ${m}_{l+1}(\gamma^{\otimes i}, \phi^{(k)}, \gamma^{\otimes l+1-i})=(-1)^{l}\binom{l}{i}{m}_{l+1}(\phi^{(k)},\gamma,\dots, \gamma)$,
\item Let $B^{\bullet}\subset \operatorname{Tot}_{N}\left(A_{DR}\left( \left( \C-\left\lbrace \mathbb{Z}+\tau\mathbb{Z}\right\rbrace\right)_{\bullet}\Z^{2} \right)\right)$ be $C_{\infty}$-sub algebra generated by 
$
1,\gamma,\left\lbrace \phi^{(i)}\right\rbrace_{i\in \N}.
$
Then $$\left( {m}_{\bullet},B^{\bullet}\right) \subset \left({m}_{\bullet}, \operatorname{Tot}_{N}\left(A_{DR}\left( \left( \C-\left\lbrace \mathbb{Z}+\tau\mathbb{Z}\right\rbrace\right)_{\bullet}\Z^{2} \right)\right)\right)$$ is a $1$-model. 
\end{enumerate}
\end{lem}
\begin{proof}
We use Theorem 2.7 in \cite{Sibilia1}. Fix a $(m,n)\in \Z^{2}.$
Let $l>1$. Then \begin{equation}
\label{firstexpr}{m}_{l+1}\left( \gamma, \dots,  \gamma, \phi^{(p-l)} \right){(m,n)}= \frac{B_{l}}{l!}\sum _{i\geq 1}^{p-l}\frac{\phi^{(p-l-i)}\left(-2 \pi i n \right)^{i+l} }{i!}.
\end{equation}
If $l=1$, \begin{equation}
\label{secondexpr}{m}_{2}\left( \gamma,  \gamma, \phi^{(p-1)} \right){(m,n)}= -\left(\left( 1- B_{1}\right) \sum _{i\geq 1}^{p-l}\frac{\phi^{(p-1-i)}\left(-2 \pi i n \right)^{i+1} }{i!} + \left(-2 \pi i n \right)\phi^{(p-1)} \right).
\end{equation}
More easily we have
 \begin{equation}\label{totalon}
D\left(- \phi^{(p)}\right){(m,n)}=-\sum _{i\geq 1}^{p}\frac{\phi^{(p-i)}\left(-2 \pi i n \right)^{i} }{i!}
\end{equation}
We prove that $\eqref{firstexpr}+\eqref{secondexpr}=\eqref{totalon}$ by comparing the coefficients. For $p=1,2$ the two expressions agree. For $n>2$ $\eqref{firstexpr}+\eqref{secondexpr}=\eqref{totalon}$ is equivalent to
\[
\frac{1-B_{1}}{(i-1)!}-\sum_{j=2}^{i-1}\frac{B_{j}}{j!}\frac{1}{(i-j)!}=\frac{1}{i!}.
\]
The first Bernoulli numbers ${B'}_{j}$ satisfy $\sum_{j=0}^{\infty}\frac{{B'}_{j}}{j!}=\frac{t}{e^t-1}$. Since $B_{j}={B'}_{j}$ for $j\neq 1$ and $B_{1}=\frac{1}{2}=-{B'}_{1}$, the condition above is equivalent to
\[
\left( 1-\frac{t}{1-e^t}\right) \left(e^{t}-1 \right)=e^{t}-1-t.
\]
This proves the first statement. 
The second statement follows by \cite[Theorem 2.7]{Sibilia1}. We prove point 3. Let $i\: : \: B\to \operatorname{Tot}_{N}\left(A_{DR}\left( \left( \C-\left\lbrace \mathbb{Z}+\tau\mathbb{Z}\right\rbrace\right)_{\bullet}\Z^{2} \right)\right)$ be the inclusion. We have to show that the induced map $[i]$ in cohomology is an isomorphism between elements of degree $1$ and an is injective for the elements of degree $2$. Assume that there is an element $\sum_{i\in I}\lambda_{i}\phi^{(i)}$ such that $\sum_{i\in I}\lambda_{i}D\phi^{(i)}=0$ and $0\notin I$. Let $p$ be the maximal integer of $I$ such that $\lambda_{p}\neq 0$. Let $V$ be the complex vector space generated by $\left\lbrace \phi^{(i)}\right\rbrace_{i\in \N}$. Then each $\phi^{(i)}$ defines a map $D\phi^{(i)}\: : \: \Z^{2}\to V$ via $$((m,n))\mapsto \left( Df\right) {(m,n)}=\sum_{j=0}^{i}\frac{\phi^{(i-j)}(\xi)(-2\pi in)^{j}}{j!}.$$
Set $x=n$, hence $\sum_{i\in I}\lambda_{i}D\phi^{(i)}$ can be written as a polynomial $p(x)$ with coefficients in $V$ of the forms
\[
p(x)=\sum h^{a}x^{a}
\]
where $h^{a}\in V$. In particular the monomial that multiplies $\phi^{\left( p-1\right)}$ is $x$, i.e we get
\[
h^{1}=\lambda_{p}\phi^{\left( p-1\right)}+h
\]
where $h$ lies in the subvector space of $V$ generated by $\phi^{(0)}, \dots, \phi^{(p-2)}$. Since $\sum_{i\in I}\lambda_{i}D\phi^{(i)}=0$ we get $p(x)=0$ for any $x\in \N$. Lemma \ref{linindfunct} implies that $h^{1}=0$ and hence $\lambda_{p}=0$. This shows that the vector space of closed forms in $B^{1}$ is generated by $\phi^{0}$ and $\gamma$. We show that $H^{2}(B)$ vanishes.
Suppose that there is a closed element
\[
f:=\sum_{i,j}^{n}\lambda_{i,j}{m}_{i}\left(\gamma, \dots, \gamma, \phi^{(j)} \right),  
\] 
 for $i>1$ and $j\geq 0$. By point 1 we can assume that $\lambda_{2,j}=0$. Since ${m}_{i}\left(\gamma, \dots, \gamma, \phi^{(0)} \right)=0$ for $i>2$ we conclude that $j>0$. Let $p$ be the maximal integer such that $\lambda_{i,p}\neq 0$ for some $i$ and let $l$ be the maximal integer such that $\lambda_{l,p}\neq 0$. The element $Dm_{l}\left(\gamma, \dots, \gamma, \phi^{p} \right) $ is of type $(2,1)$, then for $((m,n),(m',n'))\in \Z^{2}\times \Z^{2}$ we have
\begin{align*}
\left( Dm_{l}\left(\gamma, \dots, \gamma, \phi^{p} \right)\right) {((m,n),(m',n'))} &=\frac{B_{l-1}}{\left(l-1 \right) !} \sum _{i\geq 1}^{p}\frac{(m,n)\cdot\phi^{(p-i)}\left(-2 \pi i n' \right)^{i+l-1} }{i!}\\
&+ \frac{B_{l-1}}{\left(l-1 \right) !}\sum _{i\geq 1}^{p}\frac{\phi^{(p-i)}\left(-2 \pi i (n+n') \right)^{i+l-1} }{i!}\\
&+\frac{B_{l-1}}{\left(l-1 \right) !}\sum _{i\geq 1}^{p}\frac{\phi^{(p-i)}\left(-2 \pi i n \right)^{i+l-1} }{i!} 
\end{align*}
Then $Df$ defines a map $Df\: : \: \Z^{2}\times \Z^{2}\to V$ via $$((m,n),(m',n'))\mapsto \left( Df\right) {((m,n)(m',n'))}.$$ Now set $x=n$, $y=n'$, then $Df$ can be written as a polynomial $p(x,y)$ with coefficients in $V$ of form
\[
p(x,y)=\sum h^{a,b}x^{a}y^{b}.
\]
The monomials that multiply $\phi^{(p-1)}$ are $$\frac{\lambda_{l,p}B_{l-1}(-2\pi i)^{l}}{\left(l-1 \right) !}\left(x^{l}-(x+y)^{l}+y^{l}\right)+r(x,y) $$ where $r(x,y)$ is a polynomial containing monomials of total degree smaller than $l$. We conclude that 
\[
h^{l-1,1}=\frac{\lambda_{l,p}B_{l-1}(-2\pi i)^{l}}{\left(l-1 \right) !}\phi^{(p-1)}+h, \quad h^{1,l-1}=\frac{\lambda_{l,p}B_{l-1}(-2\pi i)^{l}}{\left(l-1 \right) !}\phi^{(p-1)}+h'
\]
where $h,h'\in W'$, where $W'$ is the vector space generated by $\phi^{(i)}$ for $1\leq i\leq p-2 $ . Since $p(x,y)$ vanishes on $\Z^{2}\times \Z^{2}$ we conclude that $h^{a,b}=0$ for any $a,b$. We get relations
\[
\frac{\lambda_{l,p}B_{l-1}(-2\pi i)^{l}}{\left(l-1 \right) !}\phi^{(p-1)}+h=0, \quad \frac{\lambda_{l,p}B_{l-1}(-2\pi i)^{l}}{\left(l-1 \right) !}\phi^{(p-1)}+h=0
\]
which implies $\lambda_{l,p}=0$. Hence the only closed form in $B^{2}$ is the exact form
\[
\lambda {m}_{2}\left(\gamma,  \phi^{(0)} \right)
\]
 and $H^{2}(B)=0$.
\end{proof}
\begin{cor}\label{decoB}
Let $\left( B, {m}_{\bullet}\right) $ be as above. There exists a Hodge decomposition
\begin{equation}\label{elliptic curvedec}
B=W^{\bullet}\oplus D\mathcal{M}\oplus \mathcal{M}
\end{equation}
such that
\begin{enumerate}
	\item  $W^{1}$ is generated by $\phi^{(0)}$ and $\gamma$, $\mathcal{M}^{1}$ is generated by $\phi^{(i)}$ for $i>0$, and $\left( D\mathcal{M}\right)^{0}=0$,
	\item  $W^{2}=0$, $\mathcal{M}^{2}$ is generated by ${m}_{l}\left(\phi^{(i)}, \gamma, \dots, \gamma\right) $ for $i>0, l>2$, and $\left( D\mathcal{M}\right)^{2}$ is generated by $D\left( \phi^{(i)}\right) $ for $i>1$.
\end{enumerate}
\end{cor}
\begin{proof}
The desired decomposition in higher degree follows by the proof of \cite[Lemma 9.4.7]{lodayVallette}.
\end{proof}
Notice that the $1$-model $B$ is completely determined by the choiche of holomorphic functions $f^{(i)}\: : \: \C-\left\lbrace \mathbb{Z}+\tau\mathbb{Z}\right\rbrace \to \C$ that satisfy $\eqref{conditions}$.
The above facts are true for any $\tau\in \mathbb{H}$ as well. 
\begin{cor}\label{cormodel}
For any $\tau\in \mathbb{H}$ there exists a holomorphic $1$-model $B\subset \operatorname{Tot}_{N}\left(A_{DR}\left( \left( \C-\left\lbrace \mathbb{Z}+\tau\mathbb{Z}\right\rbrace\right)_{\bullet}\Z^{2} \right)\right)$ with logarithmic singularities and equipped with a Hodge type decomposition.
\end{cor}
We fix a $\tau$ and a family of holomorphic functions $f^{(i)}\: : \: \C-\left\lbrace \mathbb{Z}+\tau\mathbb{Z}\right\rbrace \to \C$ that satisfy $\eqref{conditions}$. We consider $\left( B, {m}_{\bullet}\right) $ equipped with the Hodge type decomposition above. We calculate the degree zero geometric connection of the action groupoid $\left( \mathbb{C}-\left\lbrace \mathbb{Z}+\tau\mathbb{Z}\right\rbrace \right) _{\bullet}\Z^{2}$ associated to the decomposition \eqref{elliptic curvedec}. We apply Lemma \ref{generallemmadecomp} to the decomposition \eqref{elliptic curvedec} and get chain maps $f,g$ and a chain homotopy $h$. Let $p_{n}=\sum _{T\in \mathcal{P}_{n}}(-1)^{\theta(T)}P_{T}$ be the p-kernel of Definition \ref{explicitpkernel}. There is a $C_{\infty}$-structure $m_{\bullet}^{W}$ on $W$ given by
\[
{m}_{n}:=f\circ p_{n}\circ g^{\otimes n},\quad n\geq 1
\]
and a $C_{\infty}$-quasi-isomorphism ${{g}}_{n}\: : \: W^{\otimes n}\to W$ is given by 
\[
{{g}}_{n}:=h\circ p_{n}\circ g^{\otimes n},\quad n\geq 1.
\]
Since $H^{2}(B)=0$, the $C_{\infty}$-structure $m_{n}^{W}$ vanishes on $W^{1}$. 
\begin{lem}\label{finalcomp}
Let $g_{\bullet}$ be the above quasi-isomorphism. Let $v_{i_{1}},\dots,v_{i_{n}}\in\left\lbrace \gamma, \phi^{(0)} \right\rbrace  $ for $l=1,\dots,n$.
\begin{enumerate}
\item  $g_{n}\left(v_{{1}},\dots ,v_{{n}} \right)=0$ if $\left( v_{{1}},\dots ,v_{{n}}\right) $ is not of the following form: there exists exactly one $j$ such that $v_{{j}}=\phi^{(0)}$ and $v_{s}=\gamma$, for $s\neq j$.
\item $g_{n}\left(  \phi^{(0)},\gamma,\dots ,\gamma \right)=\phi^{(n)}$.
\end{enumerate}
\end{lem}
\begin{proof}
For each $n\geq 2$, the p-kernels are
\[
p_{n}:=\sum _{T\in \mathcal{P}_{n}}(-1)^{\theta(T)}P_{T}.
\]
Since ${m}_{\bullet}$ is $C_{\infty}$ we have ${m}_{l}(\gamma, \dots, \gamma)=0$ for $l>1$. Consider ${m}_{l}(w_{i_{1}}, \dots, w_{i_{l}})$ such that $w_{i_{l}}\in \left\lbrace \gamma, \phi^{(0)}, \phi^{(1)}, \dots \right\rbrace $. Assume that there exists exactly  $j_{1}, \dots j_{r}$ where $l>r>1$ such that $w_{j_{i}}=\phi^{(k)}$ for $i=1, \dots, r$ for some $k\geq 0$. By \cite[Theorem 2.7]{Sibilia1} we have ${m}_{l}(w_{i_{1}}, \dots, w_{i_{l}})=0$.  The same conclusion is true if $r=l>2$. For $l=2$ we get ${m}_{2}(w_{i_{1}}, w_{i_{2}})=w_{i_{1}}\wedge w_{i_{2}}$ which vanishes for dimensional reasons.
Let $T\in \mathcal{P}_{n}$ be an oriented planar tree and consider the induced map $P_{T}$. Lemma \ref{modelelliptic curve} and the decomposition \eqref{elliptic curvedec} imply
\begin{equation}\label{homotopyexpr}
h\left({m}_{2}(\gamma, \phi^{(k-1)}) \right)=h\left({m}_{1}\left( \phi^{(k)}\right) -\sum_{l=2}^{k}(-1)^{l+1}{m}_{l+1}(\gamma,\dots, \gamma, \phi^{(k-l)})\right)= \phi^{(k)}
\end{equation}
and
\[
h\left({m}_{l}(\gamma,\dots,\gamma, \phi^{(k)}, \gamma,\dots ,\gamma) \right)=0\quad l>2
\]
for any $k\geq 1$. This proves point 1. The only binary tree $T$ such that $h\circ P_{T}\circ g^{\otimes n}\left( \phi^{(0)},\gamma,\dots ,\gamma\right)\neq 0 $ is the one in Figure \ref{fig2}.
\begin{figure}[ht]
\centering
  \includegraphics[scale=.5]{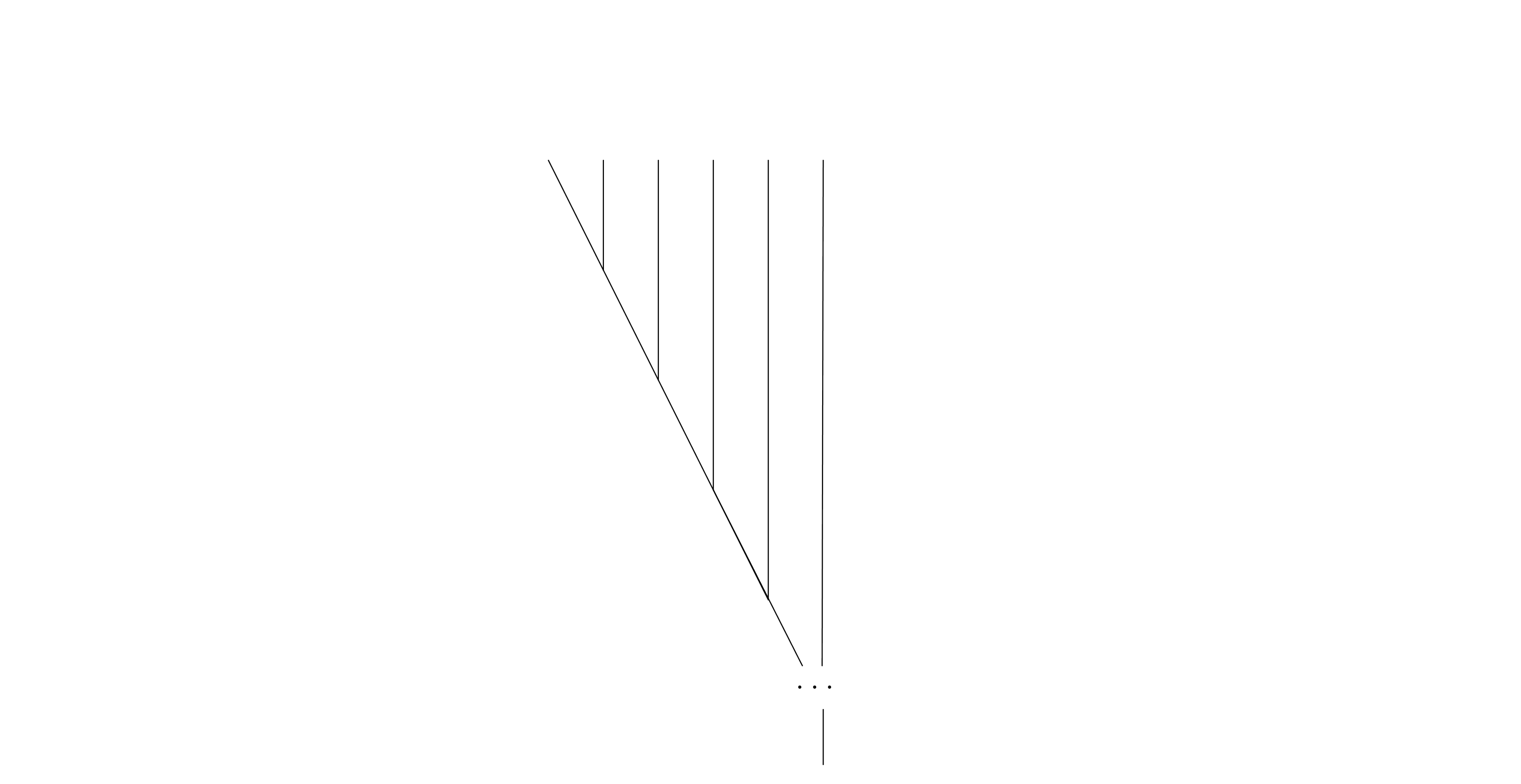}
  \caption{ Three with non-trivial induced map }
  \label{fig2}
\end{figure}
Here we have $m_{k}=0$ for $k>2$ even. A direct calculation shows $h\circ p_{T}\circ g^{\otimes n}\left( \phi^{(0)},\gamma,\dots ,\gamma \right)=\phi^{(n)}$.
\end{proof}

By \cite[Proposition 1.20]{Sibilia1}, $\mathcal{F}\left( g\right) _{\bullet}$ corresponds to a Maurer-Cartan element in $$\operatorname{Conv}_{1-C_{\infty}}\left(\mathcal{F}\left(W, m_{\bullet}^{W}\right) ,\mathcal{F}\left(B, m_{\bullet} \right) \right).$$  We consider $W$ equipped with the basis $-\gamma$, $-\phi^{(0)}$. We consider $\left( W^{1}[1]\right)^{*} $ as the vector space generated by the degree zero elements $X_{0},X_{1}$, where $X_{1}$, and resp. $X_{0}$ denote the dual of $-s\phi^{(0)}$ and of $-s\gamma$ in $W[1]$ respectively. Since $m_{\bullet}^{W}$ vanishes on elements of degree $1$, by \cite[Proposition 1.25]{Sibilia1} we get that $\operatorname{Conv}_{1-C_{\infty}}\left(\left(W, m_{\bullet}^{W}\right) ,\mathcal{F}\left(B, m_{\bullet} \right) \right)\cong B\widehat{\otimes} \widehat{\mathbb{L}}\left( X_{0}, X_{1} \right)$. Hence $\mathcal{F}\left( g\right) _{\bullet}$ can be written as a formal power series
\[
C\in B^{1}\widehat{\otimes} \widehat{\mathbb{L}}\left( X_{0}, X_{1} \right).
\]
\begin{prop}\label{deg0geom}
Consider a family of holomorphic functions $f^{(i)}\: : \: \C-\left\lbrace \mathbb{Z}+\tau\mathbb{Z}\right\rbrace \to \C$ indexed by $i\in \N$ that satisfy \eqref{conditions}. Let $\left(B, {m}_{\bullet} \right) $ be the $1$-model constructed in Lemma \ref{modelelliptic curve} and equipped with the decomposition $\left(W,\mathcal{M} \right) $ as in \eqref{elliptic curvedec}. We consider $W$ equipped with the basis $-\gamma$, $\phi^{(0)}$. 
\begin{enumerate}
	\item There exists a unique minimal model of $C_{\infty}$-algebras $g_{\bullet}\: : \: \left(W, m_{\bullet} \right) \to \left( B, m_{\bullet}\right) $ such that $g_{n}\: : \: W^{\otimes n}\to \mathcal{M}\subset B$ for $n>1$.
	\item The degree zero geometric connection corresponding to $\mathcal{F}(g_{\bullet})$ is given by 
	\[
	C =-\gamma X_{0}-\sum_{p\geq 0}\phi^{(p)}\operatorname{Ad}^{p}_{X_{0}}(X_{1}).
	\]
	\end{enumerate}
\end{prop}
\begin{proof}
Let $L\subset \widehat{\mathbb{L}}\left( X_{0}, X_{1} \right)$ be the complete vector subspace spanned by $X_{0}$ and by all the Lie monomials $P$ such that its associated monomial (see \cite{Lothaire}, chapter 5) is
\begin{equation}\label{conditionmonomial}
\operatorname{Mon}(P)=X_{1}^{\otimes s}\otimes X_{0}\otimes X_{1}^{\otimes r}
\end{equation}for some $s,r\geq 0$. In particular $X_{0}$ and the set of Lie monomials of the form $\left[\left[X_{1},X_{0} \right],\dots,X_{0} \right] $ form a basis of $L$. Let $\operatorname{Ad}^{0}_{X_{0}}(X_{1}):=X_{1}$ and $\operatorname{Ad}^{p}_{X_{0}}(X_{1}):=\left[X_{0}, \operatorname{Ad}^{p-1}_{X_{0}}(X_{1})\right] $ for $p>0$. By Lemma \ref{modelelliptic curve} and Lemma \ref{finalcomp} we have
\begin{align*}
C& =-\gamma X_{0}-\sum_{p\geq 0}(-1)^{p}\phi^{(p)}\left[\left[X_{1},X_{0} \right],\dots,X_{0} \right]\\
& =-\gamma X_{0}-\sum_{p\geq 0}\phi^{(p)}\operatorname{Ad}^{p}_{X_{0}}(X_{1}).
\end{align*}
\end{proof}
\subsection{Flat connections on the punctured torus}
We denote by $\xi=s+\tau r$ the coordinates on $\C$. We define the smooth differential forms $\omega^{(i)}\in A^{1}_{DR}\left(\mathcal{E}_{\tau}^{\times}\right) $ via
\[
\sum_{k\geq 0}\omega^{(k)}\alpha^{k}:= \left(\exp\left(2 \pi i r \alpha \right) \right) \left( \sum_{k\geq 0}\phi^{(k)}\alpha^{k}\right) 
\]
where $\alpha$ is a formal variable. Let $\nu:=2 \pi i dr$. It is easy see that the forms above satisfiy the following relations:
\[
D\omega^{(k)}=d \omega^{(k)}=\nu\wedge \omega^{(k-1)}, \quad  \omega^{(0)}=d \xi
\]
Let $\bar{B}\subset \left( A_{DR}\left(\mathcal{E}_{\tau}^{\times}\right), D, \wedge\right)  $ be the differential graded algebra generated by $1$, $\nu$ and the $\omega^{(i)}$ for $i\geq 0$. Notice that $\bar{B}$ depends by the choice of functions $f^{(i)}\: : \: \C-\left\lbrace \mathbb{Z}+\tau\mathbb{Z}\right\rbrace \to \C$ indexed by $i\in \N$ that satisfy \eqref{conditions}.
\begin{prop}\label{Brownmodel}
	Let $\bar{B}\subset  A_{DR}\left(\mathcal{E}_{\tau}^{\times}\right)$ as above.
	\begin{enumerate}
		\item $\bar{B}$ is a model for $  A_{DR}\left(\mathcal{E}_{\tau}^{\times}\right) $.
		\item There exists a Hodge type decomposition
		\begin{equation}\label{Browndec}
		\bar{B}=\bar{W}\oplus D\bar{\mathcal{M}}\oplus \bar{\mathcal{M}}
		\end{equation}
		such that
		\begin{enumerate}
			\item  $\bar{W}^{1}$ is generated by $\omega^{(0)}$ and $\nu$, $\bar{\mathcal{M}}^{1}$ is generated by $\omega^{(i)}$ for $i>0$, and $\left( D\bar{\mathcal{M}}\right)^{0}=0$,
			\item  $\bar{W}^{2}=0$, $\bar{\mathcal{M}}^{1}$ is generated by $\omega^{(i)}\nu $ for $i>0,$, and $\left( D\bar{\mathcal{M}}\right)^{1}$ is generated by $D\left( \omega^{(i)}\right) $ for $i>1$.
		\end{enumerate}
	\end{enumerate}
\end{prop}
\begin{proof}
	In \cite{LevBrown}, it is proved point 1 and the fact that $\omega^{(0)}$, $\omega^{(1)}$, \dots, $\nu$ are linearly independent and the fact that $\omega^{(0)}\nu $, $\omega^{(1)}\nu $, \dots are linearly independent. This proves the second point.
\end{proof}
We consider $\left( \bar{W^{1}}[1]\right)^{*} $ as the vector space generated by the degree zero elements $\bar{X}_{0},\bar{X}_{1}$, where $\bar{X}_{1}$, and resp. $\bar{X}_{0}$ denote the dual of $-s\omega^{(0)}$ and resp. of $s\nu$.
By using the same calculation of the previous section, we have the following.
\begin{prop}\label{Brown}
	Consider a family of holomorphic functions $f^{(i)}\: : \: \C-\left\lbrace \mathbb{Z}+\tau\mathbb{Z}\right\rbrace \to \C$ indexed by $i\in \N$ that satisfy \eqref{conditions} and let $\bar{B}$ equipped with the decomposition \eqref{Browndec}. The degree zero geometric connection is given by
	\[
	\bar{C} =\nu \bar{X}_{0}-\sum_{p\geq 0}\omega^{(p)}\operatorname{Ad}^{p}_{\bar{X}_{0}}\left( \bar{X}_{1}\right) 
	\]
	In particular, $d-\bar{C}_{0}$ is a smooth flat connection on the trivial bundle $\mathcal{E}_{\tau}^{\times}\times \widehat{\mathbb{L}}\left( \bar{X}_{0}, \bar{X}_{1} \right)$.
\end{prop}
We define the factor of automorphy $F\: : \: \Z^2\times \C-\left\lbrace \mathbb{Z}+\tau\mathbb{Z}\right\rbrace\to $ via ${F}_{(n,m)}(\xi)=\exp(-2 \pi im X_{0})$. We define the bundle $E_{F}=\left( \C-\left\lbrace \mathbb{Z}+\tau\mathbb{Z}\right\rbrace\right)\times \widehat{\mathbb{L}}\left( X_{0}, X_{1} \right)/ \Z^2$ on the punctured torus where the $\Z^{2}$-action is induced by $F$, i.e.
\[
(n,m)(\xi,v)=\left( \xi+n+m\tau, \exp(-2 \pi im X_{0})\cdot v\right).
\]
where $X_{0}^{p}\cdot v:=\operatorname{Ad}^{p}_{X_{0}}(v)$ for $v\in \widehat{\mathbb{L}}\left( X_{0}, X_{1} \right)$. The sections of $E_{F}$ satisfy: 
\begin{equation}\label{bundlep}
s(\xi + l)=s(\xi), \quad s(\xi + l\tau)=\exp(-2 \pi il X_{0})\cdot s(\xi)
\end{equation}
 The conditions \eqref{conditions} implies that
\begin{eqnarray}
\sum_{p\geq 0}\phi^{(p)}(\xi+l)\operatorname{Ad}^{p}_{X_{0}}(X_{1})& =& \sum_{p\geq 0}\phi^{(p)}(\xi)\operatorname{Ad}^{p}_{X_{0}}(X_{1}),\\
\sum_{p\geq 0}\phi^{(p)}(\xi+l\tau)\operatorname{Ad}^{p}_{X_{0}}(X_{1})&=&\exp(-2 \pi i X_{0})\cdot \sum_{p\geq 0}\phi^{(p)}(\xi)\operatorname{Ad}^{p}_{X_{0}}(X_{1})
\end{eqnarray}
for $l\in \Z$.

\begin{thm}\label{main1}
Consider a family of holomorphic functions $f^{(i)}\: : \: \C-\left\lbrace \mathbb{Z}+\tau\mathbb{Z}\right\rbrace \to \C$ indexed by $i\in \N$ that satisfy \eqref{conditions}. Let $C$ be as in Proposition \ref{deg0geom} and let $\bar{C}$, $F$ and $E_{F}$ be as above. Let $$\left( {m}_{\bullet},{B'}^{\bullet}\right) \subset \left({m}_{\bullet}, \operatorname{Tot}_{N}\left(A_{DR}\left( \left( \C-\left\lbrace \mathbb{Z}+\tau\mathbb{Z}\right\rbrace\right)_{\bullet}\Z^{2} \right)\right)\right)$$
		be a $1$-model equipped with a Hodge space decomposition. Let ${g'}_{\bullet}\: : \: (W', m_{\bullet}^{W'})\to \left( {m}_{\bullet},{B'}^{\bullet}\right)$ be the minimal model obtained via the homotopy transfer theorem and let $C'$ be the degree zero geometric connection corresponding to $\mathcal{F}\left( {{g'}_{\bullet}}\right) $.
	\begin{enumerate}
		\item $d-r^{*}C$ is a well-defined holomorphic connection on the holomorphic bundle $E_F$,
		\item $\left( d-r^{*}C, E_{F}\right) $ is smoothly isomorphic to $(d-r^{*}\overline{C},\mathcal{E}_{\tau}^{\times}\times \widehat{\mathbb{L}}\left( \bar{X}_{0}, \bar{X}_{1} \right) )$,
		\item  There exists a  smooth factor of automorphy $F'$ and a bundle isomorphism $T\: : \: E_{F}\to E_{F'}$ such that $\left( d-r^{*}C, E_{F}\right) $ and $\left( d-r^{*}C', E_{F'}\right) $ are isomorphic.
		\item If $B'$ is holomorphic, then $F'$ and $T'$ are holomorphic.
		\item If $B'$ is holomorphic with logarithmic singularities, then $F'$ is holomorphic on the torus and $T'$ is an holomorphic isomorphism between bundles on the torus.
			\end{enumerate}
		\end{thm}
\begin{proof}
We define a Lie algebra isomorphism $K^{*}$ by $\bar{X}_{i}\mapsto X_{i}$ for $i=1,0$. There is a gauge equivalence between 
\[
k^{*}\overline{C}=\nu {X}_{0}-\sum_{p\geq 0}\omega^{(p)}\operatorname{Ad}^{p}_{{X}_{0}}\left( {X}_{1}\right)
\]
and $r_{*}{C}_{0}$ given by $h(\xi)=2\pi rX_{0}$. The factor of automorphy induced by $h$ (see \cite[Theorem 4.9]{Sibilia1}) is precisely ${F}_{(n,m)}(\xi)=e^{-h((n,m)\xi)}e^{h(\xi)}=\exp(-2 \pi im X_{0})$. All the other statements follows from \cite[Theorem 4.9, Theorem 4.10]{Sibilia1}. 
\end{proof}
\begin{rmk}
	Consider a family of holomorphic functions $f^{(i)}\: : \: \C-\left\lbrace \mathbb{Z}+\tau\mathbb{Z}\right\rbrace \to \C$ indexed by $i\in \N$ that satisfy \eqref{conditions}.
	Note that all the statements above are true even by assuming that $f^{0}(z)$ is a $\Z^{2}$-invariant holomorphic function (an elliptic function with pole in $0$). In this case we can construct a new model in the same spirit as above. We construct $\bar{\phi}^{(i)}$ by 
	\[
	\bar{\phi}^{(i)}:= {m}_{2}({\phi}^{(i)},f^{0}(z)),\quad i\geq 0
	\]
	This gives a family of $1$-forms $\bar{\phi}^{(i)}$ indexed by $i\in \N$ that satisfies the relation \eqref{totalon}. Hence $\left( \bar{\phi}^{(i)},\gamma\right)$ generates a new holomorphic $1$-model.
\end{rmk}
\subsection{Kronecker function and KZ connection}\label{ComparisonHain0}
We show that the KZB connection can be constructed as degree zero geometric connection. 
Let $\theta\left(\xi,\tau\right)$ be the ``two thirds of the Jacobi triple formula'':
\begin{equation}\label{thetadueterz}
\theta\left(\xi,\tau\right)=q^{1/12}\left(z^{1/2}-z^{-1/2}\right)\prod_{j=1}^\infty\left(1-q^{j}z\right)\prod_{j=1}^\infty\left(1-q^{j}z^{-1}\right)
\end{equation}
where $\xi\in \mathbb{C}$ and $\tau\in \mathbb{H}$, $z:=\exp(2\pi i \xi)$, $q:=\exp(2 \pi i \tau)$.
\begin{prop}Set $\theta\left(0,\tau\right)':=\frac{\partial}{\partial \xi}\theta\left(0,\tau\right)$. The Kronecker function\footnote{There is a conflict of notation with respect to \cite{Don}. Our function $F$ is the one used in \cite{LevBrown}, \cite{Levrac},\cite{Damien} and \cite{Hain}. Let $F^{Z}$ be the function in \cite{Don}, then
		$
		F(\xi,\eta,\tau)=2 \pi iF^{Z}(2 \pi i\xi,2 \pi i\eta,\tau).
		$} is defined as
	$$F\left(\xi,\eta,\tau \right):=\frac{\theta\left(0,\tau\right)'\theta\left(\xi+\eta,\tau\right)}{\theta\left(\xi,\tau\right)\theta\left(\eta,\tau\right)}.$$
	
\begin{enumerate}
\item[i)] $F$ is a meromorphic function with simple pole at $(\xi, \eta, \tau)$ where $\xi\in \mathbb{Z}+\tau\mathbb{Z}$ and $\eta\in\mathbb{Z}+\tau\mathbb{Z}$,
\item[ii)] It satisfies the quasi-periodicity
\begin{equation}\label{qper}
F\left(\xi+1,\eta,\tau \right)=F\left(\xi,\eta,\tau \right),\quad F\left(\xi+\tau,\eta,\tau \right)=\exp{\left( -2\pi i \eta\right) }F\left(\xi,\eta,\tau \right).
\end{equation}
\end{enumerate}	
\end{prop}
\begin{proof}
See \cite{Don}, Theorem 3. 
\end{proof}
In \cite{Don}, there is a Fourier expansion for $F$
\begin{equation}\label{Fourier}
F\left(\xi,\eta,\tau \right)=\pi i \left( \coth(\pi i \xi)+\coth(\pi i \eta)\right) +4 \pi \sum _{n=1}^{\infty}\left(\sum_{d|n}\sin\left(2\pi \left(\frac{n}{d}\xi+d\eta \right)  \right)  \right) q^{n}
\end{equation}
We fix a $\tau\in \mathbb{H}$ and consider $F$ restricted at $\tau$. We consider $\eta$ as a formal variable and we define the function ${g}^{(i)}$, $i\geq 0$ as the coefficients of
\[
\eta F\left(\xi,\eta,\tau \right)=\sum_{i\geq 0} g^{(i)}\eta^{i}.
\]
In particular the functions $g^{(i)}$ are meromorphic.
Formula \eqref{Fourier} give a way to describe the functions $g^{(i)}$ explicitly. First notice that
\[
\frac{1}{2}\coth\left(\frac{t}{2} \right)=\frac{1}{2}+\frac{1}{e^{t}-1} =\sum_{m=0}^{\infty}\frac{B_{2m}}{(2m)!}t^{2m-1}.
\]
Hence 
\begin{align*}
\pi i \eta\left( \coth(\pi i \xi)+\coth(\pi i \eta)\right) & = \pi i \eta +\frac{2 \pi i \eta}{e^{2 \pi i \xi}-1}+\sum_{m=0}^{\infty}\frac{B_{2m}}{(2m)!}(2 \pi i \eta)^{2m},
\end{align*}
by the de Moivre formula
\begin{align*}
&  4\pi \sum _{n=1}^{\infty}\left(\sum_{d|n}\sin\left(2\pi \left(\frac{n}{d}\xi+d\eta \right)  \right)  \right) q^{n} \\
& =-2\pi i\sum _{n=1}^{\infty}\sum_{d|n}\left( e^{2 \pi i\frac {n}{d}\xi}\left(\sum_{l=0}^{\infty}\frac{\left( 2 \pi i d \eta\right)^{l}}{l!}\right) -e^{-2 \pi i\frac {n}{d}\xi}\left(\sum_{l=0}^{\infty}\frac{\left( -2 \pi i d \eta\right)^{l}}{l!}\right) \right) q^{n}.
\end{align*}
Hence
\begin{align*}
g^{(0)}(\xi)& =1,\\
g^{(1)}(\xi)& = \pi i  +\frac{2 \pi i }{e^{2 \pi i \xi}-1}- 2 \pi i \sum_{n=1}^{\infty}\sum_{n|d}\left( e^{2 \pi i\frac {n}{d}\xi}-e^{-2 \pi i\frac {n}{d}\xi}\right)q^{n},\\
g^{(l)}(\xi)& =-\frac{\left( 2 \pi i \right)^{l}}{l!} \sum_{n=1}^{\infty}\left( \sum_{n|d}ld^{l}\left( e^{2 \pi i\frac {n}{d}\xi}+(-1)^{l}e^{-2 \pi i\frac {n}{d}\xi}\right)\right) q^{n} +\frac{\left( 2 \pi i \right)^{l}B_{l}}{l!}, \text{ for }l>1.
\end{align*}
In particular, for any $\tau$, $g^{(0)}=1$ and $g^{(1)}$ is a meromorphic function with simple poles at $\xi \in \Z+\tau \Z$ and $g^{(i)}$ for $i>1$ is holomorphic (see \cite{Don}). The quasi-periodicity of $F$ implies that the functions $\left\lbrace g^{(i)}\right\rbrace_{i\in \N}$ satisfies \eqref{conditions}. 
\begin{defi}[\cite{Damien},\cite{Levrac},\cite{Hain}]
Let $\tau$ be as above. The KZB connection on the punctured torus is a flat connection 
\[
d-\omega^{\tau}_{KZB,1}
\]
on the bundle $E_{F}$ as defined in the previous section, where
\[
\omega^{\tau}_{KZB,1}=-\sum_{p\geq 0}g^{(p)}(\xi)d\xi\operatorname{Ad}^{p}_{X_{0}}(X_{1})
\]
\end{defi}
\begin{thm}\label{main2}
We fix a $\tau\in\mathbb{H}$ and we consider the functions $\left\lbrace g^{(i)} \right\rbrace_{i\in\N} $. Let $C$ be as Proposition \ref{deg0geom}. Then $(d-r^*C, E_{F})$ is the KZB connection $(d-\omega^{\tau}_{KZB,1}, E_{F})$.
\end{thm}
\begin{proof} Let $\left(B, {m}_{\bullet} \right) $ be the $1$-model constructed in Lemma \ref{modelelliptic curve} and equipped with the decomposition \eqref{elliptic curvedec}. The forms $g^{(i)}(\xi)d\xi$ are holomorphic for $i\neq 1$ and $g^{(1)}(\xi)d\xi$ is a form with a logarithmic singularity. The proof follows by Theorem \ref{main1}.
\end{proof}

\subsection{A comparison with the KZ connection via $C_{\infty}$-algebras}\label{ComparisonHain}
We study some results of Section 12 of \cite{Hain}, where a link between universal KZ and KZB connection is given by considering the restriction of the connection to the first order Tate curve. We give an interpretation of such a result in terms of $C_{\infty}$-algebras over $\Q(2\pi i)$. As noticed by Hain in \cite{Hain}, $\lim_{\tau\to i\infty }\omega_{KZB,1}^{\tau}$ is equal to $\omega_{KZ,1}$ modulo a certain endomorphism $Q^{*}$ of complete Lie algebra. We use the argument of Subsection  \eqref{Morphism of homological pairs} to show that $Q^{*}$ is induced by a strict $C_{\infty}$-morphism $p_{\bullet}$. We denote by $z$ the coordinate on $\C^{*}$. We define the action of $\Z$ on $\C^{*}$ via
\begin{equation}\label{Zaction}
n \cdot z:= q^{n}z.
\end{equation}
There is a morphism $h_{\bullet}\: : \: \C_{\bullet}\Z^{2}\to \C^{*}_{\bullet}\Z$ of action groupoids
\[
h_{0}(\xi)=e^{2 \pi i\xi}, \quad h_{1}(\xi, (m,n) )=\left( e^{2 \pi i\xi}, n\right) 
\]
which induces an isomorphism on the quotient. We have $ \left\lbrace q^{\Z} \right\rbrace \subset \C^{*}$ and the maps above give a morphism $\left( \C_{\bullet}-\left\lbrace \Z^{2}+\tau\Z^{2}\right\rbrace\right)_{\bullet}\Z^{2} \to \left( \C^{*}-\left\lbrace q^{\Z} \right\rbrace\right)  _{\bullet}\Z$ between action groupoids that induce an isomorphism on the punctured elliptic curve. The quasi-periodicity of $F$ allows us to rewrite the $g^{(i)}$ as functions on $\left( \C^{*}-\left\lbrace q^{\Z} \right\rbrace\right)$. Then 
\begin{align}
\nonumber g^{(0)}(z)& =1,\\
\nonumber g^{(1)}(z)& = \pi i  +\frac{2 \pi i }{z-1}-\left( 2 \pi i \right)\sum_{n=1}^{\infty}\sum_{n|d}d\left( z^\frac{n}{d}-z^\frac {-n}{d}\right)q^{n},\\
\label{functionformal} g^{(l)}(z)& =-\frac{\left( 2 \pi i \right)^{l}}{l!}\left(  \sum_{n=1}^{\infty}\left( \sum_{n|d}ld^{l}\left( z^\frac{n}{d}+(-1)^{l}z^\frac {-n}{d}\right)\right) q^{n}\right)  +\frac{\left( 2 \pi i \right)^{l}B_{l}}{l!}, \text{ for }l>1.
\end{align}
\begin{lem}
The functions $g^{(l)}$ for $l\geq 0$ can be written as power series on $q$ where the coefficients are rational functions on $\C$ with poles on $0,1$ of the form $\frac{p_{1}}{p_{2}}$, where $p_{i}$ are polynomials over the field $\Q(2 \pi i)$ for $i=1,2$. 
\end{lem}
Given a subfield $\Q\subset \Bbbk\subset \C$, let $D\subset \C^{n}$ be a normal crossing divisor. We denote by $\operatorname{Rat}^{0}_{\Bbbk}(\C^{n},D)$ the algebra of rational functions $\frac{p_{1}}{p_{2}}$ with poles along $D$ such that $p_{1}, p_{2}$ are polynomials over the field $\Bbbk$. We denote by $\operatorname{Rat}^{\bullet}_{\Bbbk}(\C^{n},D)$ the differential graded $\Bbbk$-subalgebra of differential forms generated by forms of type $f dx_{I}$, with $f\in\operatorname{Rat}^{0}_{\Bbbk}(\C^{n},D)$. In particular $\operatorname{Rat}^{\bullet}_{\Bbbk}(\C^{n},D)\otimes \C\subset A_{DR}^{*}(\C^{n}-D)$.\\

We consider the differential graded  $\Q(2\pi i)$-algebra $\operatorname{Rat}^{\bullet}_{\Q(2 \pi i)}(\C,\left\lbrace 0,1 \right\rbrace )$ and we consider $q$ as a formal variable of degree zero. We denote by $\underline{g^{(i)}}\in \operatorname{Rat}^{0}_{\Q(2 \pi i)}(\C,\left\lbrace 0,1 \right\rbrace )((q))$, $i\geq 0$, the formal power series written as in \eqref{functionformal}. We have a differential graded algebra of formal Laurent series
\[
\left(d, \wedge, \operatorname{Rat}^{\bullet}_{\Q(2 \pi i)}(\C,\left\lbrace 0,1 \right\rbrace )\left( \left(  q \right) \right)   \right). 
\] 
We extend the action of $\Z$ defined in \eqref{Zaction}. For a module $A$ over a ring $\Bbbk$ and a group $G$, we denote by $\Map \left(G,A \right) $ the $\Bbbk$-module of maps from $G$ to $A$. An \emph{action of $G$ on $A$} is a morphism of module $\rho^{c}\: : \: A\to \Map \left(G,A \right) $ such that $\rho^{c}(a)(e)=a$ for any $a$ where $e$ is the identity, and $\rho^{c}(a)(gh)=\rho^{c}\left( \rho^{c}(a)(g)\right)(h)$, for any $g,h\in G$ and any $a$. The tensor product of two  actions of  $G$ on modules $A$ and $B$ is naturally an action of $G$ on $A\otimes B$. We call the \emph{trivial action} the action given by $\rho^{c}(a)(g)=a$ for any $g\in G$. 
We define $$\rho^{c}\: : \:\left(d, \wedge, \operatorname{Rat}^{\bullet}_{\Q(2 \pi i)}(\C,\left\lbrace 0,1 \right\rbrace )\left( \left(  q \right) \right)   \right)\to\Map\left(\Z, \operatorname{Rat}^{\bullet}_{\Q(2 \pi i)}(\C,\left\lbrace 0,1 \right\rbrace )\left( \left(  q \right) \right) \right) $$ via
\[
\rho^{c}\left( q\right)(n) :=q, \quad \rho^{c}\left(dz\right)(n):=q^{n}dz, \quad \rho^{c}\left(\frac{1}{z}\right)(n):= \frac{q^{-n}}{z} 
\]
and 
\[
 \rho^{c}\left(\frac{1}{z-1}\right) (n):=\begin{cases}
 \sum_{l=0}^{\infty}\left( q^{n}z\right)^{l} & n> 0\\
 \sum_{l=0}^{\infty}\frac{q^{-n}}{z}\left( \frac{1}{z^{l}}\right)^{l}& n<0.
 \end{cases}
\]
Let $d^{1}$ be the trivial action and set $d^{0}=\rho^{c}$. Hence $\left( d^{0},d^{1}\right) $ are the two cofaces map of a $1$-truncated cosimplicial unital differential graded algebra. The conerve gives rise to a cosimplicial unital differential commutative graded algebra $A^{\bullet, \bullet}$, where
\[
A^{l, m}:=\Map\left(\Z^{l},  \operatorname{Rat}^{m}_{\Q(2 \pi i)}(\C,\left\lbrace 0,1 \right\rbrace )\left( \left(  q \right) \right)\right).
\]
The differential graded $\Q(2\pi i)$-module $\operatorname{Tot}_{N}\left( A^{\bullet, \bullet}\right) $ carries a unital $C_{\infty}$-structure ${m}_{\bullet}$. We set $D:=m_{1}$. Let $\underline{\gamma}\in A^{1, 0}$ be the group homomorphism $\gamma\: : \:\left(  \Z,+\right) \to \left( \C,+\right) $ defined by $\underline{\gamma}(n):=2\pi i n$.
  We define 
\[
\underline{\phi}^{(i)}:=\underline{g^{(i)}}\frac{dz}{z}
\]
for any $i\geq 0$. Hence $\underline{\phi}^{(0)}$, $\underline{\gamma}$ are again closed elements in $\operatorname{Tot}_{N}\left( A^{\bullet, \bullet}\right) $.
Let $\underline{B}\subset \operatorname{Tot}_{N}\left( A^{\bullet, \bullet}\right)$ be the sub $C_{\infty}$-algebra generated by 
\[
1,\underline{\gamma},\left\lbrace \underline{\phi}^{(i)}\right\rbrace_{i\in \N}.
\]
Since all the calculations done for $B$ in the previous section are independent from the choice of $\tau\in \mathbb{H}$, we get that mutatis mutandis some of the results of Lemma \ref{modelelliptic curve} work in formal power series context. In particular there is a strict $C_{\infty}$-map
\[
f^{B}\: : \: \overline{B}\otimes \C\to B.
\]
\begin{prop} \label{modelelliptic curve2formal} Let $\left({m}_{\bullet}, \underline{B} \right) \subset \left({m}_{\bullet},\operatorname{Tot}_{N}\left( A^{\bullet, \bullet}\right) \right) $ be as above.
	\begin{enumerate}
		\item $D\left( -\underline{\phi}^{(n)}\right) =\sum_{l=1}^{n}{m}_{l+1}(\underline{\gamma},\dots, \underline{\gamma}, \underline{\phi}^{(p-l)})$, for any $n$.
		\item  $\left( {m}_{\bullet},\underline{B} \right) \subset \left({m}_{\bullet},\operatorname{Tot}_{N}\left( A^{\bullet, \bullet}\right) \right) $ is a rational sub $C_{\infty}$-algebra over $\Q(2\pi i)$
		where
		\[
		H^{1}(\underline{B}, D)= \left( \Q(2\pi i)\right) ^{2},\text{ and }	H^{2}(\underline{B}, D)=0.
		\]
		\item We have 
		\[
		\underline{B}^{0,1}\subset \operatorname{Rat}^{1}_{\Q(2 \pi i)}(\C,\left\lbrace 0,1 \right\rbrace )\left[ \left[  q \right] \right], \quad  	\underline{B}^{1,1}\subset \Map\left( \ \Z,\operatorname{Rat}^{1}_{\Q(2 \pi i)}(\C,\left\lbrace 0,1 \right\rbrace )\left[ \left[  q \right] \right]\right) 
		\]
		\item There exists a $\Q(2\pi i)$-vector space decomposition
		\begin{equation}\label{elliptic curvedecform}
			\underline{B}=W^{\bullet}\oplus D\mathcal{M}\oplus \mathcal{M}
		\end{equation}
		which is a Hodge type decomposition where
		\begin{enumerate}
			\item $W^{1}$ is generated by $-\underline{\phi}^{(0)}$ and $-\underline{\gamma}$, $\mathcal{M}^{1}$ is generated by $\underline{\phi}^{(i)}$ for $i>0$, and $\left( D\mathcal{M}\right)^{1}=0$,
			\item and $W^{2}=0$, $\mathcal{M}^{2}$ is generated by ${m}_{l}\left( \underline{\gamma}, \dots, \underline{\gamma},\underline{\phi}^{(i)}\right) $ for $i>0, l>2$, and $\left( D\mathcal{M}\right)^{1}$ is generated by $D\left( \underline{\phi}^{(i)}\right) $ for $i>1$.
		\end{enumerate}
	We denote by $\underline{g}_{\bullet}\: : \:\left(  W, m_{\bullet}^{W} \right) \to \underline{B}$ be the $1$-minimal model obtained via the homotopy transfer theorem.
		\item  We define $\left( W^{1}[1]\right)^{*} $ as the vector space generated by ${Y_{0}},{Y_{1}}$, where ${Y_{1}}$, and ${Y_{0}}$ resp. denote the dual of $-s\underline{\phi}^{(0)}$ and of $-s\underline{\gamma}$ respectively. The Maurer-Cartan in $\overline{B}\widehat{\otimes}\widehat{\mathbb{L}}\left(Y_{0}, Y_{1} \right)$ that corresponds to $\mathcal{F}(\underline{g}_{\bullet})$ is
		\[
	C_{Ell}^{\tau}=-\underline{\gamma}{Y_{0}}-\sum_{p\geq 0}\underline{\phi}^{(p)}\operatorname{Ad}^{p}_{Y_{0}}(Y_{1})
		\]	
	\end{enumerate}
\end{prop}
 \begin{proof}
 Since the functions defined in \eqref{functionformal} satisfy \eqref{conditions}, the proof of Lemma \ref{modelelliptic curve} carries over the situation above. This proves 1, 2 and 4. The statement 3 and 5 follow from \eqref{functionformal} and \eqref{Fourier} respectively.
 \end{proof}
 	Let $\underline{F}$ be denote the Kronecker function considered as a formal power series in $q$. We denote by
 	\[
 	r^{*}C_{Ell}^{\tau}:=-\sum_{p\geq 0}\underline{\phi}^{(p)}\operatorname{Ad}^{p}_{Y_{0}}(Y_{1})=-\operatorname{ad}_{Y_{0}}\underline{F}\left(\xi,\operatorname{ad}_{Y_{0}}  \right)(Y_{1})d\xi 
 	\]
 	In particular $\left( f^{B}\right)_{*}r^{*}C_{Ell}^{\tau}=\omega_{KZB,1}^{\tau}$. We consider $\operatorname{Rat}^{\bullet}_{\Q(2 \pi i)}(\C,\left\lbrace 0,1 \right\rbrace )\left( \left(  q \right) \right)$ as the trivial cosimplicial module. Notice that there is a morphism of cosimplicial differential graded modules 
 \[
 i\: : \:A^{\bullet, \bullet}\to  \operatorname{Rat}^{\bullet}_{\Q(2 \pi i)}(\C,\left\lbrace 0,1 \right\rbrace )\left( \left(  q \right) \right)   
 \]
 given as follows: for $f\in A^{l,m}$, $i(f):=f(0, \dots,0)\in \operatorname{Rat}^{m}_{\Q(2 \pi i)}(\C,\left\lbrace 0,1 \right\rbrace )\left( \left(  q \right) \right)  $. It induces a strict morphism of $C_{\infty}$-algebras
 \[
 i_{\bullet}\: : \: \operatorname{Tot}_{N}\left( A^{\bullet, \bullet}\right) \to  \operatorname{Rat}^{\bullet}_{\Q(2 \pi i)}(\C,\left\lbrace 0,1 \right\rbrace )\left( \left(  q \right) \right)    
 \]
  In particular, $i(B)$ is the commutative differential graded sub algebra of $\operatorname{Rat}^{\bullet}_{\Q(2 \pi i)}(\C,\left\lbrace 0,1 \right\rbrace )\left[ \left[   q \right] \right]  $ generated by
 \[
 1,\left\lbrace \underline{\phi}^{(i)}\right\rbrace_{i\in \N}.
 \]
 Note that $i(B^{n})=0$ for $n\geq 2$. Let $I_{q}$ be the completion of the augmentation ideal of $\operatorname{Rat}^{\bullet}_{\Q(2 \pi i)}(\C,\left\lbrace 0,1 \right\rbrace )\left[ \left[   q \right] \right] $, let $\pi'\: : \: \operatorname{Rat}^{\bullet}_{\Q(2 \pi i)}(\C,\left\lbrace 0,1 \right\rbrace )\left[ \left[   q \right] \right] \to \operatorname{Rat}^{\bullet}_{\Q(2 \pi i)}(\C,\left\lbrace 0,1 \right\rbrace )$ be the quotient with respect to $I_{q}$.
 \begin{lem}
 The map
   $$p':=\pi'\circ i\: : \: \underline{B}\to \operatorname{Rat}^{\bullet}_{\Q(2 \pi i)}(\C,\left\lbrace 0,1 \right\rbrace )$$
  is a strict $C_{\infty}$-algebra morphism such that $p'(\underline{\gamma})=0$ and
  \[
  p'(\underline{\phi^{(0)}})=\frac{dz}{2 \pi i z},\quad  p'(\underline{\phi^{(1)}})=\frac{dz}{2z}+\frac{dz}{z(z-1)}, \quad p'(\underline{\phi^{(l)}})=\frac{\left( 2 \pi i\right)^{l-1} B_{l}}{l!}\frac{dz}{z}\text{ for }l>1.
  \]
 \end{lem}
 We consider the complex variety $\C-\left\lbrace0,1 \right\rbrace  $. The cohomology is generated by the holomorphic forms
 \[
 1,\frac{dz}{z}, \quad \frac{dz}{z-1}.
 \]
 Let $A$ be the unital differential graded algebra over $\Q(2\pi i)$ generated by $1,\frac{dz}{z}, \frac{dz}{z-1}$. Notice that $A$ coincide with its cohomology. In particular, there is a canonical Hodge decomposition for $A$ given by $\left(A,0 \right) $, where $A^{1}$ is considered to equipped with the basis $\frac{dz}{z}, \frac{dz}{z-1}$. Then the $C_{\infty}$-morphism $g_{\bullet}^{KZ}\: : \: \left( W',m_{\bullet}^{W'}\right) \to A$ constructed via the homotopy transfer theorem is strict and corresponds to the identity map $\frac{dz}{z}\mapsto \frac{dz}{z},\frac{dz}{z-1} \mapsto\frac{dz}{z-1}$. We denote  by $C_{KZ}= \frac{dz}{z}Z_{0}+\frac{dz}{z-1}Z_{1}$ the Maurer-Cartan element in $A\widehat{\otimes}\widehat{\mathbb{L} }\left(Z_{0},Z_{1} \right)$ which corresponds (see \cite{Sibilia1}) to $\mathcal{F}\left(g_{\bullet}^{KZ} \right) $ where $Z_{0}:=\left( s\left( \frac{dz}{z}\right) \right) ^{*}$ and $Z_{1}:=\left(  s\left( \frac{dz}{z-1}\right) \right) ^{*}$. We denote the inverse of $g_{\bullet}^{KZ}$ by $f_{\bullet}^{KZ}$. Let $A_{KZ,1}$ be the complex differential graded algebra generated by $1,\frac{dz}{z}, \frac{dz}{z-1}$, the KZ connection is a flat connection on $\left( \C-\left\lbrace 0,1 \right\rbrace\right) \times \widehat{\mathbb{L} }\left(Z_{0},Z_{1} \right)$ given by $d-\omega_{KZ,1}$, where
\[
\omega_{KZ,1}=\frac{dz}{z}Z_{0}+\frac{dz}{z-1}Z_{1}
\]
 where the fiber is considered to be equipped with the adjoint action. In particular there is a differential graded algebra map
 \[
 f^{A}\: : \:A\otimes \C\to A_{KZ,1}
 \]
 such that $\left( f^{A}\right)_{*}C_{KZ}=\omega_{KZ,1}$. We apply the argument of Subsection \ref{Morphism of homological pairs}. We have a diagram of $1-C_{\infty}$-algebras
\[
\begin{tikzcd}
\mathcal{F}\left( \underline{B}\right) \arrow[r,  "p'"]
 &  \mathcal{F}\left(A\right) \arrow[d, shift right, "\mathcal{F}\left(f^{KZ}_{\bullet}\right) "]\\\mathcal{F}(W^{\bullet},m_{\bullet}^{W})\arrow[u, shift right, "\mathcal{F}\left( g_{\bullet}\right) "]& \mathcal{F}({W}'^{\bullet},m_{\bullet}^{W'})
\end{tikzcd}
\]
 We define $q_{\bullet}:=\mathcal{F}\left( f^{KZ}_{\bullet}\right) \circ\mathcal{F}\left( p'\right) \circ \mathcal{F}\left(g_{\bullet}\right) $, where $\mathcal{F}\left( g_{\bullet}\right) $ is the $1-C_{\infty}$-morphism that corresponds to the Maurer-Cartan element $C_{Ell}^{\tau}$ of Proposition \ref{modelelliptic curve2formal}. Notice that $g_{\bullet}^{KZ}$ has an inverse, then by Lemma \ref{connected maurerCartan}, we have that
 $q_{\bullet}\mathcal{F}\left( g^{KZ}_{\bullet}\right)=\mathcal{F}\left(f^{KZ}_{\bullet}\right) \circ \mathcal{F}\left(p'\right) \circ \mathcal{F}\left(g_{\bullet}\right) $. This equations can be written in terms of Maurer-Cartan elements in $\overline{B}\widehat{\otimes}\widehat{\mathbb{L}}\left(Y_{0}, Y_{1} \right)$ as 
 \[
 q^{*}C_{KZ,1}=\left(\mathrm{Id}\widehat{\otimes}Q^{*} \right)C_{KZ,1}=\left(p'\widehat{\otimes}\mathrm{Id} \right)C^{\tau}_{Ell}={p'}_{*}C^{\tau}_{Ell}.
 \]
 where $Q^{*}$ is a complete morphism of $\Q(2\pi i)$-Lie algebras 
 \[
 Q^{*}\: : \: \widehat{\mathbb{L}}\left(W'[1]\right)\to   \widehat{\mathbb{L}}\left(W^{1}[1]\right)
 \]
 which is an inclusion since both of the Lie algebras are free. Since
 \begin{align*}
 (p')^{*}\left(C_{Ell} \right)&=-\frac{dz}{2 \pi i z}Y_{1}+ \left( \frac{dz}{2z}+\frac{dz}{z(z-1)}\right) \left[ Y_{0}, Y_{1}\right]-\sum_{l>1}^{\infty}\frac{\left( 2 \pi i\right)^{l-1} B_{l}}{l!}\frac{dz}{z}\operatorname{Ad}^{l}_{Y_{0}}\left( Y_{1}\right) \\
 & = -\frac{dz}{ z}\frac{Y_{1}}{2 \pi i}- \left( +\frac{dz}{2z}-\frac{dz}{z-1}\right) \left[ 2\pi iY_{0}, \frac{Y_{1}}{2 \pi i}\right]+\sum_{l>1}^{\infty}\frac{ B_{l}}{l!}\frac{dz}{z}\operatorname{Ad}^{l}_{2 \pi iY_{0}}\left( \frac{Y_{1}}{2 \pi i}\right) \\
 &=-\frac{dz}{z-1}\left[ 2\pi iY_{0}, \frac{Y_{1}}{2 \pi i}\right]-\sum_{l=0}^{\infty}\frac{ {B'}_{l}}{l!}\frac{dz}{z}\operatorname{Ad}^{l}_{2 \pi iY_{0}}\left( \frac{Y_{1}}{2 \pi i}\right)\\
 & = q_{*}(C_{KZ})=\left( Id{\otimes}Q^{*}\right)C_{KZ},
 \end{align*}
 $Q^{*}$ is given by 
  \[
  Z_{0}\mapsto -\sum_{l=0}^{\infty}\frac{ {B'}_{l}}{l!}\frac{dz}{z}\operatorname{Ad}^{l}_{2 \pi iY_{0}}\left( \frac{Y_{1}}{2 \pi i}\right), \quad Z_{1}\mapsto-\left[ 2\pi iY_{0}, \frac{Y_{1}}{2 \pi i}\right] 
  \]
which is the map found by Hain in \cite{Hain} (see Section 18). The map $p'$ induces a Lie algebra morphism $Q^{*}\: : \: \widehat{\mathbb{L}}\left(W'[1]\right)\to   \widehat{\mathbb{L}}\left(W^{1}[1]\right)$ which induces a differential graded Lie algebra map
\[
q_{*}=\left( Id{\otimes}Q^{*}\right) \: : \: \left(  A'\right) \otimes\widehat{\mathbb{L}}\left(W'[1]\right)\to A'\widehat{\otimes}\widehat{\mathbb{L}}\left(W^{1}[1]\right)
\]
for any differential graded algebra $A'$.
\begin{prop}\label{thmfinale0}
 Let $\tau\in \mathbb{H}$ we denote the universal KZB-connection with $C_{Ell,\tau}$. We have $$\lim_{\tau\to i\infty }\omega_{KZB,1}^{\tau}=\left(Id{\otimes}Q^{*} \right)\omega_{KZ,1}$$
\end{prop}
\begin{proof}
\begin{align*}
\lim_{\tau\to i\infty }\omega_{KZB,1}^{\tau}& = \lim_{\tau\to i\infty }\left( f^{B}\right)^{*}C_{Ell}^{\tau}\\
& =\left( f^{B}\right)^{*} {p'}_{*}C_{Ell}^{\tau}\\
& =\left( f^{B}\right)^{*} q_{*}C_{KZ}\\
& =\left( f^{A}\right)^{*} q_{*}C_{KZ}\\
& = q_{*}\omega_{KZ,1}.
\end{align*}
\end{proof}
\subsection{A rational Maurer-Cartan element}\label{Arationalconnection}
Let $\tau$ be a fixed element of $\mathbb{H}:=\left\lbrace z\in \C \: : \: \Im(z)>0\right\rbrace $. We have considered the punctured elliptic $\mathcal{E}_{\tau}^{\times}$ curve as the quotient $\left(\C- \left\lbrace \mathbb{Z}+\tau\mathbb{Z}\right\rbrace \right)/\Z^{2}$. Equivalently it can be written as the solution set $\mathcal{E}_{\tau,\mathrm{alg}}^{\times}$ of
\[
y^2=x(x-1)(x-\lambda)
\]
on $\mathbb{P}^{2}$ minus the point at infinity. We set
\[
e_{1}:=\wp\left( \frac{1}{2},\tau\right) ,\: e_{2}:=\wp\left( \frac{\tau}{2},\tau\right) ,\text{ and }e_{3}:=\wp\left( \frac{1+\tau}{2},\tau\right) 
\]
The isomorphism between $\mathcal{E}_{\tau,\mathrm{alg}}^{\times}$ and $\mathcal{E}_{\tau}^{\times}$ is given by
\[
x=\frac{\wp(\xi,\tau)-e_{1}}{e_{2}-e_{1}},\quad y=\frac{\wp'(\xi,\tau)}{(e_{2}-e_{1})^{\frac{1}{2}}},\quad\lambda=j(\tau),
\]
where $j$ is the elliptic $j$-function. The two regular algebraic differential forms
\[
\frac{dx}{y},\quad \frac{xdx}{y}
\]
are well-defined and they generate the cohomology of $\mathcal{E}_{\tau,\mathrm{alg}}^{\times}$. Their pullback gives two holomorphic $1$-forms
\[
2(e_{2}-e_{1})^{\frac{1}{2}}d\xi,\quad\frac{ 2\wp(\xi,\tau)-2e_{1}}{(e_{2}-e_{1})^{\frac{1}{2}}}d\xi 
\]
 We denote by $A=\C\oplus\C d\xi\oplus \C\wp(\xi,\tau)d\xi$ the unital differential graded sub algebra generated by $d\xi$ and $\wp(\xi,\tau)d\xi$. We equipped this differential graded algebra with the obvious Hodge decomposition $\left(W',0 \right) $, where $W'=A$. We consider ${W'}^{1}$ equipped with the basis $2(e_{2}-e_{1})^{\frac{1}{2}}d\xi,\frac{ 2\wp(\xi,\tau)-2e_{1}}{(e_{2}-e_{1})^{\frac{1}{2}}}d\xi$. Let $\mathbb{L}({W'}^{*}_{+}[1])$ be the free Lie algebra with degree zero generator $T_{1}$ and $T_{2}$. In particular the identity map $W'\to A$ gives a $1$-minimal model ${g'}_{\bullet}\: :\: \left(W', m_{\bullet}^{W'} \right) \to A$ (here $m_{\bullet}^{W'}=0$). The degree zero geometric connection associated to $\mathcal{F}(g')$ is given by
\[
C_{\wp}^{\tau}:=2(e_{2}-e_{1})^{\frac{1}{2}}d\xi T_{0} + \frac{ 2\wp(\xi,\tau)-2e_{1}}{(e_{2}-e_{1})^{\frac{1}{2}}}d\xi T_{1}\in A \otimes\mathbb{L}({W'}^{*}_{+}[1])
\]
 We denote the Weierstrass zeta function by $\zeta\left(\xi, \tau \right)$.  It satisfies
\begin{eqnarray*}
\frac{\partial}{\partial \xi}\zeta\left(\xi, \tau \right)& =& -\wp(\xi,\tau),\\
\zeta\left(\xi+1, \tau \right)& =& \zeta\left(\xi, \tau \right)+2\eta_{1},\quad \eta_{1}=\zeta\left(\frac{1}{2},\tau \right), \\
\zeta\left(\xi+\tau, \tau \right)& =& \zeta\left(\xi, \tau \right)+2\eta_{2},\quad \eta_{2}=\zeta\left(\frac{\tau}{2},\tau \right), \\
\end{eqnarray*}
 We have
\begin{align*}
D\left( \zeta\left(\xi, \tau \right)-2\eta_{1} \xi\right) & =\wp \left(\xi, \tau \right)d\xi -2  \eta_{1}ds+\gamma\frac{ \left( 2\eta_{1}\tau -2\eta_{2}\right)}{2\pi i}\\
& = \wp \left(\xi, \tau \right)d\xi -2  \eta_{1}ds+\gamma
\end{align*}
since $2\eta_{1}\tau -2\eta_{2}=2\pi i$. Let $\left(B, {m}_{\bullet} \right) $ be the $1$-model constructed in Lemma \ref{modelelliptic curve} and equipped with the decomposition $\left(W, \mathcal{M} \right) $ as in \eqref{elliptic curvedec}. We consider $W$ equipped with the basis $-\gamma$ and $-d\xi$. Since 
\[
A,B\subset \operatorname{Tot}^{\bullet}_{N}\left(A_{DR}\left( \C- \left\lbrace \mathbb{Z}+\tau\mathbb{Z}\right\rbrace \right)_{\bullet}\Z^{2}\right),
\]
we have
\begin{eqnarray*}
[d\xi]& =& \frac{1}{2(e_{2}-e_{1})^{\frac{1}{2}}}\left[2(e_{2}-e_{1})^{\frac{1}{2}}d\xi   \right] \\
\left[ \gamma\right] & =& \frac{-(e_{2}-e_{1})^{\frac{1}{2}}}{2}\left[\frac{ 2\wp(\xi,\tau)-2e_{1}}{(e_{2}-e_{1})^{\frac{1}{2}}} \right]+\frac{2 \eta_{1}-e_{1}}{2(e_{2}-e_{1})^{\frac{1}{2}}}\left[ 2(e_{2}-e_{1})^{\frac{1}{2}}\right]   
\end{eqnarray*}
in $H^{1}\left(\operatorname{Tot}^{\bullet}_{N}\left(A_{DR}\left( \C- \left\lbrace \mathbb{Z}+\tau\mathbb{Z}\right\rbrace \right)_{\bullet}\Z^{2}\right),D \right) $. Consider the Lie algebra morphism\\ $K_{1}^{*}\: : \: \widehat{\mathbb{L}}\left( ({W'}^{1}_{+}[1])^{*}\right) \to \widehat{\mathbb{L}}\left( ({W}_{+}^{1}[1])^{*}\right) $ given by
\[
T_{0}\mapsto -\frac{1}{2(e_{2}-e_{1})^{\frac{1}{2}}}Y_{0}+\frac{2 \eta_{1}-e_{1}}{2(e_{2}-e_{1})^{\frac{1}{2}}}Y_{1}, \quad T_{1}\mapsto \frac{(e_{2}-e_{1})^{\frac{1}{2}}}{2}Y_{1}.
\]
Notice that $(\mathrm{Id}\widehat{\otimes}(K_{1})^{*})C_{\wp}$ corresponds to $$d\xi Y_{0}+\left( 2 \eta_{1}-\wp \left(\xi, \tau \right) d\xi\right)Y_{1}\in A_{DR}^{1}\left( \mathcal{E}_{\tau}^{\times}\right)\widehat{ \otimes} \widehat{\mathbb{L}}\left( ({W}_{+}^{1}[1])^{*}\right).$$
\begin{prop} There exists an isomorphism of free Lie algebra $K^{*}\: : \: \widehat{\mathbb{L}}\left( ({W'}^{1}_{+}[1])^{*}\right) \to \widehat{\mathbb{L}}\left( ({W}_{+}^{1}[1])^{*}\right) $ such that $K^{*}=K_{1}^{*}+\sum_{i=2}^{\infty}K_{i}^{*}$ and
 $\left(d-(\mathrm{Id}\widehat{\otimes}K^{*}) C_{\wp}^{\tau},\mathcal{E}_{\tau}^{\times}\times \widehat{\mathbb{L}}\left(Y_{0},Y_{1} \right)  \right)$ is isomorphic to $\left(d-\omega_{KZB,1}^{\tau}, E_{F} \right) $ via a holomorphic gauge.
\end{prop}
\begin{proof}
By \cite[Theorem 4.10]{Sibilia1}, the $\left(d- C_{\wp}^{\tau},\mathcal{E}_{\tau}^{\times}\times \widehat{\mathbb{L}}\left(T_{0},T_{1} \right)  \right)$ and $\left(d-\omega_{KZ,1}^{\tau}, E_{F} \right) $ are isomorphic, where the isomorphism is the composition of a fiber isomorhism induced by a Lie algebra isomorphism $(L^{*})$ and a gauge isomorphism. Since the fiber Lie algebra is a free Lie algebra, we can us \cite[Proposition 1.35]{Sibilia1} to conclude that $L^{*}=K^{*}$ and that $\left(d-(\mathrm{Id}\widehat{\otimes}K^{*}) C_{\wp}^{\tau},\mathcal{E}_{\tau}^{\times}\times \widehat{\mathbb{L}}\left(Y_{0},Y_{1} \right)  \right)$ is isomorphic to $\left(d-\omega_{KZ,1}^{\tau}, E_{F} \right) $ via a holomorphic gauge. The rest follows from \cite[Proposition 1.35]{Sibilia1}
\end{proof}
\begin{rmk}
As pointed out in \cite[Remark 4.11]{Sibilia1}, the morphism $K^{*}$ may be computed explicitly via the formula given in \cite{Prelie}.
\end{rmk}

\section{A connection on the configuration spaces of the punctured torus}\label{cpt}
In \cite{LevBrown} it is constructed a differential graded algebra $A_{n}$ which is a model for the differential graded algebra of complex smooth differential forms on $\operatorname{Conf}_{n}\left(\mathcal{E}_{\tau}^{\times} \right)$, moreover such a model can be constructed for any $\tau\in \mathbb{H}$ and it gives a family of differential graded algebras parametrized by $\tau\in \mathbb{H}$. We fix a $\tau\in \mathbb{H}$. In this section we construct a special $C_{\infty}$-algebra $B_{n}$ called \emph{$1$-extension for $A_{n}$} equipped with a Hodge type decomposition. We use the homotopy transfer theorem and we construct a $1$-minimal model for $B_{n}$. Such a map gives a family of Maurer-Cartan elements (parametrized by $s\in [0,1]$) in the differential graded Lie algebra $A_{DR}\left( \C^{n}-\mathcal{D}\right)\widehat{\otimes}\mathfrak{u}$, where $\mathfrak{u}$ is the Malcev Lie algebra of the pure elliptic braid group (see \cite{Damien}). In particular, for each $s\in [0,1]$ we have a smooth flat connection on $\operatorname{Conf}_{n}\left(\mathcal{E}_{\tau}^{\times} \right)$, and for $s=0$ this connection is the KZB connection on $\operatorname{Conf}_{n}\left(\mathcal{E}_{\tau}^{\times} \right)$ (see \cite{Damien}). In the last subsection we investigate the relation between the KZB connection and the KZ connection. In particular we construct a Lie algebra morphism between the fibers of the two connections and we give an $n$-dimensional version of Proposition \ref{thmfinale0}.

\subsection{$1$-Extensions}\label{1extension}
Let $G$ be a discrete group acting properly and discontinuously on a complex manifold $M$. Consider the action groupoid $M_{\bullet}G$ and assume that the cohomology of $M/G$ is connected and of finite type. By abuse of notation, we denote again by $m_{\bullet}$ the $C_{\infty}$-structure on $ \operatorname{Tot}_{N}\left(A_{DR}(M_{\bullet}G)\right)\otimes \Omega(1)$.
Let $J\subset \left( \operatorname{Tot}_{N}\left(A_{DR}(M_{\bullet}G)\right)\otimes \Omega(1), {m}_{\bullet}\right) $ be a $C_{\infty}$-ideal, i.e for any $k>1$ we have
\[
m_{k}\left(b_{1}, \dots, b_{k} \right)\in J 
\]
if some $b_{i}\in J$. Then $\left( \operatorname{Tot}_{N}\left(A_{DR}(M_{\bullet}G)\otimes \Omega(1)\right)/J, {m}_{\bullet}\right) $ is a $C_{\infty}$-algebra as well and the projection
\[\operatorname{Tot}_{N}\left(A_{DR}(M_{\bullet}G) \otimes \Omega(1)\right)\to \operatorname{Tot}_{N}\left(A_{DR}(M_{\bullet}G)\otimes \Omega(1)\right)/J
\]
is a strict $C_{\infty}$-morphism. Let $r\: : \: \operatorname{Tot}_{N}\left(A_{DR}(M_{\bullet}G)\right)\to A_{DR}(M)$ be the $C_{\infty}$ strict morphism induced by $M_{\bullet}\left\lbrace e\right\rbrace \to M_{\bullet}G$. Assume that $\left( r\otimes Id\right) (J)=0$, then we have a well-defined strict $C_{\infty}$-map
\[
 r\otimes Id\: : \:\operatorname{Tot}_{N}\left(A_{DR}(M_{\bullet}G)\right)\otimes \Omega(1)/J\to A_{DR}(M)\otimes \Omega(1).
\]
For $0\leq s \leq 1$ we denote by 
$$ev^{s}\: : \: \operatorname{Tot}_{N}A_{DR}\left(\left( \C^{n}-\mathcal{D}\right)_{\bullet}\left( \Z^{2n}\right) \right)\otimes \Omega(1)\to \operatorname{Tot}_{N}A_{DR}\left(\left( \C^{n}-\mathcal{D}\right)_{\bullet}\left( \Z^{2n}\right) \right)$$
 the evaluation map; it is a strict morphism of $C_{\infty}$-algebras.
\begin{defi}
Let $J$ be a $C_{\infty}$-ideal as above such that $\left( r\otimes Id\right) \left(J\right)=0=\left( Id\otimes ev_{1}\right)J$. Let $A$ be a model for $ \operatorname{Tot}_{N}\left(A_{DR}(M_{\bullet}G)\right)$. A $C_{\infty}$-algebra $\left(B, m_{\bullet}^{B}\right)$ is \emph{a $1$-extension for $A$} if there exist strict $C_{\infty}$-morphisms $H_{\bullet}\: : \:B\to \operatorname{Tot}_{N}\left(A_{DR}(M_{\bullet}G)\right)\otimes \Omega(1)/J$ and $f_{\bullet}\: : \: B\to A$ such that
\begin{enumerate}
\item  $f_{\bullet}$ induces an isomorphism on the 0-th and 1-th cohomology group and it is injective on the 2-th one, and
\item the diagram
\[
\begin{tikzcd}
& A_{DR}(M)\\
B\arrow{r}{f_{\bullet}}\arrow{rd}{H_{\bullet}}& A\arrow{u}{r}  &  \arrow[l,," Id\otimes ev_{1}"' ]\arrow[ul,," Id\otimes ev_{1}"' ]A_{DR}(M)\otimes \Omega(1)\\
& \left( \operatorname{Tot}_{N}\left(A_{DR}(M_{\bullet}G)\right)\otimes \Omega(1)\right) /J\arrow[ur,," r\otimes Id"' ]\arrow[u,," Id\otimes  ev_{1}"' ].
\end{tikzcd}
\]commutes.
\end{enumerate}
A \emph{compatible Hodge type decomposition for $B$} is a Hodge type decomposition $\left(W,\mathcal{M} \right) $ such that $\left( f(W), f(\mathcal{M})\right) $ is a Hodge type decomposition for $A$.
\end{defi}
Let ${g}_{\bullet}\: : \: \left(W, m_{\bullet}^{W} \right)\to \mathcal{F}\left( B, m_{\bullet}^{B}\right) $ be a $1$-minimal model. By \cite[Corollary 1.25 and Proposition 3.19]{Sibilia1} it corresponds to a Maurer-Cartan element in the $L_{\infty}$-algebra $B\widehat{ \otimes}\mathfrak{u}$, where $\mathfrak{u}$ is the Malcev Lie algebra of $\pi_{1}(M/G)$. We have a commuting diagram of Maurer-Cartan preserving maps
\[
\begin{tikzcd}
& A_{DR}(M)\widehat{\otimes}\mathfrak{u}\\
B\widehat{ \otimes}\mathfrak{u}\arrow{r}{f_{*}}\arrow{rd}{H_{*}}& A\arrow{u}{r_{*}}\widehat{ \otimes}\mathfrak{u}  &  \arrow[l,," \left( Id\otimes ev_{1}\right)_{*}"' ]\arrow[ul,," \left( Id\otimes ev_{1}\right)_{*}"' ]A_{DR}(M)\otimes \Omega(1)\\
& \left( \left( \operatorname{Tot}_{N}\left(A_{DR}(M_{\bullet}G)\right)\otimes \Omega(1)\right)/J\right) \widehat{ \otimes}\mathfrak{u} \arrow[ur,,"\left(r\otimes Id \right)_{*}"' ]\arrow[u,," \left( Id\otimes  ev_{1}\right)_{*}"'].
\end{tikzcd}
\]
In particular $\mathcal{F}\left( \left( r\otimes \mathrm{Id }\right) H_{\bullet}\right) g_{\bullet}\: : \: \left(W, m_{\bullet}^{W} \right)\to \mathcal{F}\left( A_{DR}(M)\otimes \Omega(1),d, \wedge\right)$ defines an homotopy between $ \mathcal{F}\left(rf_{\bullet}\right) g_{\bullet}$ and $\left(\mathrm{Id}\otimes ev_{0} \right)\mathcal{F}\left( \left( r\otimes \mathrm{Id }\right) H_{\bullet}\right)g_{\bullet}$. 
Let $C$ be the Maurer-Cartan elements associated to ${g}_{\bullet}$ (see \cite{Sibilia1}).
\begin{defi}\label{abuse}
We call $H_{*}C$ the \emph{good degree zero geometric connection associated to ${g}_{\bullet}$}.
\end{defi}
For $0\leq s\leq 1$, we set $$C(s):=\left( \mathrm{Id}\otimes \left( ev_{s}H\right) \right)_{*}\left(C \right) \in \left( \operatorname{Tot}^{1}_{N}A_{DR}\left(M_{\bullet}G \right)/J\right) \widehat{\otimes} \mathfrak{u}$$
In particular ${C}(0)=\left(\mathrm{Id}\otimes ev_{0} \right)_{*}r_{*}H_{*}C$ is homotopy equivalent to $C\left( 1\right) =r_{*}f_{*}C$ as a Maurer-Cartan element (see \cite{Sibilia1}).
\begin{thm}\label{thmbundleholonomy}
Assume that $A\subset A_{DR}(M/G)$ and that $\mathcal{F}\left(rf_{\bullet}\right)g_{\bullet} \: : \: \left(W, m_{\bullet}^{W} \right)\to \mathcal{F}\left( A, d, \wedge\right) $ can be lifted to a minimal model for $\left( A, d, \wedge\right)$. For $0\leq s\leq 1$, there exists factor of automorphy
\[
F_{s}\: : \: M\times \mathfrak{u}\to M\times \mathfrak{u}
\] such that $\left(d-{C}(s), E_{{F}_{s}}\right)$ is a well-defined smooth flat connection whose monodromy representation corresponds to the Malcev completion of the fundamental group of $M/G$. \begin{enumerate}
\item If $C(s)$ has holomorphic coefficients, $\left(d-C(s), E_{{F}_{s}}\right)$ is a flat connection on a holomorphic bundle
\item Assume that $M_{\bullet}G=(N-\mathcal{D})_{\bullet}G$, where $\left( N,\mathcal{D}\right) $ is a complex manifold with a normal crossing divisors and $G$ be a group acting holomorphically on $N$ and that preserves $\mathcal{D}$. If $C(s)$ has holomorphic coefficients with logarithmic singularities along $\mathcal{D}$ then Then $F_{s}\: : \:N\times \mathfrak{u}\to N\times \mathfrak{u}$ is holomorphic and $d- C(s)$ is a well-defined holomorphic flat connection with logarithmic singularities on $M/G$ on the holomorphic bundle $E_{F_{s}}$.
\end{enumerate}
For any $s,s'\in [0,1]$ there exists a bundle isomorphism $T$ that preserves the connections $T\: : \: \left(d-C(s), E_{{F}_{s}}\right)\to \left(d-C(s'), E_{{F}_{s'}}\right)$ which is holomorphic (on $N/G$) if $C(s)$ and $C(s')$ are holomorphic flat connections (with logarithmic singularities).
\end{thm}
\begin{proof}
Since $\mathcal{F}\left(rf_{\bullet}\right)g_{\bullet}  \: : \: \left(W, m_{\bullet}^{W} \right)\to \mathcal{F}\left( A, d, \wedge\right) $ can be lifted to a minimal model for $\left( A, d, \wedge\right)$, there is a good homological pair $(C', \delta')$ associated to it (see \cite{Sibilia1}) such that $\pi(C')=C\left( 1\right) =r_{*}f_{*}C$ in $A_{DR}\left(M_{\bullet}\right) \widehat{\otimes} \mathfrak{u}$. On the other hand $r_*H_{*}C$ is an homotopy equivalence between Maurer-Cartan elements, and it induces (see \cite{Sibilia1}) a gauge equivalence between $C(0)$ and $C(s')$ for any $s$. The theorem follows from \cite[Theorem 4.9 and Theorem 4.10]{Sibilia1}.
\end{proof}
Since $C(1)$ is the ordinary Chen's flat connection, by  \cite[Theorem 4.10]{Sibilia1} we get the following.
\begin{cor}\label{invariance}
Let $(B', A', {f'}_{\bullet}, {H'}_{\bullet})$  be a $1$-extension for a model $A'\subset A_{DR}(M/G)$. Let ${g'}_{\bullet}$ be a $1$-model for $\mathcal{F}(B')$. Assume that $\mathcal{F}\left(r{f'}_{\bullet}\right){g'}_{\bullet} \: : \: \left(W, m_{\bullet}^{W} \right)\to \mathcal{F}\left( A', d, \wedge\right) $ can be lifted to a minimal model for $\left( A', d, \wedge\right)$. Then the resulting connections are isomorphic with the ones obtained by Theorem \ref{thmbundleholonomy}.
\end{cor}

\subsection{Differential forms on $\operatorname{Conf}_{n}\left(\mathcal{E}_{\tau}^{\times} \right)$}\label{sectionformulas}
Let $\tau$ be a fixed element of the upper complex plane $\mathbb{H}$. Let $\mathbb{Z}+\tau\mathbb{Z}$ be the lattice spanned by $1,\tau$. Let $(\xi_{1}, \dots \xi_{n})$ be the coordinates on $\C^{n}$. We define $\xi_{0}:=0$ and for $i=1, \dots, n$ we define $r_{i},s_{i}$ via $\xi_{i}=s_{i}+\tau r_{i}$. We define $\mathcal{D}  \subset \C^{n}$ as 
\[
\mathcal{D}:=\left\lbrace (\xi_{1}, \dots \xi_{n})\: : \: \xi_{i}-\xi_{j}\in \mathbb{Z}+\tau\mathbb{Z} \text{ for some distinct }i,j=0, \dots n  \right\rbrace 
\] 
We define a $\Z^{2n}$-action on $\C^{n}$ via translation, i.e.
\[
\left( \left({l}_{1},m_{1} \right), \dots, \left({l}_{n},m_{n} \right) \right)(\xi_{1}, \dots \xi_{n}):= \left( \xi_{1}+l_{1}+m_{1}\tau, \dots ,\xi_{n}+l_{n}+m_{n}\tau\right) 
\]
Notice that $\mathcal{D}$ is preserved by the action of $\Z^{2n}$. There is a canonical isomorphism
\[
\left( \C^{n}-\mathcal{D}\right) /\left( \Z^{2n}\right) \cong \operatorname{Conf}_{n}\left({\mathcal{E}}_{\tau}^{\times} \right)
\]
since the action is free and properly discontinuous. We denote the action groupoid by $$\left( \C^{n}-\mathcal{D}\right)_{\bullet}\left( \Z^{2n}\right),$$ it is a simplicial manifold equipped with a simplicial normal crossing divisor (see  \cite[Section 3]{Sibilia1}). Its de Rham complex $A_{DR}^{\bullet}(\left( \C^{n}-\mathcal{D}\right)_{\bullet}\left( \Z^{2n}\right))$ is a cosimplicial commutative (non-negatively graded) differential graded algebra. By the simplicial de Rham theorem (see \cite{Dupont2}) and the discussion above we have
\[
H^{\bullet}\left( \operatorname{Tot}_{N}\left(A_{DR}\left( \C^{n}-\mathcal{D}\right)_{\bullet}\left( \Z^{2n}\right) \right) \right) \cong H^{\bullet}\left(\left( \C^{n}-\mathcal{D}\right)/ \left( \Z^{2n}\right), \C \right) \cong H^{\bullet}\left(\operatorname{Conf}_{n}\left({\mathcal{E}}_{\tau}^{\times} \right),\C \right).
\]
By \cite{Getz} the normalized total complex of a cosimplicial commutative algebras carries a natural $C_{\infty}$-structure ${m}_{\bullet}$ (see \cite[Section 2]{Sibilia1} for explicit formulas).  
\begin{prop}
Let $F(\xi,\eta,\tau)$ be the Kronecker function (see Section \ref{ComparisonHain0}), where $\eta$ is a formal variable, $F(\xi,\eta,\tau)$ satisfies the \emph{Fay}'s identity, i.e.
\begin{align}
& F\left(\xi_{1},\eta_{1},\tau \right)F\left(\xi_{2},\eta_{2},\tau \right)=F\left(\xi_{1},\eta_{1}+\eta_{2},\tau \right)F\left(\xi_{1}-\xi_{2},\eta_{2},\tau \right)\nonumber\\
& +F\left(\xi_{2},\eta_{1}+\eta_{2},\tau \right)F\left(\xi_{1}-\xi_{2},\eta_{1},\tau \right).\label{fay} 
\end{align}
\end{prop}

\begin{proof}
It follows from standard properties of theta series.
\end{proof}
We define the $1$-forms $\phi_{i,j}^{(k)}$ for $k\geq 0$ $i,j=0,1,\dots,n$ as follows. Let $\alpha$ be a formal variable, then
\[
\sum_{k\geq 0}\phi_{i,j}^{(k)}\alpha^{k}:=\alpha F(\xi_{i}-\xi_{j}, \alpha, \tau)d(\xi_{i}-\xi_{j}).
\]
Thanks to the Fourier expansion \ref{Fourier}
\begin{align*}
\phi^{(0)}_{i,j}& =d \xi_{i}-d \xi_{j},\\
\phi^{(1)}_{i,j}& = \left( \pi i  +\frac{2 \pi i }{e^{2 \pi i \left(  \xi_{i}-\xi_{j}\right)}-1}-\left( 2 \pi i \right)^{2}\sum_{n=1}^{\infty}\sum_{n|d}d\left( e^{2 \pi i\frac {n}{d}\left(  \xi_{i}-\xi_{j}\right)}-e^{-2 \pi i\frac {n}{d}\left(  \xi_{i}-\xi_{j}\right)}\right)q^{n}\right)d\left(  \xi_{i}-\xi_{j}\right)  ,\\
\phi^{(l)}_{i,j}& =-\left( \frac{\left( 2 \pi i \right)^{l+1}}{l!} \sum_{n=1}^{\infty}\left( \sum_{n|d}d^{l}\left( e^{2 \pi i\frac {n}{d}\left(  \xi_{i}-\xi_{j}\right) }+(-1)^{l}e^{-2 \pi i\frac {n}{d}\left(  \xi_{i}-\xi_{j}\right) }\right)\right) q^{n} +\frac{B_{l}}{2 \pi i}\right) d\left(  \xi_{i}-\xi_{j}\right) , \text{ for }l>1.
\end{align*}
i.e. $\phi^{(l)}_{i,j}$ are holomorphic $1$-forms for $l\neq 1$ and for $l=1$ the are meromorphic with a pole of order $1$ along the hyperplane $\xi_{i}=\xi_{j}$. 
Let $\Omega(1)$ be the differential graded algebra of polynomials forms on the $1$ dimensional simplex with coordinate $0<u<1$. We define a parametrized $1$-form $\Omega_{u}\left(\xi, \alpha \right) :=\exp(2\pi iur\alpha )F(\xi, \alpha, \tau)d \xi$ on $\C-\left\lbrace \Z+\Z\tau\right\rbrace $. For $0\leq i\leq j\leq n$, we define the  $w(u)_{i,j}^{(k)}\in A_{DR}^{1}\left( \C^{n}-\mathcal{D}\right)\otimes \Omega^{0}(1)$ as 
\begin{equation}\label{generalforms}
\Omega_{u}\left(\xi_{i}-\xi_{j}, \alpha \right):=\sum_{k\geq 0}w(u)^{(k)}_{i,j}\alpha^{k-1}.
\end{equation}
Thanks to the discussion above we get that the $w(u)_{i,j}^{(k)}$ are smooth $u$-valued forms on $\C^{n}$ with logarithmic singularities along $\mathcal{D}$. Notice that 
\begin{enumerate}
\item $w(0)^{(k)}_{i,j}=\phi^{(k)}_{i,j}$ for any $i,j$ distinct, and 
\item $w(1)^{(k)}_{i,j}\in A_{DR}\left(\operatorname{Conf}_{n}\left(\mathcal{E}_{\tau}^{\times} \right) \right)\subset A_{DR}\left(\C^{n}-\mathcal{D} \right) $.
\end{enumerate}
 The quasi-periodicity of $F$ implies that the pullback of the $\Z^{2n}$-action is 
\begin{equation}\label{rel4}
w(u)^{(k)}_{i,j}\left( \xi+{l}_{1}+m_{1}\tau , \dots, \xi_{n}+{l}_{n}+\tau m_{n} \right)= \sum_{p=0}^{k}{w(u)^{(k-p)}_{i,j}}\frac{2\pi i(u-1) \left( l_{i}-l_{j}\right)^{p}}{p!}.
\end{equation}
 The Fay identity gives the following quadratic relations between the $w(u)_{i,j}^{(k)}$
\begin{align}
& \Omega_{u}\left(\xi_{i}-\xi_{l}, \alpha \right)\wedge \Omega_{u}\left(\xi_{j}-\xi_{l}, \beta \right)+\Omega_{u}\left(\xi_{i}-\xi_{i}, \beta \right)\wedge \Omega_{u}\left(\xi_{i}-\xi_{l},\alpha+\beta \right)\nonumber\\
& +\Omega_{u}\left(\xi_{j}-\xi_{l},\alpha+\beta \right)\wedge \Omega_{u}\left(\xi_{i}-\xi_{j}, \alpha \right)=0.\label{rel1}
\end{align}
For distinct indices we have
\begin{eqnarray*}
	w(u)_{i,l}^{a-1}w(u)_{j,l}^{b}& +& w(u)_{i,j}^{(a)}w(u)_{j,l}^{(b-1)}
	+\sum _{m = 0}^{b}\binom{a+b-1-m}{a-1}w(u)_{j,i}^{(m)}w(u)_{i,l}^{(a+b-1-m)}\\&+&\sum_{k=0}^{a}\binom{a+b-1-k}{b-1}w(u)_{j,l}^{(a+b-1-k)}w(u)_{i,j}^{(k)}=0
\end{eqnarray*}
The function $F(\xi, \alpha, \tau)$ satisfies $F(\xi, \alpha, \tau)=-F(-\xi, -\alpha, \tau)$. This implies
\begin{equation}\label{rel2}
w(u)^{(k)}_{i,j}+(-1)^{k}w(u)^{(k)}_{j,i}=0.
\end{equation}
On the other hand we have 
\begin{equation}\label{rel3}
w(u)^{(0)}_{i,j}=w(u)^{(0)}_{i,0}-w(u)^{(0)}_{j,0}.
\end{equation}
 For $i=0,\dots n$ we define the group homomorphism $$\gamma(u)_{i}\: : \: \Z^{2n}\to \left( A_{DR}^0\left( \C^{n}-\mathcal{D} \right)\otimes \Omega^{0}(1),+ \right) $$ via $\gamma(u)_{0}:=0$ and for $j\neq 0$ via $\gamma(u)_{j}\left( \left({l}_{1},m_{1} \right), \dots, \left({l}_{n},m_{n} \right) \right):=2 \pi i m_{j}\otimes (1-u)$. We define $\gamma(u)_{i,j}:=\gamma(u)_{i}-\gamma(u)_{j}$.  For $i=0,\dots ,n$ we define $\nu(u)_{i}\in A_{DR}^{1}\left( \C^{n}-\mathcal{D} \right)\otimes \Omega^{0}(1)$ via $\nu(u)_{0}:=0$ and $\nu(u)_{i}:= 2 \pi i dr_{i}\otimes u $ for $i\neq 0$. We set $\nu(u)_{i,j}:= \nu(u)_{i}-\nu(u)_{j}$. Finally,  for $i=0,\dots ,n$ we define $\beta(u)_{i}\in A_{DR}^{0}\left( \C^{n}-\mathcal{D} \right)\otimes \Omega^{1}(1)$ via $\beta(u)_{0}:=0$ and $\beta(u)_{i}:= 2 \pi i r_{i}\otimes du $ for $i\neq 0$. We set $\beta(u)_{i,j}:= \beta(u)_{i}-\beta(u)_{j}$. We denote by $d$ the differential of $A_{DR}\left( \C^{n}-\mathcal{D} \right)\otimes \Omega(1)$. 
\begin{lem}\label{lemmarel1}
We have the following relations: $d\left( \nu(u)_{i,j}+\beta(u)_{i,j}\right)=dw(u)_{i,j}^{(0)}= 0$ for $0\leq i \leq j \leq n$. For $k>0$ we have
\[
d w(u)_{i,j}^{(k)}=\left( \nu(u)_{i,j}+\beta(u)_{i,j}\right) w(u)_{i,j}^{(k-1)}.
\]
\end{lem}
\begin{proof}
It follows by $d\Omega_{u}\left(\xi_{i}-\xi_{j}, \alpha \right)=\left( \nu(u)_{i,j}\alpha+\beta(u)_{i,j}\alpha \right)\Omega_{u}\left(\xi_{i}-\xi_{j}, \alpha \right) $.
\end{proof}
The action of $\Z^{2n}$ on $\C^{n}-\mathcal{D}$ induces an action on $A_{DR}\left( \C^{n}-\mathcal{D} \right)$ We consider the differential graded algebra $A_{DR}\left( \C^{n}-\mathcal{D} \right)$.  We put the trivial $\Z^{2n}$-action on $\Omega(1)$. We denote the resulting action by $\rho^{c}\: : \: A_{DR}\left( \C^{n}-\mathcal{D} \right)\otimes \Omega(1)\to \Map\left(\Z^{2n}, A_{DR}\left( \C^{n}-\mathcal{D} \right)\otimes \Omega(1)\right)$. 
\medskip
\begin{lem}\label{lemmarel2}
Let $g=\left( \left({l}_{1},m_{1} \right), \dots, \left({l}_{n},m_{n} \right) \right),g'=\left( \left({l'}_{1},{m'}_{1} \right), \dots, \left({l'}_{n},{m'}_{n} \right) \right)\in \Z^{2n}$. For $0\leq i\leq j\leq n$ we have
\[
\rho^{c}(\gamma(u)_{i,j}(g'))(g)=\gamma(u)_{i,j}(g'), \quad \rho^{c}\left( \nu(u)_{i,j}\right)(g)=\nu(u)_{i,j}, 
\]
$\rho^{c}\left( \beta(u)_{i,j}\right)(g)=-d\gamma(u)_{i,j}(g)+\beta(u)_{i,j}$ and 
\[
\rho^{c}\left( w(u)_{i,j}^{(k)}\right)(g)=\sum_{p=0}^{k}{ w(u)_{i,j}^{(k-p)}}\frac{\left( -\gamma(u)_{i,j}(g)\right) ^{p}}{p!}.
\]
Moreover the action preserves the differential and the wedge product, i.e. $\rho^{c}(da)(g)=d\left( \rho^{c}(a)(g)\right) $, $\rho^{c}(ab)(g)=\left( \rho^{c}(a)(g)\right) \left( \rho^{c}(b)(g)\right) $ for any $a,b\in A_{DR}\left( \C^{n}-\mathcal{D} \right)\otimes \Omega(1)$ and $g\in \Z^{2n}$.
\end{lem}
\begin{proof}
The first two and the last identities are immediate. The third one follows from the shifting property of $F$.
\end{proof}
\subsection{The $1$-extension $B_{n}$}
We define the $C_{\infty}$-algebras $A_{n},B_{n}, {A'}_{n}$ and ${B'}_{n}$. They are rational $C_{\infty}$-algebras, but we consider them as $C_{\infty}$-algebras over $\C$.
In \cite{LevBrown}, a model $A_{n}$ for the configuration space of points of the punctured elliptic curve is constructed. We extend the ideas of \cite{LevBrown} for $C_{\infty}$-algebras.
\begin{defi}\label{defAnBr}
Let ${A}_{n}$ be the complex unital commutative graded algebra generated by the degree 1 symbols $w(1)_{i,j}^{(k)}$, $\nu(1)_{i}$, for $k\geq 0$ $i,j=0,1,\dots,n$
modulo the relations \eqref{rel1}, \eqref{rel2}, \eqref{rel3}. We define a differential $d$ via
$dw(1)_{i,j}^{(0)}=d\nu(1)_{i}=0$ and
\[
dw(1)_{i,j}^{(k)}=\left( \nu(1)_{i}-\nu(1)_{j}\right)w(1)_{i,j}^{(k-1)}
\]
for $k>0$.
\end{defi}
Notice that the elements $w(1)_{i,j}^{(k)}$, $\nu(1)_{i}$, for $k\geq 0$ $i,j=0,1,\dots,n$ denote elements in $A_{DR}\left( \C^{n}-\mathcal{D}\right)$ as well (see previous section).\\ There is an obvious map $\psi^{1}\: : \: A_{n}\to A_{DR}\left(\operatorname{Conf}_{n}\left(\mathcal{E}_{\tau}^{\times} \right) \right)\subset A_{DR}\left(\C^{n}-\mathcal{D} \right)$ defined by
\[
\psi^{1}\left(w(1)_{i,j}^{(k)} \right):=w(1)_{i,j}^{(k)},\quad \psi^{1}\left(\nu(1)_{i} \right):=\nu(1)_{i}
\]
for any $i,j,k$.
\begin{thm}[\cite{LevBrown}]\label{model}
The map $\psi^{1}\: : \: A_{n}\to A_{DR}\left(\operatorname{Conf}_{n}\left(\mathcal{E}_{\tau}^{\times} \right) \right)$ is an inclusion and a quasi-isomorphism.
\end{thm}
We define a parametrized version of $A_{n}$.
\begin{defi}\label{defAn}
Let ${A'}_{n}$ be the complex unital commutative graded algebra generates by
\begin{enumerate}
	\item the degree 0 symbol $\tilde{\gamma}(u)$  and,
	\item the degree 1 symbols $w(u)_{i,j}^{(k)}$, $\nu(u)_{i}$, $\beta(u)_{i}$ for $k\geq 0$ $i,j=0,1,\dots,n$
\end{enumerate}
 modulo the relations \eqref{rel1}, \eqref{rel2}, \eqref{rel3} and such that
\[
\beta(u)_{i}\beta(u)_{j}=0, \quad  \nu(u)_{0}=\beta(u)_{0}=0
\]
We denote  $\beta(u)_{i,j}:= \beta(u)_{i}-\beta(u)_{j}$ and $\nu(u)_{i,j}:= \nu(u)_{i}-\nu(u)_{j}$. We define a differential $d$ via the relations of Lemma \ref{lemmarel1}. This makes ${A'}_{n}$ a differential graded commutative algebra.
\end{defi}
 The notation of the generators is justified by the following fact.
\begin{prop}\label{injbrown}
The differential graded algebra map $\Psi\: : \: {A'}_{n}\to A_{DR}\left(\C^{n}-\mathcal{D}\right)\otimes \Omega(1)$ sending $w(u)_{i,j}^{(k)}$,$\nu(u)_{i,j}$ and $\beta(u)_{i,j}$ to the 1-forms represented by these symbols and $\tilde{\gamma}(u)$ to $2 \pi i \otimes (1-u)$ is injective.
\end{prop}
\begin{proof}
We use some results of Section 4. of \cite{LevBrown}. We consider the obvious map of commutative differential graded algebra $\Psi\: : \: {A'}_{n}\to A_{DR}\left(\C^{n}-\mathcal{D}\right)\otimes \Omega(1)$. In particular \cite[Lemma 15]{LevBrown} works as well and the same argument of \cite[Corollary 16]{LevBrown} implies that $\Psi$ is injective. 
\end{proof}
Let ${A'}_{n}$ as above. For $i=0,\dots n$ we define the group homomorphism $\gamma(u)_{i}\: : \: \Z^{2n}\to  A_n$ via $\gamma(u)_{0}:=0$ and for $j\neq 0$ via $\gamma(u)_{j}\left( \left({l}_{1},m_{1} \right), \dots, \left({l}_{n},m_{n} \right) \right):= m_{j}\tilde{\gamma}(u)$. We define $\gamma(u)_{i,j}:=\gamma(u)_{i}-\gamma(u)_{j}$.
\begin{lem}\label{Anstruct}
	 We define a map $\rho^{c}\: : \: {A'}_{n}\to  \Map\left(\Z^{2n},  {A'}_{n}\right) $ via the relations of Lemma \ref{lemmarel2}, i.e
\[
\rho^{c}(\tilde{\gamma}(u))(g):=\tilde{\gamma}(u), \quad \rho^{c}\left( \nu(u)_{i,j}\right)(g):=\nu(u)_{i,j}, \quad \rho^{c}\left( \beta(u)_{i,j}\right)(g):=-d\left( \gamma(u)_{i,j}(g)\right) +\beta(u)_{i,j}
\]
and 
\[
\rho^{c}\left( w(u)_{i,j}^{(k)}\right)(g):=\sum_{p=0}^{k}{ w(u)_{i,j}^{(k-p)}}\frac{\left( -\gamma(u)_{i,j}(g)\right) ^{p}}{p!}.
\]
	\begin{enumerate}
		\item $\rho^{c}$ is a $(\Z^{2n})$-action.
		\item We have $d\rho^{c}(a)(g)=\rho^{c}(da)(g)$ and $\rho^{c}(ab)(g)=\rho^{c}(a)(g)\rho^{c}(b)(g)$ for any $a,b\in {A'}_{n}$.
		\item $\Psi$ respects the $\Z^{2n}$-action.
	\end{enumerate}
\end{lem}
	\begin{proof}
		The proof is a direct verification.
	\end{proof}
 The (co)nerve of the action defines a cosimplicial commutative differential graded algebra, we denote it by $A^{\bullet,\bullet}$. Concretely $A^{p,\bullet}$ is the differential graded algebra $\Map\left(\left( \Z^{2n}\right)^{p},  {A'}_{n}\right)$. The conormalization $N(A)^{\bullet, \bullet}$ is a bidifferential bigraded module where the second differential $\partial_{\Z^{2n}}$ is induced by the action. 
 By \cite{Getz} the differential graded module $\operatorname{Tot}_{N}(A)$ carries a natural $C_{\infty}$-structure ${m}_{\bullet}$ (see \cite[Section 2]{Sibilia1} for explicit formulas). 
	\begin{defi}\label{generalalgebra}
 We denote by $\left( {B'}_{n},{m}_{\bullet}\right) $ the rational $C_{\infty}$-subalgebra of $\operatorname{Tot}_{N}(A)$ generated by
	\[
w(u)_{i,j}^{(k)},\quad \alpha(u)_{i,j}:=\gamma(u)_{i,j}-\nu(u)_{i,j}-\beta(u)_{i,j}, 
	\]
for $k\geq 0,\,i,j=0,1,\dots,n$.	
\end{defi}
\begin{prop}\label{propquot}
Consider the differential graded algebra $A_{n}$. There is a strict $C_{\infty}$-morphism $p^{1}\: :\: {B'}_{n}\to A_{n}$ defined via
\begin{equation}
\label{morphism}
p^{1}\left(w(u)_{i,j}^{(k)} \right):=w_{i,j}^{(k)},\quad p^{1}\left(\alpha(u)_{i,j}\right):=-\nu_{i,j}
\end{equation}
for any $0\leq i\leq j\leq n$ and $k\geq 0$.
\end{prop}
\begin{proof}
We construct $p^{1}$ in a functorial way. First we consider the map $q\: :\: {A'}_{n}\to A_{n}$ defined via
\begin{equation*}
q\left(w(u)_{i,j}^{(k)} \right):=w_{i,j}^{(k)},\quad p^{1}\left(\nu_{i,j}\right):=\nu_{i,j},\quad q\left(\beta_{i,j}\right)=0
\end{equation*}
 and $p^{1}(\tilde{\gamma}(u))=0$ for any $0\leq i\leq j\leq n$ and $k\geq 0$. Then $q$ can be extended to a differential graded algebra map. Now consider $A_{n}$ equipped with the trivial $\Z^{2n}$ action. The map $q$ is $\Z^{2n}$ equivariant and it can be extended to a map between differential graded cosimplicial algebras $ q\: :\: \Map\left(\Z^{2n}, {A'}_{n}\right)\to \Map\left(\Z^{2n}, {A}_{n}\right)$. Then $q$ induces a map of $C_{\infty}$-algebras
 \[
 \left( \operatorname{Tot}_{N}\left(\Z^{2n}, {A'}_{n}\right), {m'}_{\bullet}\right) \to \left( \operatorname{Tot}_{N}\left(\Z^{2n}, {A}_{n}\right), {m'}_{\bullet}\right)
 \]
 Moreover this map is strict. The normalized total complex of $\Map\left(\Z^{2n}, {A}_{n}\right)$ is $A_{n}$ and its $C_{\infty}$-structure corresponds to the ordinary differential graded algebra structure of $A_{n}$. In conclusion we have a strict $C_{\infty}$-map
 \[
  q\: : \: \left( \operatorname{Tot}_{N}\left(\Z^{2n}, {A'}_{n}\right), {m}_{\bullet}\right) \to \left( A_{n}, d, \wedge\right) 
  \]
  Then we set $p^{1}:=q|_{{B'}_{n}}$.
\end{proof}
\begin{prop}\label{generalrelation}
 Let $\left( {B'}_{n}, m_{\bullet}\right) $ as above.
 \begin{enumerate}
 	\item The restriction of $m_{2}$ on ${A'}_{n}$ coincides with the wedge product.
 	\item  
 	Let $x_{1}, \dots, x_{l}\in \left\lbrace w(u)_{i,j}^{(k)}, \alpha(u)_{i,j}\:|\:k\geq 0,\,i,j=0,1,\dots,n \right\rbrace $. 
 	\begin{enumerate}
 		\item If there exists at least one $x_{s}$ such that $x_{s}=w(u)_{i,j}^{(k)}$ for some $i,j,k$, we have $$m_{l}(x_{1}, \dots, x_{l})= 0$$ for $l>2$ and $l$ even.
 		\item Let $l>2$ odd. If there exist more than one $x_{s}$ such that $x_{s}=w(u)_{i,j}^{(k)}$ for some $i,j,k$, then
 		$$m_{l}(x_{1}, \dots, x_{l})= 0.$$
 		\item Let $l>2$ odd. If there exist exactly one $x_{s}$ such that $x_{s}=w(u)_{i,j}^{(k)}$ for some $i,j,k$, then
 		\[m_{l}\left(\alpha(u)_{i_{1},j_{1}}, \dots,\alpha(u)_{i_{l-1},j_{l-1}}, w(u)_{i_{l},j_{l}}^{(k)}\right)=m_{l}\left(\gamma(u)_{i_{1},j_{1}}, \dots,\gamma(u)_{i_{l-1},j_{l-1}}, w(u)_{i_{l},j_{l}}^{(k)}\right)
 		\]
 		and $m_{l}\left(\alpha(u)_{i_{1},j_{1}}, \dots,\alpha(u)_{i_{l-1},j_{l-1}}, w(u)_{i_{l},j_{l}}^{(k)}\right)(g)\in {A'}^{1}$ is given by
 		\[
 		\gamma(u)_{i_{1},j_{1}}(g)\cdots\gamma(u)_{i_{l-1},j_{l-1}}(g)\sum_{p=0}^{k} \frac{{w(u)^{(k-p)}_{i_{l},j_{l}}}\left( -\gamma(u)_{i,j}(g)\right)^{p}}{p!}
 		\]
 		for $g\in \Z^{2n}$. Moreover
 		\[
 		m_{l}\left(\alpha(u)_{i_{1},j_{1}}, \dots,\alpha(u)_{i_{l-1},j_{l-1}}, w(u)_{i_{l},j_{l}}^{(0)}\right)=0.
 		\]
 		The $m_{l+1}(w(u)^{(k)}_{i,j},\alpha(u)_{i_{1},j_{1}},\dots, \alpha(u)_{i_{l},j_{l}})$ are invariant under the permutation of the $\alpha(u)$ terms and $m_{l+1}(\alpha(u)_{i_{1},j_{1}}, \dots, \alpha(u)_{i_{r},j_{r}},w(u)^{(k)}_{i,j}, \alpha(u)_{i_{r+1},j_{r+1}} ,\dots,\alpha(u)_{i_{l},j_{l}})$ is given by $$
 		\binom{l}{r}m_{l+1}(w(u)^{(k)}_{i,j},\alpha(u)_{i_{1},j_{1}},\dots, \alpha(u)_{i_{l},j_{l}}).$$
 	\end{enumerate} 
 	\item Let $l>2$, then $m_{l}\left(\alpha(u)_{i_{1},j_{1}}, \dots,\alpha(u)_{i_{l-1},j_{l-1}}, \alpha(u)_{i_{l},j_{l}}\right)(g)$ is a form of degree $2$ with $(2,0)$-part equal $0$, $(1,1)$-part equal to $$-\sum_{r=1}^{l}\binom{l}{r}\gamma(u)_{i_{1},j_{1}}(g)\cdots d\gamma(u)_{i_{r},j_{r}}(g) \cdots \gamma(u)_{i_{l},j_{l}}(g)\in {A'}^{1}$$
 	for $g\in \Z^{2n}$ and $(2,0)$-part equal to
 	\[
 	m_{l}\left(\gamma(u)_{i_{1},j_{1}}, \cdots \gamma(u)_{i_{l},j_{l}} \right)(g_{1},g_{2})
 	\]
 	for $g_{1},g_{2}\in \Z^{2n}$.
 	\item $D\left( w(u)_{i,j}^{(n)}\right) =\sum_{l=1}^{n}(-1)^{l+1}m_{l+1}(\alpha(u)_{i,j},\dots, \alpha(u)_{i,j}, w(u)^{(n-l)}_{i,j})$, for any $n$.
 \end{enumerate}
\end{prop}
\begin{proof}
The first three points are a consequence of \cite[Theorem 2.7 and 2.9 ]{Sibilia1}. The proof of the last point follows from the proof of Lemma \ref{modelelliptic curve}.  
\end{proof}
By Proposition \ref{injbrown}, for each $k>0$, we have an inclusion $$\Psi_{*}\: : \: \Map\left(  \left( \Z^{2n}\right)^{k} ,{A'}_{n}^{\bullet}\right)\to \Map\left(  \left( \Z^{2n}\right)^{k} ,\left( A_{DR}\left(\C^{n}-\mathcal{D}\right)\otimes \Omega(1)\right)^{\bullet} \right).$$ Since $\Psi$ preserves the $\Z^{2n}$-action, the above map is simplicial and we get an inclusion
\[
\Psi_{*}\: : \: \operatorname{Tot}_{N}\left( A\right) \hookrightarrow\operatorname{Tot}_{N}\left(A_{DR}\left(\left( \C^{n}-\mathcal{D}\right)_{\bullet}\left( \Z^{2n}\right) \right)\right)  \otimes   \Omega(1).
\]
We define $H$ via the commutative diagram
\[
\begin{tikzcd}
\operatorname{Tot}_{N}\left( A\right)  \arrow{r}{\Psi_{*}}&\operatorname{Tot}_{N}\left(A_{DR}\left(\left( \C^{n}-\mathcal{D}\right)_{\bullet}\left( \Z^{2n}\right) \right)\right)\otimes\Omega(1) \\
{B'}_{n}\arrow{ur}{H}\arrow[u,hook,,]
\end{tikzcd}
\]
In particular $H$ is injective.
\begin{thm}\label{theoremmodel}
Consider $H\: : \: {B'}_{n}\to \operatorname{Tot}_{N}\left(A_{DR}\left(\left( \C^{n}-\mathcal{D}\right)_{\bullet}\left( \Z^{2n}\right) \right)\right)\otimes\Omega(1)$. Let $\tilde{W}\subset {B'}_{n}$ be the graded vector space generated by
\begin{enumerate}
\item $1$ in degree zero,
\item $w(u)^{(0)}_{i,0}$, $\alpha(u)_{i,0}$ for $i=1, \dots ,n$ in degree $1$,
\item  
\begin{align*}
m_{2}(w(u)^{(0)}_{i,0},\alpha(u)_{j}),\quad m_{2}(w(u)^{(0)}_{i,0},w(u)_{j}),\quad ,m_{2}(\alpha(u)_{i},\alpha(u)_{j}) 
\end{align*}
 and 
\begin{align*}
&(i<j,j)_{1}:=m_{2}(w(u)^{(1)}_{i,j},\alpha(u)_{j})-m_{2}(w(u)^{(1)}_{i,0},\alpha(u)_{j})-m_{2}(w(u)^{(1)}_{j,0},\alpha(u)_{i})\\
&(i<j,j)_{2}:=m_{2}(w(u)^{(1)}_{i,j},w(u)^{(0)}_{j,0})-m_{2}(w(u)^{(1)}_{i,0},w(u)^{(0)}_{j,0})-m_{2}(w(u)^{(1)}_{j,0},w(u)^{(0)}_{i,0}),
\end{align*}
for $1\leq i< j\leq n$ and by
\begin{align*}
(i<j<k)_{1}:=m_{2}(w(u)^{(1)}_{i,j},\alpha(u)_{k})+m_{2}(w(u)^{(1)}_{i,k},\alpha(u)_{j})+m_{2}(w(u)^{(1)}_{k,j},\alpha(u)_{i})\\
-m_{2}(w(u)^{(1)}_{i,0},\alpha(u)_{k})-m_{2}(w(u)^{(0)}_{i,k},\alpha(u)_{k})-m_{2}(w(u)^{(1)}_{k,0},\alpha(u)_{j})\\
-m_{2}(w(u)^{(1)}_{i,0},\alpha(u)_{j})-m_{2}(w(u)^{(0)}_{j,0},\alpha(u)_{i})-m_{2}(w(u)^{(1)}_{k,0},\alpha(u)_{i}),
\end{align*}
and
\begin{align*}
(i<j<k)_{2}:=m_{2}(w(u)^{(1)}_{i,j},w(u)^{(0)}_{k,0})+m_{2}(w(u)^{(1)}_{i,k},w(u)^{(0)}_{j,0})+m_{2}(w(u)^{(1)}_{k,j},w(u)^{(0)}_{i,0})\\
-m_{2}(w(u)^{(1)}_{i,0},w(u)^{(0)}_{k,0})-m_{2}(w(u)^{(0)}_{i,k},w(u)^{(0)}_{k,0})-m_{2}(w(u)^{(1)}_{k,0},w(u)^{(0)}_{j,0})\\
-m_{2}(w(u)^{(1)}_{i,0},w(u)^{(0)}_{j,0})-m_{2}(w(u)^{(0)}_{j,0},w(u)^{(0)}_{i,0})-m_{2}(w(u)^{(1)}_{k,0},w(u)^{(0)}_{i,0}),
\end{align*}
for $1\leq i<j< k\leq n$ in degree $2$.
\end{enumerate}
Then $H\: : \: \bar{W}\to \operatorname{Tot}_{N}\left(A_{DR}\left(\left( \C^{n}-\mathcal{D}\right)_{\bullet}\left( \Z^{2n}\right) \right)\right)\otimes\Omega(1)$ is an inclusion which is a quasi-isomorphism in degrees $0,1,2$.
\end{thm}
\begin{proof}
See Subsection \ref{sectionmodel}.
\end{proof}

\begin{defi}
We denote by $J\subset {B'}_{n}$ the $C_{\infty}$-ideal generated by all the $2$-forms
$$m_{l}\left(\alpha(u)_{i_{1},j_{1}}, \dots,\alpha(u)_{i_{l-1},j_{l-1}}, \alpha(u)_{i_{l},j_{l}}\right)$$
for $l>2$. We denote by $B_{n}:={B'}_{n}/J$ the quotient $C_{\infty}$-algebra.
\end{defi}
We denote by $J_{DR}$ the image of $J$ via the map $H$. The map $H$ induces a $C_{\infty}$-strict morphism
$$H\: : \: {B}_{n}\to \operatorname{Tot}_{N}\left(A_{DR}\left(\left( \C^{n}-\mathcal{D}\right)_{\bullet}\left( \Z^{2n}\right) \right)\right)\otimes\Omega(1)/J_{DR}.$$
Let $J_{s,DR}:=ev^{s}\left(J_{DR}\right)$. 
\begin{thm}\label{1-ext}
The diagram
\[
\begin{tikzcd}
{B}_{n}\arrow{rr}{p^{1}}\arrow{drrr}{H}& &{A}_{n}\arrow{r}{\psi^{1}}&A_{DR}\left( \C^{n}-\mathcal{D}\right)\\
& &&\left( \operatorname{Tot}_{N}A_{DR}\left(\left( \C^{n}-\mathcal{D}\right)_{\bullet}\left( \Z^{2n}\right) \right)\otimes \Omega(1)\right)/J_{DR}\arrow{u}{ev^{1}}
\end{tikzcd}
\]
commutes and $B_{n}$ is a $1$-extension for $A_{n}$.
\end{thm}
\begin{proof}
See Subsection \ref{sectionmodel}.
\end{proof}
\subsection{The $1$-minimal model}
\begin{thm}\label{vspacedect}
There exists a compatible Hodge type decomposition of $B_{n}$ 
\begin{equation}\label{vspgen}
B_{n}=W\oplus \mathcal{M}\oplus d\mathcal{M}
\end{equation}
such that\footnote{Note that we make a small abuse of notation here, we consider these elements as elements in $B_{n}$.} 
\begin{enumerate}
\item $W^{1}$ is the vector space generated $w(u)^{(0)}_{i,0}$, $\alpha(u)_{i,0}$ for $i=1, \dots, n$.
\item $\mathcal{M}^{1}$ is the vector space generated $w(u)^{(k)}_{i,j}$,  $i,j=1, \dots, n$ and $k>0$.
\item $W^{2}$ is the vector space generated for $1\leq i< j\leq n$ by
\begin{align*}
m_{2}(w(u)^{(0)}_{i,0},\alpha(u)_{j}),\quad m_{2}(w(u)^{(0)}_{i,0},w(u)_{j}),\quad m_{2}(\alpha(u)_{i},\alpha(u)_{j}) 
\end{align*}
 and 
\begin{align*}
(i<j,j)_{1},\quad (i<j,j)_{2},
\end{align*}
 and for $1\leq i<j< k\leq n$ by
\begin{align*}
(i<j<k)_{1},\quad (i<j<k)_{2}.
\end{align*}
\item $\mathcal{M}^{2}$ is the vector space generated for $1\leq i< j\leq n$
\begin{align*}
&m_{2}\left( w(u)^{(k)}_{i,j}, w(u)^{(l)}_{r,s}\right) ,\quad l,k\geq 1,\\
&m_{2}\left( w(u)^{(k)}_{i,0},w(u)^{(0)}_{j,0}\right) ,\quad k\geq 1, \\
&m_{2}\left( w(u)^{(1)}_{i,j},w(u)^{(0)}_{k,0}\right) ,\quad k<i\text{ or }k<j,\\
&m_{2}\left( w(u)^{(1)}_{i,j},w(u)^{(0)}_{i,0}\right) ,\quad i<j,  \\
&m_{2}(w(u)^{(k)}_{i,j},\alpha(u)_{r,0}),\quad k>1,i\neq r\neq j , \\
&m_{2}(w(u)^{(1)}_{i,j},\alpha(u)_{r,0}),\quad r<i\text{ or }r<j , \\
&m_{2}(w(u)^{(k)}_{i,j},\alpha(u)_{j,0}),\quad k>1,i< j,  \\
&m_{2}(w(u)^{(1)}_{i,0},\alpha(u)_{j,0}),\quad k\geq 1 , \\
&m_{l+1}(w(u)^{(k)}_{i,j},\alpha(u)_{i_{1},j_{1}},\dots, \alpha(u)_{i_{l},j_{l}}),\quad l>1,k>0 . \\
\end{align*}
\end{enumerate}
\end{thm}
\begin{proof}
See subsection \ref{sectionproofvspacevect}.
\end{proof}
\begin{rmk}\label{conj}
By a computer assisted proof we have calculated that $$m_{l}\left({\gamma}(u)_{i_{1},j_{1}}, \dots, {\gamma}(u)_{i_{l},j_{l}}\right)=0$$ for $k=3,4$. We conjecture that is true for any $k$. A consequence is that $J\subset{B'}_{n}$ doesn't contain any closed forms and hence that ${B'}_{n}$ is a $1$-model and the above decomposition is a Hodge type decomposition for ${B'}_{n}$ as well.
\end{rmk}
By Lemma \ref{generallemmadecomp} the above decomposition corresponds to a diagram of type type \eqref{homretract} such that $f\circ g=1_{W}$, $d_{W}=0$, $f\circ h=0$, $h\circ g=0$ and $h\circ h=0$. In particular $g\: : \: W\hookrightarrow B_{n}$ is the inclusion and $f\: : \: B_{n}\to W$
 is the projection. The map $h\: : \: B_{n}^{\bullet}\to B_{n}^{\bullet-1}$ is as follows: let $b\in B_{n}$, the decomposition \eqref{vspgen} implies that $b$ can be written in a unique way as $b=b_{1}+b_{2}+Db_{3}$, then
 \[
 h(b):=b_{3}\in \mathcal{M}.
 \]
 By the homotopy transfer theorem  (see Theorem \ref{markglobal}) there exists a $1-C_{\infty}$-structure $m_{\bullet}^{W}$ on $W$ and a $1-C_{\infty}$-morphism $g_{\bullet}\: : \: \left(W, m_{\bullet}^{W} \right)\to \mathcal{F}\left(B_{n}, m_{\bullet} \right)  $ such that $g_{\bullet}$ is a $1$-minimal model. The maps are given by
\[
m_{k}^{W}:=f\circ p_{k}\circ g^{\otimes k}, \quad {{g}}_{k}:= h\circ p_{k}\circ g^{\otimes k}.
\]  
where $p_{k}\: : \: B_{n}^{\otimes k}\to B_{n}$ is a family of linear maps of degree $2-k$.
Let $X_{1}, \dots, X_{n}, Y_{1}, \dots, Y_{n}$ be the basis of $\left( W^{1}_{+}[1]\right)^{*}$ dual to $$ s^{-1} \left(- w(u)^{(0)}_{1,0} \right) ,\dots ,s^{-1}\left( - w(u)^{(0)}_{n,0}\right),s^{-1} \left( -\alpha(u)_{1,0} \right) ,\dots ,s^{-1}\left(  -\alpha(u)_{n,0}\right)\in W^{1}_{+}[1].$$ We denote by
\begin{enumerate}
\item $X_{i,j}$ for $i<j$ the element dual to  $s^{-1}m_{2}(-w(u)^{(0)}_{i,0},-w(u)_{j,0})$;
\item $Y_{i,j}$ for $i<j$ the element dual to  $s^{-1}m_{2}(-\alpha(u)_{i},-\alpha(u)_{j})$;
\item $U_{i,j}$ for $i<j$ the element dual to  $s^{-1}m_{2}(-w(u)^{(0)}_{i,0},-\alpha(u)_{j})$;
\item $T_{i,j,j}$  for $i<j$ the element dual to  $-s^{-1}(i<j,j)_{1}$,
\item $Z_{i,j,j}$  for $i<j$ the element dual to  $-s^{-1}(i<j,j)_{2}$,
\item $T_{i,j,k}$  for $i<j<k$ the element dual to  $-s^{-1}(i<j<k)_{1}$,
\item $Z_{i,j,k}$  for $i<j<k$ the element dual to  $-s^{-1}(i<j<k)_{2}$.
\end{enumerate}
These elements form a basis of $\left( W[1]^{1}\right)^{*}$. 

\begin{thm}\label{gemetriczeroconn} 
Let $g_{\bullet}\: : \: \left(W, m_{\bullet}^{W} \right)\to \mathcal{F}\left(B_{n}, m_{\bullet} \right)  $ be as above.
\begin{enumerate}
\item The maps $m_{\bullet}^{W}$ correspond to a coderivation $\delta\: :\: \left( T^{c}\left( W_{+}[1]\right)\right) ^{1}\to T^{c}\left(  W^{1}_{+}[1]\right) $ such that
\begin{eqnarray}
\nonumber& &\delta^{*}X_{i,j}=\left[X_{i},X_{j} \right], \quad  \delta^{*}Y_{i,j}=\left[Y_{i}, Y_{j} \right], \quad \delta^{*}X_{i,j}=\left[X_{i},Y_{j} \right]-\left[X_{j},Y_{i} \right]\\
\label{coarserel}& &\delta^{*}T_{i,j,j}=-\left[ \left[X_{j},Y_{i} \right],Y_{i}+Y_{j}\right], \quad \delta^{*}Z_{i,j,j}=\left[ X_{j}+X_{i},\left[Y_{i},X_{j}\right] \right],\\
 \nonumber& &\delta^{*}T_{i,j,k}=2\left[Y_{k},\left[ Y_{i}, X_{j}\right] \right], \quad \delta^{*}Z_{i,j,k}=2\left[X_{k},\left[ Y_{i}, X_{j}\right] \right],
\end{eqnarray}
\item  Let $\mathcal{R}_{0}\subset \widehat{\mathbb{L}}\left( \left( W_{+}[1]^{0}\right)^{*}\right)$ be the completion of the Lie ideal generated by
\begin{eqnarray*}
\delta^{*}X_{i,j},\delta^{*}Y_{i,j},\delta^{*}X_{i,j},\delta^{*}T_{i,j,j}, \delta^{*}Z_{i,j,j},
\delta^{*}T_{i,j,k}, \delta^{*}Z_{i,j,k}.
\end{eqnarray*}
We set $\mathfrak{u}:=\left( \widehat{\mathbb{L}} \left( \left(W^{1}_{+}[1]\right) ^{*}\right)/\mathcal{R}_{0}\right)$ then
\[
\operatorname{Conv}_{1-C_{\infty}} \left( \left(W, m_{\bullet}^{W} \right), \left(B,m^{B}_{\bullet} \right)\right)\cong B_{n}\widehat{ \otimes}\mathfrak{u}
\]
and $g_{\bullet}$ corresponds to the Maurer-Cartan $C$ element in the $L_{\infty}$-algebra $B_{n}\widehat{ \otimes}\mathfrak{u}$
\begin{eqnarray*}
{C}&=&-\sum_{i}w(u)^{(0)}_{i,0}X_{i}-\alpha(u)_{i}Y_{i}-\sum_{i,k\geq 1}(-1)^{k}w(u)^{(k)}_{i,0}\left[ \cdots \left[ \left[X_{i}, Y_{i} \right], \cdots\right] , Y_{i}\right]\\
& - & \sum_{j<i,k\geq 1}(-1)^{k}\left( w(u)^{(k)}_{j,0}-w(u)^{(k)}_{j,i}\right) \left[ \cdots \left[ \left[X_{i}, Y_{j} \right], \cdots\right] , Y_{j}\right]\\
& - & \sum_{j<i,k\geq 1}(-1)^{k}\left( w(u)^{(k)}_{i,0}\right) \left[ \cdots \left[ \left[X_{i}, Y_{j} \right], \cdots\right],Y_{i}\right].
\end{eqnarray*}
\item The degree zero geometric connection associated to $g_{\bullet}$ is given by
\begin{eqnarray*}
H_{*}C&=&-\sum_{i}w(u)^{(0)}_{i,0}X_{i}-\alpha(u)_{i}Y_{i}-\sum_{i,k\geq 1}w(u)^{(k)}_{i,0}\operatorname{Ad}^{(k)}_{Y_{i}}\left(X_{i} \right)\\
& - & \sum_{j<i,k\geq 1}\left( w(u)^{(k)}_{j,0}+w(u)^{(k)}_{0,i}-w(u)^{(k)}_{j,i}\right) \operatorname{Ad}^{(k)}_{Y_{j}}\left(X_{i} \right).
\end{eqnarray*}
\end{enumerate}
\end{thm}
\begin{proof}
In Subsection \ref{section p kernel} we give a formula for
\[
p_{k}\circ g^{\otimes k}|_{\left( W^{1}\right)^{\otimes k} }
\]
for any $k>0$. By \cite[Proposition 1.20]{Sibilia1} the founded maps can be turned into Coalgebra settings, this proves 1 and 2. We prove 3. The relation \eqref{coarserel} implies
\begin{equation}\label{shiftrel}
\operatorname{Ad}^{(k)}_{Y_{i}}(X_{j})=(-1)^{k+1}\operatorname{Ad}^{(k)}_{Y_{j}}(X_{i})
\end{equation}
in $\mathfrak{u}$.
\end{proof}
For $0\leq s\leq 1$, we set $$C(s):=\left( ev_{s}H\right)_{*}\left(C \right) \in \left( \operatorname{Tot}^{1}_{N}A_{DR}\left(\left( \C^{n}-\mathcal{D}\right)_{\bullet}\left( \Z^{2n}\right) \right)/J_{s,DR}\right) \widehat{\otimes} \mathfrak{u}$$ 
 We denote by $w(s)_{i,j}^{(k)}$, $\nu(s)_{i,j}$ and $\gamma_{i,j}(s)$ resp. the elements $ev^{s}\left(w_{i,j}(u)^{(k)}\right) $, $ev^{s}\left(\nu(u)_{i,j}^{(k)}\right) $ and  $ev^{s}\left(\gamma(u)_{i,j}\right) $ resp. for some $s\in [0,1]$. Hence $C(s)$ is given by
\begin{eqnarray*}
C(s)&=-&\sum_{i}w(s)^{(0)}_{i,0}X_{i}-\alpha(s)_{i}Y_{i}-\sum_{i,k\geq 1}w(s)^{(k)}_{i,0}\operatorname{Ad}^{(k)}_{Y_{i}}\left(X_{i} \right)\\
& - & \sum_{j<i,k\geq 1}\left( w(s)^{(k)}_{j,0}+w(s)^{(k)}_{0,i}-w(s)^{(k)}_{j,i}\right) \operatorname{Ad}^{(k)}_{Y_{j}}\left(X_{i} \right)
\end{eqnarray*}
Notice that $\alpha(s)_{i}=\gamma(s)_{i}-\nu(s)_{i}$.
\begin{thm}\label{partialfinal}
Consider $ \mathfrak{u}$ equipped with the action $\operatorname{ad}$. For each $0\leq s\leq 1$,\\ $r_{*}C(s)\in A_{DR}\left( \C^{n}-\mathcal{D}\right) \widehat{\otimes} \mathfrak{u} $ is given by 
\begin{eqnarray*}
r_{*}C(s) &=-&\sum_{i}w(s)^{(0)}_{i,0}X_{i}+\nu(s)_{i}Y_{i}-\sum_{i,k\geq 1}w(s)^{(k)}_{i,0}\operatorname{Ad}^{(k)}_{Y_{i}}\left(X_{i} \right) \\
& - & \sum_{j<i,k\geq 1}\left( w(s)^{(k)}_{j,0}+w(s)^{(k)}_{0,i}-w(s)^{(k)}_{j,i}\right)\operatorname{Ad}^{(k)}_{Y_{i}}\left(X_{j} \right).
\end{eqnarray*}
It is a flat connection on $\C^{n}-\mathcal{D}$ on the trivial bundle with fiber $ \mathfrak{u}$. In particular, for $s=0$ the connection is holomorphic with logarithmic singularities. Moreover $r_*C$ induces a gauge-equivalence between $r_*C(1)$ and $r_*C(s)$, where the gauge is given by $\sum_{i} 2 \pi i r_{i}(1-s)Y_{i}\in A_{DR}^{0}\left( \C^{n}-\mathcal{D}\right) \widehat{\otimes} \mathfrak{u}$.
\end{thm}
\begin{proof}
 the map $(r ev_{s}H)$ preserves Maurer-Cartan elements. Hence the flatness corresponds to the Maurer-Cartan equations (see \cite[Section 3]{Sibilia1}).
\end{proof}
\subsection{The KZB connection}
In \cite{Damien} a meromorphic flat connection $\omega_{KZB,n} $ on the configuration space of the punctured elliptic curve with value in the bundle $\overline{\mathcal{P}}^{n+1}$ is constructed. We show that $r_{*}C(0) $ corresponds to $\omega_{KZB,n}$. We will use the same notation of \cite{Damien}.
For $n\geq 0$, we define the algebra $\mathfrak{t}_{1,n}$ as the free Lie algebra with generators $X_{1}, \dots, X_{n}, Y_{1}, \dots , Y_{n}$ and $t_{i,j}$ for $1\leq i\neq j\leq n$ modulo
\begin{eqnarray}
\label{relLieDamien1}t_{ij}=t_{ij}, \quad \left[ t_{ij},t_{ik}+t_{jk}\right]=0, \quad  \left[ t_{ij},t_{kl}\right]=0\\
\nonumber t_{ij}= \left[ X_{i},Y_{j}\right], \quad  \left[ X_{i},X_{j}\right]=\left[ Y_{i},Y_{j}\right]=0, \quad \left[X_{i}, Y_{i} \right]=-\sum_{j|j\neq i}t_{ij}  \\
\nonumber \left[ X_{i},t_{jk}\right]=\left[Y_{j},t_{ik} \right]=0, \quad \left[ X_{i}+X_{j},t_{jk}\right]=\left[Y_{i}+Y_{j},t_{ik} \right]=0  
\end{eqnarray}
for $i,j,k,l$ distinct. 
\begin{rmk}
Notice that the the relations $\left[ t_{ij},t_{ik}+t_{jk}\right]=0$, and $ \left[ t_{ij},t_{kl}\right]=0$ follow from $ \left[ x_{i},t_{jk}\right]=\left[Y_{j},t_{ik} \right]=0$,$ \left[ x_{i}+x_{j},t_{jk}\right]=\left[Y_{i}+Y_{j},t_{ik} \right]=0$ and the Jacobi identity. 
\end{rmk} 
The elements $\sum_{i}X_{i}$ and $\sum_{i}Y_{i}$ are central in $\mathfrak{t}_{1,n}$. We denote by $\overline{\mathfrak{t}}_{1,n}$ the quotient of $\mathfrak{t}_{1,n}$ modulo 
\begin{eqnarray}
\label{relLieDamien2}
\sum_{i}X_{i}=\sum_{i}Y_{i}=0
\end{eqnarray}

\begin{prop}\label{equality}
The lie algebra $\overline{t}_{1,n+1}$ admits the following presentation: the generators are $X_{1}, \dots, X_{n}$ and $Y_{1}, \dots, Y_{n}$ and the relations are
\begin{eqnarray}\label{relationLie}
\left[X_{i},X_{j} \right]=\left[Y_{i}, Y_{j} \right]=0, \quad \left[X_{i},Y_{j} \right]-\left[X_{j},Y_{i} \right]=0,\,i<j\\\nonumber
\left[ \left[X_{j},Y_{i} \right],Y_{i}+Y_{j}\right]=\left[ X_{j}+X_{i},\left[Y_{i},X_{j}\right] \right]=0,\,i<j\\\nonumber
\left[Y_{k},\left[ Y_{i}, X_{j}\right] \right]=\left[X_{k},\left[ Y_{i}, X_{j}\right] \right]=0,\,i<j<k.
\end{eqnarray}
In particular, $\mathfrak{u}=\widehat{\overline{t}}_{1,n+1}$.
\end{prop}
\begin{proof}
Let $\tilde{L}_{n}$ be the free Lie algebra on with generators $\tilde{X}_{1}, \dots, \tilde{X}_{n}, \tilde{Y}_{1}, \dots , \tilde{Y}_{n}$ and $\tilde{t}_{i,j}$ for $1\leq i\neq j\leq n$ modulo relations 
\begin{eqnarray}
\label{relLieDamien3}\tilde{t}_{ij}=\tilde{t}_{ij}\\
\nonumber \tilde{t}_{ij}= \left[ \tilde{X}_{i},\tilde{Y}_{j}\right], \quad  \left[ \tilde{X}_{i},\tilde{X}_{j}\right]=\left[ \tilde{Y}_{i},\tilde{Y}_{j}\right]=0, \\
\nonumber \left[ \tilde{X}_{i},\tilde{t}_{jk}\right]=\left[\tilde{Y}_{j},\tilde{t}_{ik} \right]=0, \quad \left[ \tilde{X}_{i}+\tilde{X}_{j},\tilde{t}_{jk}\right]=\left[\tilde{Y}_{i}+\tilde{Y}_{j},\tilde{t}_{ik} \right]=0.  
\end{eqnarray} 
for $i,j,k,l$ distinct. We first show that the map $\tilde{h}\: : \: \tilde{L}_{n}\to \overline{\mathfrak{t}}_{1,n+1}$ defined by $\tilde{h}(\tilde{X}_{i})=X_{i}$, $\tilde{h}(\tilde{Y}_{i})=Y_{i}$ is an isomorphism of Lie algebras. The map is clearly well-defined. We define an inverse via $\tilde{h}^{-1}({X}_{i}):=\tilde{X}_{i}$, $\tilde{h}^{-1}({Y}_{i}):=\tilde{Y}_{i}$ for $i<n+1$ and with $\tilde{h}^{-1}({X}_{n+1}):=-\sum_{i=1}^{n}\tilde{X}_{i}$, $\tilde{h}^{-1}({Y}_{n+1}):=-\sum_{i=1}^{n}\tilde{Y}_{i}$. In order to prove that $\tilde{h}^{-1}$ is well-defined, we have to check that $\tilde{h}^{-1}$ sends the relations \eqref{relLieDamien1} and \eqref{relLieDamien2} into  \eqref{relLieDamien3}. This is immediate if we consider distinct index $i,j,k,l$ smaller that $n+1$. It is also immediate to show that $\tilde{h}^{-1}(t_{ij}-t_{ij})=\tilde{h}^{-1}(t_{ij}-t_{ij})=\tilde{h}^{-1}(\left[ X_{i},X_{j}\right])=\tilde{h}^{-1}(\left[ Y_{i},Y_{j}\right])=0$ if one of the index is equal to $n+1$.
On the other hand, consider the cubic relation $\left[ X_{n+1},t_{jk}\right]=0$. 
We have
\begin{eqnarray*}
\tilde{h}^{-1}\left(\left[ X_{n+1},t_{jk}\right]  \right)& : = & \left[-\sum_{i=1}^{n}\tilde{X}_{i},\tilde{t}_{jk} \right]\\
& =&\displaystyle\sum_{\substack{i=1 \\ i\neq j,k}}^{n}\left[\tilde{X}_{i},\tilde{t}_{jk} \right]-\left[\tilde{X}_{j}+\tilde{X}_{k},\tilde{t}_{jk} \right]=0
\end{eqnarray*}
Similarly, we have $\tilde{h}^{-1}\left(\left[ X_{i},t_{(n+1)k}\right]  \right)=\tilde{h}^{-1}\left(\left[ Y_{n+1},t_{jk}\right]  \right)=\tilde{h}^{-1}\left(\left[ Y_{i},t_{(n+1)k}\right]  \right)=0$. We have
\begin{eqnarray*}
\tilde{h}^{-1}\left(\left[ \tilde{X}_{(n+1)}+\tilde{X}_{j},t_{(n+1)j}\right]  \right)& : = &\left[ -\sum_{i=1}^{n}\tilde{X}_{i}+\tilde{X}_{j},\left[-\sum_{i=1}^{n}\tilde{X}_{i},\tilde{Y}_{j} \right] \right] \left[-\sum_{i=1}^{n}\tilde{X}_{i},\tilde{t}_{jk} \right]\\
& =&\left[\displaystyle\sum_{\substack{i=1 \\ i\neq j}}^{n}\tilde{X}_{i},\displaystyle\sum_{\substack{l=1 \\ l\neq j}}^{n}\left[\tilde{X}_{l},\tilde{Y}_{j} \right]+\left[ \tilde{X}_{j},\tilde{Y}_{j}\right]\right] \\
& =&\displaystyle\sum_{\substack{i,l=1 \\ i\neq j\neq l,i\neq l}}^{n}\left[\tilde{X}_{i},\left[\tilde{X}_{l},\tilde{Y}_{j} \right]\right]+\displaystyle\sum_{\substack{i=1 \\ i\neq j}}^{n}\left[\tilde{X}_{i},\left[\tilde{X}_{i},\tilde{Y}_{j} \right]  \right]  +\left[\tilde{X}_{i},\left[\tilde{X}_{j},\tilde{Y}_{j} \right]  \right] =0  
\end{eqnarray*}
since the first summand is zero and by \eqref{relLieDamien3} we have $\left[\tilde{X}_{i},\left[\tilde{X}_{j},\tilde{Y}_{j} \right]  \right]=\left[\tilde{X}_{j},\left[\tilde{X}_{i},\tilde{Y}_{j} \right]  \right]$, i.e the second summand is zero as well. The same arguments work for the rest of the cubic relations. This shows that $\tilde{h}$ is an isomorphism of Lie algebras. 
We define the map $\tilde{h'}\: : \: \overline{t}_{1,n+1}\to \tilde{L}_{n}$ via $\tilde{h'}(X_{i})=\tilde{X}_{i}$, $\tilde{h'}(Y_{i})=\tilde{Y}_{i}$ for $i=1, \dots,n$. The Jacobi identity and the relations $[X_{i}, X_{j}]=[Y_{i}, Y_{j}]=0$ for $i\neq j$ allow us to extend the relation \eqref{relationLie} for unordered indices. This shows that $\tilde{h'}$ is an isomorphism and so is $\tilde{h}\tilde{h'}$.
\end{proof}
Let $\tau\in \mathbb{H}$ be as above. We define $\mathcal{D}\subset \C^{n+1}$
\[
\mathcal{D}:=\left\lbrace (\xi_{1}, \dots \xi_{n})\: : \: \xi_{i}-\xi_{j}\in \mathbb{Z}+\tau\mathbb{Z} \text{ for some distinct }i,j=1, \dots, n+1  \right\rbrace. 
\]
 We define an action of $(\C,+)$ on $\C^{n+1}-\mathcal{D}$ via $z(\xi_{1}, \dots, \xi_{n+1}):=(\xi_{1}-z, \dots, \xi_{n+1}-z)$. This induces an action of $\mathcal{E}$ on $\operatorname{Conf}_{n+1}(\mathcal{E})$ via $\xi'(\xi_{1}, \dots, \xi_{n+1}):=(\xi_{1}-\xi', \dots, \xi_{n+1}-\xi')$. We get a projection $\pi_{1}\: : \: \C^{n+1}-\mathcal{D}\to \left( \C^{n+1}-\mathcal{D}\right)/\C$ defined  via $\pi_{1}(\xi_{1}, \dots, \xi_{n+1})=\xi_{n+1}\left( \xi_{1}, \dots, \xi_{n}\right)$ which induces $\pi_{2}\: : \: \operatorname{Conf}_{n+1}(\mathcal{E})\to\operatorname{Conf}_{n+1}(\mathcal{E})/\mathcal{E}$. We fix a section
 $h_{1}\: : \:\left( \C^{n+1}-\mathcal{D}\right)/\C\to\left( \C^{n+1}-\mathcal{D}\right)$ which sends $\left[ \xi_{1}, \dots, \xi_{n}\right]$ to $(\xi_{1}, \dots,\xi_{n}, 0)$, this induces also a section $h_{2}\: : \: \operatorname{Conf}_{n+1}(\mathcal{E})/\mathcal{E}\to\operatorname{Conf}_{n+1}(\mathcal{E})$.
There is an isomorphism $\chi_{1}\: : \: \C^{n}-\mathcal{D}\to\left( \C^{n+1}-\mathcal{D}\right)/\C $ given by $\chi_{1}\left( \xi_{1}, \dots, \xi_{n}\right)=\left[\xi_{1}, \dots, \xi_{n},0 \right]$. Its inverse is $\chi^{-1}_{1}\left[\xi_{1}, \dots, \xi_{n},\xi_{n+1} \right]=(\xi_{1}-\xi_{n+1}, \dots, \xi_{n}-\xi_{n+1})$. In particular such an isomorphism induces another isomorphism $\chi_{2}\: : \:\operatorname{Conf}_{n}(\mathcal{E^{\times}})\to\operatorname{Conf}_{n+1}(\mathcal{E})/\mathcal{E}$. We define smooth functions on $\left( \C-\left\lbrace \Z+\tau \Z\right\rbrace \right) \times [0,1]$ $f(u)^{(k)}_{i,j}$ via $ w(u)^{(k)}_{i,j}=f(u)^{(k)}_{i,j}d(z_{i}-z_{j})$. We fix an integer $n$. For $0\leq i,j\leq n+1$, we define
\[
k(u)_{ij}:=\sum_{k}f(u)^{(k)}_{i,j}\operatorname{Ad}^{(k)}_{Y_{i}}(X_{j})\in \left( A^{0}_{DR}(\C^{n+1}-\mathcal{D})\otimes\Omega^{0}(1)\right) \widehat{\otimes }\widehat{\mathfrak{t}}_{1,n+1}.
\]
We define
 \[
 K(u)_{i}:=-X_{i}+\displaystyle\sum_{\substack{j=1 \\ j\neq i}}k(u)_{ij}
 \]
and 
\[
\varpi(u):=\sum_{i=1}^{n+1}K(u)_{i}d\xi_{i}\in A^{1}_{DR}(\C^{n+1}-\mathcal{D})\otimes\Omega^{0}(1)\widehat{\otimes }\widehat{\mathfrak{t}}_{1,n+1}.
\]
For $0\leq s\leq 1$ we define the bundle $\mathcal{P}^{n}_{s}$ with fiber $ \widehat{\mathfrak{t}}_{1,n}$ on $\operatorname{Conf}_{n}\left( \mathcal{E}\right) $ via the following equation (see \cite{Levrac}, \cite{Damien}): each section $f$ of $\mathcal{P}^{n}_{s}$ satisfies
\begin{align*}
f(\xi_{1}, \cdots , \xi_{j} + l, \cdots ,\xi_{n})& =  f(\xi_{1}, \cdots , \xi_{n}), \\
f(\xi_{1}, \cdots , \xi_{j} + l\tau, \cdots ,\xi_{n})& =\exp(-2 \pi il \left(1-s \right) Y_{j})\cdot f(\xi_{1}, \cdots , \xi_{n})
\end{align*}
for any integer $l$, where $y_{j}^{k}\cdot a:=\operatorname{Ad}^{k}_{y_{j}}(a)$ for $a\in\widehat{\mathfrak{t}}_{1,n}$. For $0\leq s\leq 1$ we define the connection form $$\varpi(s)\in A^{1}_{DR}(\C^{n+1}-\mathcal{D})\widehat{\otimes }\widehat{\mathfrak{t}}_{1,n+1}$$ as the evaluation of $\varpi(u)$ at $s$. For $0\leq s\leq 1$ we define the bundle $\tilde{\mathcal{P}}^{n}_{s}$ with fiber $ \widehat{\overline{\mathfrak{t}}}_{1,n}$ on $\operatorname{Conf}_{n}\left( \mathcal{E}\right) $ as the fiber quotient of $\mathcal{P}^{n}_{s}$ via the relation \eqref{relLieDamien3}. We denote by $\overline{\mathcal{P}}^{n}_{s}$ the pullback of $\tilde{\mathcal{P}}^{n}_{s}$ along $h\chi_{2}$ and by $\tilde{\varpi}(s)\in  A^{1}_{DR}\left( \C^{n+1}-\mathcal{D}\right) \widehat{\otimes }\widehat{\overline{\mathfrak{t}}}_{1,n+1}$ the image of $\varpi(s)$ via the quotient map $A^{1}_{DR}(\C^{n+1}-\mathcal{D})\widehat{\otimes }\widehat{\mathfrak{t}}_{1,n+1}\to A^{1}_{DR}(\C^{n+1}-\mathcal{D})\widehat{\otimes }\widehat{\overline{\mathfrak{t}}}_{1,n+1}$. Consider the linear map
\[
\left(h_{2}\chi_{2} \right)^{*}\: : \:A^{1}_{DR}(\C^{n+1}-\mathcal{D})\widehat{\otimes }\widehat{\overline{\mathfrak{t}}}_{1,n+1}\to A^{1}_{DR}(\C^{n}-\mathcal{D})\widehat{\otimes }\widehat{\overline{\mathfrak{t}}}_{1,n+1}
\]
We define
\[
\overline{\varpi}(s):= \left(h_{2}\chi_{2} \right)^{*} \tilde{\varpi}(s)+\sum\nu(s)_{i}Y_{i}\in A^{1}_{DR}(\C^{n}-\mathcal{D})\widehat{\otimes }\widehat{\overline{\mathfrak{t}}}_{1,n+1},
\]
\begin{defi}[\cite{Damien}]
 The KZB connection is the holomorphic flat connection on the configuration space of the punctured torus  $\omega_{KZB,n}^{\tau}:=\overline{\varpi}(0)$ on the holomorphic bundle $\overline{\mathcal{P}}^{n}_{0}$.
\end{defi} 
Let $r_{*}C(s)$ be the connection obtained in Theorem \ref{partialfinal}. There is a gauge equivalence between 
 $r_*C(1)$ and $r_*C(s)$ given by $\sum_{i} 2 \pi i r_{i}(1-s)Y_{i}\in A_{DR}^{0}\left( \C^{n}-\mathcal{D}\right) \widehat{\otimes} \mathfrak{u}$. By Theorem \ref{thmbundleholonomy}, this gives a factor of automorphy $F^{s}$ and hence a bundle $E_{F^{s}}$ such that $(d-r_*C(s), E_{F^{s}})$ is a flat connection on $\operatorname{Conf}_{n}\left(\mathcal{E}_{\tau}^{\times} \right)$.
\begin{thm}\label{final}
Let $r_{*}C(s)$ be the connection obtained in Theorem \ref{partialfinal}.
We have  $$\left( d-\overline{\varpi}(s),\overline{\mathcal{P}}^{n+1}_{s} \right) =\left( d-r_{*}C(s), E_{F^{s}}\right) $$ and for each $0\leq s\leq 1$, $\overline{\varpi}(s)$ defines a flat connection form on the bundle $\overline{\mathcal{P}}^{n+1}_{s}$ such that its monodromy representation is the Malcev completion of $\pi_{1}\left(\operatorname{Conf}_{n}\left(\mathcal{E}_{\tau}^{\times} \right) \right) $. 
In particular
$$
\left( d-\omega_{KZB,n}^{\tau},\overline{\mathcal{P}}^{n+1}_{0} \right) =\left( d-r_{*}C(0), E_{F^{0}}\right).
$$ 
\end{thm}
\begin{proof}
By Proposition \ref{equality}, $\mathfrak{u}=\widehat{\overline{t}}_{1,n+1}$. Let $g_{\bullet}\: : \: \left(W, m_{\bullet}^{W} \right)\to \mathcal{F}\left(B_{n}, m_{\bullet} \right)  $ be as in Theorem \ref{partialfinal}. Since $A_{n}$ is a model and the Hodge decomposition on $B_{n}$ is compatible, we can restrict the decomposition on $A_{n}$ and by using the homotopy transfer theorem we can extend $g_{\bullet}p_{1}$ into a minimal model for $A_{DR}\left(\operatorname{Conf}_{n}\left(\mathcal{E}_{\tau}^{\times} \right) \right) $. By Theorem \ref{partialfinal} for each $0\leq s\leq 1$,  $r_{*}C(s)$ is gauge equivalent to $r_{*}C(1)$ via gauge $h_s(\xi_{1}, \dots, \xi_{n}):=\sum_{i} 2 \pi i (1-s)r_{i}Y_{i}\in A_{DR}^{0}\left( \C^{n}-\mathcal{D}\right) \widehat{\otimes} \overline{t}_{1,n+1}$. Consider the factor of automorphy 
$$F^{s}_{g}(\xi_{1}, \dots, \xi_{n})):=e^{-h_s\left( g(\xi_{1}, \dots, \xi_{n})\right) }e^{h_s\left( \xi_{1}, \dots, \xi_{n}\right) }$$ where $h_s(\xi_{1}, \dots, \xi_{n}):=\sum_{i} 2 \pi i sr_{i}Y_{i}\in A_{DR}^{0}\left( \C^{n}-\mathcal{D}\right) \widehat{\otimes} \overline{t}_{1,n+1}$. For a $g\in \Z^{2n}$ we have that
\[
g((\xi_{1}, \dots, \xi_{n}), X):=\left(g(\xi_{1}, \dots, \xi_{n}),F^s_{g}(\xi_{1}, \dots,\xi_{n})X \right) 
\]
defines a $\Z^{2n}$-action on $\C^{n}-\mathcal{D}\times \overline{t}_{1,n+1}$ and $\overline{\mathcal{P}}^{n+1}_{s}=E_{F^{s}}$. By Theorem \ref{thmbundleholonomy} for each $0\leq s\leq 1$, $\left(d- r_{*}C(s), E_{F^{s}}\right) $ is a flat connection such that  its monodromy representation is the Malcev completion of $\pi_{1}\left(\operatorname{Conf}_{n}\left(\mathcal{E}_{\tau}^{\times} \right) \right) $. We show the equality between the connections. Since $X_{n+1}=-X_{1}-\dots -X_{n}$, we have
\begin{align*}
\left(h_{2}\chi_{2} \right)^{*}\tilde{\varpi}(s)& = \sum_{i=1}^{n}K(s)_{i}d\xi_{i}\\
& = \sum_{i=1}^{n}\left(-X_{i}+\sum_{j\neq i} k(s)_{ij} \right) d\xi_{i}\\
& = \sum_{i=1}^{n} \big(-X_{i}+\sum_{k\geq 1}f(s)_{i,n+1}^{(k)}\operatorname{Ad}^{(k)}_{Y_{i}}(-X_{i})d\xi_{i}\\
&+\sum_{j,\,k\geq 1}f(s)_{i,j}^{(k)}\operatorname{Ad}^{(k)}_{Y_{i}}(X_{j})d\xi_{i}-f(s)_{i,n+1}^{(k)}\operatorname{Ad}^{(k)}_{Y_{i}}(X_{j})d\xi_{i}\big)\\
&= \sum_{i=1}^{n} \left( -X_{i}-\sum_{k\geq 1}f(s)_{i,n+1}^{(k)}\operatorname{Ad}^{(k)}_{Y_{i}}(X_{i})d\xi_{i}\right) \\
& + \sum_{(i>j),\, k\geq 1}\big(f(s)_{i,j}^{(k)}d\xi_{i}\operatorname{Ad}^{(k)}_{Y_{i}}(X_{j})+ f(s)_{j,i}^{(k)}d\xi_{j}\operatorname{Ad}^{(k)}_{Y_{j}}(X_{i})\\
&-f(s)_{i,n+1}^{(k)}d\xi_{i}\operatorname{Ad}^{(k)}_{Y_{i}}(X_{j})-f(s)_{j,n+1}^{(k)}d\xi_{j}\operatorname{Ad}^{(k)}_{Y_{j}}(X_{i})\big).\\
\end{align*}
We have
\[
-f(s)_{i,n+1}^{(k)}d\xi_{j}\operatorname{Ad}^{(k)}_{Y_{i}}(X_{j})=f(s)_{0,i}^{(k)}\operatorname{Ad}^{(k)}_{Y_{j}}(X_{i})=-w(s)_{0,i}^{(k)}\operatorname{Ad}^{(k)}_{Y_{j}}(X_{i})
\]
and 
\begin{align*}
f(s)_{i,j}^{(k)}d\xi_{i}\operatorname{Ad}^{(k)}_{Y_{i}}(X_{j})+ f(s)_{j,i}^{(k)}d\xi_{j}\operatorname{Ad}^{(k)}_{Y_{j}}(X_{i})=& \left( -f(s)_{j,i}^{(k)}d\xi_{i}+ f(s)_{j,i}^{(k)}d\xi_{j}\right) \operatorname{Ad}^{(k)}_{Y_{j}}(X_{i}).
\end{align*}
and we conclude
\begin{align*}
\left(h_{2}\chi_{2} \right)^{*}\tilde{\varpi}(s) =& -\sum_{i}w(s)^{(0)}_{i,0}X_{i}-\sum_{i,k\geq 1}w(s)^{(k)}_{i,0}\operatorname{Ad}^{(k)}_{Y_{i}}\left(X_{i} \right) \\
& -  \sum_{j<i,k\geq 1}\left( w(s)^{(k)}_{j,0}+w(s)^{(k)}_{0,i}-w(s)^{(k)}_{j,i}\right)\operatorname{Ad}^{(k)}_{Y_{i}}\left(X_{j} \right).
\end{align*}
\end{proof}
The fact that the monodromy representation of $d-\omega_{KZB,n}^{\tau}$ is the Malcev completion of $\pi_{1}\left(\operatorname{Conf}_{n}\left(\mathcal{E}_{\tau}^{\times} \right) \right) $ is proved in \cite{Damien}. 
\subsection{KZB and KZ connection}\label{KZKZB}
 In the previous chapter (see Proposition \ref{thmfinale0}), we give an interpretation of that in terms of $C_{\infty}$-morphism. In this section, we prove the same facts for $\omega_{KZB,n}^{\tau}$ as well: $\lim_{\tau\to i\infty }\omega_{KZB,n}^{\tau}$ is equal to $\omega_{KZ,n}$ modulo a morphism of Lie algebra $Q^{*}\: : \: \widehat{\mathfrak{t}}_{n}\widehat{\otimes} \Q(2\pi i)\to \widehat{\overline{\mathfrak{t}}}_{1,n+1}\widehat{\otimes} \Q(2 \pi i)$ of complete Lie algebra. We use the argument of Subsection  \eqref{Morphism of homological pairs} to show that $Q^{*}$ is induced by a strict $C_{\infty}$-morphism $p_{\bullet}$. Let $\tau\in\mathbb{H}$ be fixed as above. We set $q:=\exp(2 \pi i \tau)$.
We define the action of $\Z^{n}$ on $\left( \C^{*}\right)^{n}$ via
\begin{equation}\label{Zaction1}
(m_{1}, \dots , m_{n}) \cdot \left( z_{1}, \dots, z_{n}\right) := \left( q^{m_{1}}z_{1}, \dots, q^{m_{n}}z_{n}\right) 
\end{equation}
where $\left( z_{1}, \dots, z_{n}\right)$ are coordinate on $\left( \C^{*}\right)^{n}$. We set $z_{0}:=1$.  We define a map $\boldmath{e}\: : \: \C^{n}\to \left( \C^{*}\right)^{n} $ via $\boldmath{e}\left( \xi_{1}, \dots, \xi_{n}\right):=\left(\exp(2 \pi i \xi_{1} ), \dots, \exp(2 \pi i \xi_{n} ) \right)$. This map extend ot a simplicial map $\boldmath{e}_{\bullet}\: : \: \C^{n}_{\bullet}\Z^{2n}\to \left( \C^{*}\right)^{n}_{\bullet}\Z^{n}$ between the two action groupoid 
\begin{align*}
\boldmath{e}_{0}\left( \xi_{1}, \dots, \xi_{n}\right):=&\left(\exp(2 \pi i \xi_{1} ), \dots, \exp(2 \pi i \xi_{n} ) \right),\\ \boldmath{e}_{1}\left( \left( \xi_{1}, \dots, \xi_{n}\right),\left( \left( l_{1},m_{1}\right) , \dots, \left( l_{n},m_{n}\right) \right)\right) :=&\left(\exp(2 \pi i \xi_{1} ), \dots, \exp(2 \pi i \xi_{n} ), \left( m_{1} , \dots, m_{n}\right) \right)
\end{align*}
and it induces an isomorphism on the quotient. Let $\mathcal{D}\subset \C^{n}$ be the divisor defined above. Then $e(\mathcal{D})$ is the normal crossing divisor
\[
\left\lbrace  \left( z_{1}, \dots, z_{n}\right)\: |\: z_{i}\neq q^{\Z}z_{j}\text{ for }0\leq i<j\leq n  \right\rbrace \subset \left( \C^{*}\right)^{n}.
\]
It is clearly preserved by the action of $\Z^{n}$ and hence  the restriction gives rise to a morphism of simplicial manifolds with simplicial normal crossing divisor $$\boldmath{e}_{\bullet}\: : \: \left( \C^{n}-\mathcal{D}\right)_{\bullet} \Z^{2n}\to \left( \left( \C^{*}\right)^{n}-e(\mathcal{D})\right) _{\bullet}\Z^{2n}$$ 
By \eqref{Fourier} and we have
\begin{align}
\nonumber \phi_{i,j}^{(0)}& =\frac{1}{2\pi i}\frac{dz_{i}}{z_{i}}-\frac{1}{2\pi i}\frac{dz_{i}}{z_{i}},\\
\nonumber \phi^{(1)}_{i,j}& =\frac{1}{2\pi i} \left( \pi i  +\frac{2 \pi iz_{j} }{z_{i}-z_{j}}-\left( 2 \pi i \right)^{2}\sum_{n=1}^{\infty}\sum_{n|d}d\left( \left(\frac{z_{i}}{z_{j}} \right)^\frac{n}{d}-\left(\frac{z_{i}}{z_{j}}\right) ^\frac {-n}{d}\right)q^{n}\right)\left( \frac{dz_{i}}{z_{i}}-\frac{dz_{i}}{z_{i}}\right)  ,\\
\label{functionformal1} \phi^{(l)}_{i,j}& = -\frac{1}{2\pi i}\left( \frac{\left( 2 \pi i \right)^{l+1}}{l!}\left(  \sum_{n=1}^{\infty}\left( \sum_{n|d}d^{l}\left( \left(\frac{z_{i}}{z_{j}} \right)^\frac{n}{d}+(-1)^{l}\left(\frac{z_{i}}{z_{j}} \right)^\frac {-n}{d}\right)\right) q^{n} +\frac{B_{l}}{2 \pi i}\right)\right)\left( \frac{dz_{i}}{z_{i}}-\frac{dz_{i}}{z_{i}}\right)   
\end{align}
for $l>1$, where $B_{l}$ are the Bernoulli numbers. 
\begin{lem}
Let $D_{0}, D_{1}, D_{d}\subset \C^{n}$ be defined as follows: $D_{0}$ is the set of points $\left( z_{1}, \dots, z_{n}\right)$ such that $z_{i}=0$ for some $i$; $D_{1}$ is the set of points $\left( z_{1}, \dots, z_{n}\right)$ such that $z_{i}=1$ for some $i$; $D_{d}$ is the set of points $\left( z_{1}, \dots, z_{n}\right)$ such that $z_{i}=z_{j}$ for some $i\neq j$. The forms $\phi^{(l)}_{i,j}$ for $l\geq 0$ can be written as power series on $q$ where the coefficients are $1$-forms of the form $fdz_{i}$ for some $i$, where $f$ is rational function on $\C^{n}$ of the form $\frac{p_{1}}{p_{2}}$, where $p_{j}$ are polynomials over the field $\Q(2 \pi i)$ for $j=1,2$. Moreover $f$ has only poles of order $1$ located in the normal crossing divisor $\underline{\mathcal{D}}:=D_{0}\cup D_{1}\cup D_{d}\subset \C^{n}$.
\end{lem}
Given a subfield $\Q\subset \Bbbk\subset \C$, consider a normal crossing divisor $\mathcal{D}'\subset \C^{n}$, we denote by $\operatorname{Rat}^{0}_{\Bbbk}( \C^{n},\mathcal{D}')$ the algebra of rational functions $\frac{p_{1}}{p_{2}}$ with poles along $\mathcal{D}$ such that $p_{1}, p_{2}$ are polynomials over the field $\Bbbk$. We denote by $\operatorname{Rat}^{\bullet}_{\Bbbk}( \C^{n},\mathcal{D}')$ the differential graded $\Bbbk$- subalgebra of differential forms generated by forms of type $f dx_{I}$, with $f\in\operatorname{Rat}^{0}_{\Bbbk}(\C^{n},\mathcal{D}')$. In particular $\operatorname{Rat}^{\bullet}_{\Bbbk}( \C^{n},\mathcal{D}')\otimes \C\subset A_{DR}^{*}( \C^{n}-\mathcal{D}')$.\\
We consider the differential graded algebra $\operatorname{Rat}^{\bullet}_{\Q(2 \pi i)}(\C^{n},\underline{\mathcal{D}})$ and we now assume that $q$ is a formal variable of degree zero. In particular, notice that the function $\phi^{(k)}_{i,j}$, $0\leq i,j\geq n$, $k\geq 0$, as defined in \eqref{functionformal1},  are elements of  $\operatorname{Rat}^{1}_{\Q(2 \pi i)}\left( \C^{n},\mathcal{D}'\right) ((q))$. We have a differential graded algebra (over $\Q(2 \pi i)$) of formal Laurent series
\[
\left(d, \wedge, \operatorname{Rat}^{\bullet}_{\Q(2 \pi i)}\left( \C^{n},\underline{\mathcal{D}} \right)\left( \left(  q \right) \right)   \right). 
\] 
We extend the action of $\Z^{n}$ defined in \eqref{Zaction1} extend to an action $\rho^{c}\: : \: \operatorname{Rat}^{\bullet}_{\Q(2 \pi i)}(\C^{n},\underline{\mathcal{D}} )\left( \left(  q \right) \right)  \to\operatorname{Map}\left(\Z^{n}, \operatorname{Rat}^{\bullet}_{\Q(2 \pi i)}(\C^{n},\underline{\mathcal{D}} )\left( \left(  q \right) \right)   \right)$ via
\[
 \rho^{c}\left(q\right)(m_{1}, \dots , m_{n}) :=q, \quad  \rho^{c}\left(dz_{i}\right)(m_{1}, \dots , m_{n}) :=q^{m_{i}}dz_{i}, \quad  \rho^{c}\left(\frac{1}{z_{i}}\right)(m_{1}, \dots , m_{n}):= \frac{q^{-m_{i}}}{z_i}
\]
and 
\[
 \rho^{c}\left(\frac{1}{z_{i}-1}\right)(m_{1}, \dots , m_{n}) :=\begin{cases}
 \sum_{l=0}^{\infty}\left( q^{n}z_{i}\right)^{l} & n> 0\\
 \sum_{l=0}^{\infty}\frac{q^{-n}}{z_{i}}\left( \frac{1}{z_{i}^{l}}\right)^{l}& n<0
 \end{cases}
\]
The nerve gives rise to a cosimplicial unital commutative differential graded algebra $A^{\bullet, \bullet}$, where $A^{p,q}:=\operatorname{Map}\left(\left( \Z^{n}\right)^{p} , \operatorname{Rat}^{q}_{\Q(2 \pi i)}(\C^{n},\underline{\mathcal{D}} )\left( \left(  q \right) \right)   \right)$. For $0\leq i,j\leq n$, we denote by $\underline{\gamma}_{i,j}\in A^{1,0}$ the group homomorphism ${\gamma(0)}_{i,j}\: : \: \Z^{2n}\to \C$ defined in the previous section. We denote by $\underline{\phi}_{i,j}\in A^{0,1}$ the $1$-forms in \eqref{functionformal1} considered as formal power series in $q$. $\operatorname{Tot}_{N}(A)$ carries a $C_{\infty}$-structure ${m}_{\bullet}$. 
Let $\underline{B'}_{n}$ be the $C_{\infty}$-subalgebra of $\operatorname{Tot}_{N}(A)$ generated by
	\[
\underline{\phi}_{i,j}^{(k)}, \underline{\gamma}_{i,j} \text{ for }k\geq 0,\,i,j=0,1,\dots,n.
	\]
Let $\underline{J}\subset \underline{B'}$ be the $C_\infty$ ideal generated by all the two forms
$$m_{l}\left(\underline{\gamma}_{i_{1},j_{1}}, \dots, \underline{\gamma}_{i_{l},j_{l}}\right)$$
for $l>2$. We denote by $\underline{B}_{n}:= \underline{B'}_{n}/\underline{J}$ the quotient $C_{\infty}$-algebra. Notice that a consequence of the conjecture in Remark \ref{conj} is that $\underline{B'}_{n}$ is a $1$-model.
\begin{prop}
Let ${A'}_{n}\subset {B'}_{n}$ as in Definition \ref{defAn}. 
\begin{enumerate}
\item There is a strict morphism of complex differential graded algebras
\[
\varphi\: : \: ev_{0}H({A'}_{n})\to \operatorname{Rat}^{\bullet}_{\Q(2 \pi i)}\left( \C^{n},\underline{\mathcal{D}} \right)\left( \left(  q \right) \right)\otimes \C
\]
which preserves the action $\rho^{c}$,
\item $\varphi$ induces a strict morphism of $C_{\infty}$-algebras
\[
\varphi\: : \: ev_{0}H({B'}_{n})\to \underline{B'}_{n}\otimes \C
\]
such that $\varphi\left(J_{1,DR} \right)\subset \underline{J}$, and
\item $\varphi$ induces a strict morphism of $C_{\infty}$-algebras
\[
\varphi\: : \: ev_{0}H({B'}_{n})/J_{1,DR}\to \underline{B}_{n}\otimes \C.
\]
\end{enumerate}
\end{prop}
\begin{proof}
Point 1 follows by \eqref{functionformal1}. Point 2 and 3 are straightforward.
\end{proof}
\begin{cor}\label{formal property}
Proposition \ref{generalrelation}  holds mutatis mutandis for $\underline{B'}_{n}$, i.e by replacing $w(u)^{(k)}_{i,j}$ with $\underline{\phi}^{(k)}_{i,j}$, ${\gamma(u)}_{i,j}$ with $\underline{\gamma}_{i,j}$ and setting ${\beta(u)}_{i,j}={\nu(u)}_{i,j}=0$.
\end{cor}
\begin{proof}
Notice that $\underline{B'}_{n}$ is a formal version of the image of $ev_{0}H \left({B'}_{n} \right) $, in particular the proof of Proposition \ref{generalrelation} is independent by the choice of $\tau$. We get that the statements hold for $\underline{B'}_{n}$ as well.
\end{proof}

\begin{cor}\label{formalconn}
The vector space decomposition of Theorem \ref{vspacedect} induces a well defined Hodge type decomposition on $\left( \underline{B}_{n}, m_{\bullet}\right) $. By the homotopy transfer theorem it gives a $1$-minimal model $g_{\bullet}\: : \: \left(W, m_{\bullet}^{W} \right)\to \mathcal{F} \left( \underline{B}_{n}, m_{\bullet}\right)$ which corresponds to the Maurer-Cartan element
\begin{eqnarray*}
\underline{C(0)} &=-&\sum_{i}\underline{\phi}^{(0)}_{i,0}X_{i}-\underline{\gamma}_{i}Y_{i}-\sum_{i,k\geq 1}\underline{\phi}^{(k)}_{i,0}\operatorname{Ad}^{(k)}_{Y_{i}}\left(X_{i} \right)\\
& - & \sum_{j<i,k\geq 1} \left( \underline{\phi}^{(k)}_{j,0}+\underline{\phi}^{(k)}_{0,i}- \underline{\phi}^{(k)}_{j,i}\right) \operatorname{Ad}^{(k)}_{Y_{j}}\left(X_{i} \right) \\
\end{eqnarray*}
\end{cor}
in the $L_{\infty}$-algebra $\operatorname{Conv}_{1-C_{\infty}}\left(\left(W_{+}, m_{\bullet}^{W_{+}}\right) ,\left(\underline{B}_{n} ,m_{\bullet}\right) \right)$.
\begin{proof}
Consider the vector space decomposition of Theorem \ref{vspacedect} and the degree zero geometric connection $H_{*}C$ of Theorem \ref{gemetriczeroconn}. The map $\varphi$ induces a strict morphism of $L_{\infty}$-algebras
\[
\left( \varphi ev_{0}H\right)_{*}  \: : \: B_{n}\widehat{\otimes}\widehat{\overline{\mathfrak{t}}}_{1,n+1}\to \left( \underline{B}_{n}\otimes \C\right)\widehat{\otimes}\widehat{\overline{\mathfrak{t}}}_{1,n+1}
\]
that preserves Maurer-Cartan elements. In particular $\underline{C(0)}=\left( \varphi ev_{0}H\right)_{*}C$.
\end{proof}
The quotient map $\Z^{n}\to \left\lbrace e\right\rbrace $ induce a map between cosimplicial graded module
\[
i^{\bullet, \bullet}\: : \: A^{\bullet, \bullet}\to  \operatorname{Rat}^{\bullet}_{\Q(2 \pi i)}\left( \C^{n},\underline{\mathcal{D}} \right)\left( \left(  q \right) \right)   
\]
where the latter carries a trivial cosimplicial structure. This induces a morphism of $C_{\infty}$-structure
\[
i\: : \: \operatorname{Tot}^{\bullet}_{N}(A)\to \operatorname{Rat}^{\bullet}_{\Q(2 \pi i)}\left( \C^{n},\underline{\mathcal{D}} \right)\left( \left(  q \right) \right)
\]
where the latter is a unital commutative differential graded algebra. In particular we have $i(\underline{B'}_{n})\subset \operatorname{Rat}^{\bullet}_{\Q(2 \pi i)}\left( \C^{n},\underline{\mathcal{D}} \right)\left[ \left[  q \right] \right] $ and $i(\underline{J})=0$. Let $J_{q}$ be the completion of the augmentation ideal of $\operatorname{Rat}^{\bullet}_{\Q(2 \pi i)}\left( \C^{n},\underline{\mathcal{D}} \right)\left[ \left[  q \right] \right] $. We have a differential graded algebra map $\pi\: : \: \operatorname{Rat}^{\bullet}_{\Q(2 \pi i)}\left( \C^{n},\underline{\mathcal{D}} \right)\left[ \left[  q \right] \right] \to \operatorname{Rat}^{\bullet}_{\Q(2 \pi i)}\left( \C^{n},\underline{\mathcal{D}} \right) $ given by the quotient. Hence we get a strict morphism of $C_{\infty}$-algebras
\[
p:=\pi\circ i\: : \: \underline{B}_{n}\to \operatorname{Rat}^{\bullet}_{\Q(2 \pi i)}\left( \C^{n},\underline{\mathcal{D}} \right)
\]
such that 
\begin{eqnarray*}
p(\underline{\gamma}_{i,j})=0,&\quad&
p(\underline{\phi}_{i,j}^{(0)})=\frac{1}{2\pi i}\frac{dz_{i}}{z_{i}}-\frac{1}{2\pi i}\frac{dz_{i}}{z_{i}},\\
p(\underline{\phi}_{i,j}^{(1)})=\left( \frac{1}{2}  +\frac{z_{j} }{z_{i}-z_{j}} \right)\left( \frac{dz_{i}}{z_{i}}-\frac{dz_{j}}{z_{j}}\right) ,&\quad&
p(\underline{\phi}_{i,j}^{(l)})=\frac{-\left( 2 \pi i\right)^{l-1} B_{l}}{l!}\left( \frac{dz_{i}}{z_{i}}-\frac{dz_{j}}{z_{j}}\right)\text{ for }l>1
\end{eqnarray*}
We define the complex differential graded algebra ${A}_{KZ,n}$ as follows. Let $A_{KZ,n}$ be the unital differential graded subalgebra of $A_{DR}\left( \operatorname{Conf}_{n}\left(\C-\left\lbrace0,1 \right\rbrace  \right)\right)  $ generated by $\omega_{i,-1}$, $\omega_{i,0}$, and $\omega_{i,j}$ given by
\[
\omega_{ij}:=d \log(z_{i}-z_{j})=\frac{dz_{i}-dz_{j}}{z_{i}-z_{j}}
\] 
for $1\leq i\neq j\leq n$ such that $z_{-1}:=1$ and $z_{0}:=0$. These differential forms satisfy the so called \emph{Arnold relations}
\begin{equation}\label{Arnold}
\omega_{ij}\omega_{jk}+\omega_{ki}\omega_{ij}+\omega_{jk}\omega_{ki}=0
\end{equation}
for $i,j,k,l$ distinct.
The next result is proved in \cite{Arnold}.
\begin{prop}
$A_{KZ,n}$ is a model for $A_{DR}\left( \operatorname{Conf}_{n}(\C-\left\lbrace0,1 \right\rbrace )\right)  $.
\end{prop}
We define the rational Lie algebra $\mathfrak{t}_{n}$ with generators $T_{i,j}$ for $-1\leq i\neq j\leq n$ with $j>0$ or $i>0$ such that
\begin{eqnarray}\label{Kohniodrinfeld}
T_{ij}=T_{ji}, \quad \left[ T_{ij},T_{ik}+T_{jk}\right]=0, \quad  \left[ T_{ij},T_{kl}\right]=0
\end{eqnarray}
for $i,j,k,l$ distinct. We call $\mathfrak{t}_{n}$ the \emph{Kohno-Drinfeld Lie algebra}. The KZ connection 
is a flat connection $d-\omega_{KZ,n}$ on $\operatorname{Conf}_{n}\left(\C-\left\lbrace 0,1\right\rbrace\right) \times \mathfrak{t}_{n} $ where
\[
\omega_{KZ,n}:=\sum_{-1\leq i\neq j\leq n, i>0\text{ or }j>0} \omega_{ij}T_{ij}.
\]
 Let $\underline{A}_{KZ,n}$ be the unital differential commutative graded algebra over $\Q(2\pi i)$ generated by the degree $1$ closed forms ${\omega}_{i,-1}$, ${\omega}_{i,0}$, and $\omega_{i,j}$ for $1\leq i<j\leq n $. There is a canonical Hodge decomposition for $\underline{A}_{KZ,n}$ given by $\left(\underline{A}_{KZ,n},0 \right) $, where $\underline{A}^{1}_{KZ,n}$ is considered to equipped with the basis ${\omega}_{i,-1}$, ${\omega}_{i,0}$, and $\omega_{i,j}$ for $-1\leq i\neq j\leq n, i>0\text{ or }j>0$. Then the $C_{\infty}$-morphism $g_{\bullet}^{KZ,n}\: : \: \left( W',m_{\bullet}^{W'}\right) \to \underline{A}_{KZ,n}$ constructed via the homotopy transfer theorem is strict and corresponds to the identity map $\omega_{l,k}\mapsto \omega_{l,k}$.. We denote its inverse by $f_{\bullet}^{KZ,n}$.
\begin{cor}
The degree zero geometric connection associated to $\mathcal{F}\left(g_{\bullet}^{KZ,n} \right) $
\[
\underline{\omega}_{KZ,n}:=\sum_{-1\leq i\neq j\leq n, i>0\text{ or }j>0} \omega_{ij}T_{ij}
\]
where the fiber is given by the Lie algebra $\mathfrak{t}_{n}$.
\end{cor}
\begin{proof}
We set $T_{i,j}:=\left( s\left( \omega_{i,j}\right) \right)^{*}$ where $-1\leq i\neq j\leq n, i>0\text{ or }j>0$. The equation \eqref{Arnold} implies the second relation in \eqref{Kohniodrinfeld}. The third relation comes from the fact that there are no relations between forms $\omega_{ij},\omega_{kl}$ with distinct indices.
The formula for the connection comes from $g_{\bullet}^{KZ,n}$ which is strict.
\end{proof}
There is a differential graded algebra map
\[
 f^{A}\: : \:\underline{A}_{KZ,n}\otimes \C\to A_{KZ,n},
\]
such that $f^{A}_{*}\left( \underline{\omega}_{KZ,n}\right)={\omega}_{KZ,n}$. We have a diagram of $1-C_{\infty}$-algebras over $\Q(2\pi i)$
\[
\begin{tikzcd}
 \mathcal{F}\left(\underline{B}_{n} ,m_{\bullet}\right) \arrow[r,  "\mathcal{F}\left( p\right) "]
 &  \mathcal{F}\left(\underline{A}_{KZ,n}, d, \wedge\right) \arrow[d, shift right, "\mathcal{F}\left( f^{KZ,n}_{\bullet}\right) "]\\\left( W,m_{\bullet}^{W}\right) \arrow[u, shift right, "g_{\bullet} "]&  \mathcal{F}\left( {W'}, m_{\bullet}^{W'}\right) 
\end{tikzcd}
\]
where $g_{\bullet}$ is the $1-C_{\infty}$-algebra morphism of Corollary \ref{formalconn}. We set $q_{\bullet}:= p\circ g_{\bullet}$ is a morphism of $1-C_{\infty}$-algebras. This corresponds to a morphism of differential graded coalgebras 
\[
Q\: : \: T^{c}\left( W^{1}_{+}[1]\right) \to T^{c}\left( {W'}^{1}_{+}[1]\right).
\]
The restriction of its dual gives a Lie algebra morphism
\[
Q^{*}\: : \: \widehat{\mathfrak{t}}_{n}\widehat{\otimes} \Q(2\pi i)\to \widehat{\overline{\mathfrak{t}}}_{1,n+1}\widehat{\otimes} \Q(2 \pi i)
\]
We calculate $Q^{*}$ via the method of Subsection \ref{Morphism of homological pairs}. We have $q_{*}(\underline{\omega}_{KZ,n} )=p^{*}\left( \underline{C(0)}\right)$ since the connection is quadratic, where $q_{*}({\omega}_{KZ,n})=\sum_{\substack{-1\leq i\neq j\leq n\\j>0}}\omega_{i,j}Q^{*}\left( T_{i,j}\right)$. Moreover 
\begin{eqnarray*}
p^{*}\left( \underline{C_{0}}\right)  &=&\sum_{i}p\left( \underline{\phi}^{(0)}_{i,0}\right) X_{i}-p\left( \underline{\gamma}_{i}\right) Y_{i}-\sum_{i,k\geq 1}p^{*}\left( \underline{\phi}^{(k)}_{i,0}\right) \operatorname{Ad}^{(k)}_{Y_{i}}(X_{i})\\
& - & \sum_{j<i,k\geq 1}p\left( \underline{\phi}^{(k)}_{j,0}+\underline{\phi}^{(k)}_{0,i}- \underline{\phi}^{(k)}_{j,i}\right) \operatorname{Ad}^{(k)}_{Y_{j}}(X_{i})\\
& = & \sum_{i}\left( -\frac{dz_{i}}{z_{i}}\sum_{i=0}^{\infty}\frac{B_{k}}{k!}\operatorname{Ad}^{(k)}_{2 \pi i Y_{i}}\left( \frac{X_{i}}{2\pi i}\right)-\left[2\pi iY_{i},\frac{X_{i}}{2\pi i} \right]\frac{dz_{i}}{z_{i}-1}\right) \\
& - & \sum_{j<i}\frac{dz_{j}}{z_{j}-1}+\frac{dz_{i}}{z_{i}-1}-\frac{dz_{j}-dz_{i}}{z_{j}-z_{i}}\left[Y_{j}, X_{i} \right].
\end{eqnarray*}
Hence $Q^{*}\: : \: {\widehat{\mathfrak{t}}}_{n}\otimes \Q(2 \pi i)\to \widehat{\overline{\mathfrak{t}}}_{1,n+1}\widehat{\otimes} \Q(2 \pi i)$ is given by
\[
Q^{*}\left(T_{-1,i} \right)=-\sum_{j}\left[ Y_{j},X_{i}\right] ,\quad  Q^{*}\left(T_{0,i} \right)=-\sum_{i=0}^{\infty}\frac{B_{k}}{k!}\operatorname{Ad}^{(k)}_{2 \pi i Y_{i}}\left( \frac{X_{i}}{2\pi i}\right),\quad Q^{*}\left(T_{i,j} \right)=-\left[Y_{i}, X_{j} \right] .
\]
for $1\leq i,j\leq n$.
\begin{thm}\label{thmfinale}
Let $B_{n}$, $\underline{B_{n}}$ as above.
\begin{enumerate}
 \item The map $p$ induces a Lie algebra morphism $Q\: : \: \widehat{\mathfrak{t}}_{n}\otimes \Q(2 \pi i) \to \widehat{\overline{\mathfrak{t}}}_{1,n+1}\widehat{\otimes} \Q(2 \pi i)$ which induces a differential graded Lie algebra map
 \[
 q_{*}=\left( Id{\otimes}Q^{*}\right) \: : \:  A_{KZ,n} \otimes\widehat{\mathfrak{t}}_{n}\to A_{KZ,n} \widehat{\otimes}\widehat{\overline{\mathfrak{t}}}_{1,n+1}
 \]
 \item For a fixed $n>0$ and $\tau\in \mathbb{H}$. We denote the KZB connection with $${\omega}_{KZB,n}^{\tau}\in A_{DR}\left( \C^{n}-\mathcal{D}\right)\widehat{\otimes}\widehat{\overline{\mathfrak{t}}}_{1,n+1}.$$ We have $$\lim_{\tau\to i\infty }{\omega}_{KZB,n}^{\tau}=q_{*}\left( {\omega}_{KZ,n}\right) $$
\end{enumerate}
\end{thm}
\begin{proof}
The first part is proved above. The second part follows as the proof of Proposition \ref{thmfinale0}.
\end{proof}

\section{Proofs and calculations}

\subsection{Proof of Theorem \ref{vspacedect}}\label{sectionproofvspacevect}
The commutative differential graded algebra $\mathcal{Y}_{n}$ is the the free differential commutative graded algebra generated by elements of degree $1$
\[
w_{j}, v_{j}, \text{ for }j=1, \dots, n, \quad w_{ij}\text{ for }0\leq i,j\leq n.
\]
modulo the following relations
\begin{eqnarray}
\label{Rel1}w_{ii}& =& 0,\\
\label{Rel2}w_{ij}-w_{ji}& =& 0,\\
\label{Rel3}w_{i}\wedge v_{i}& =& 0,\\
\label{Rel4}w_{ij}\wedge w_{i}-w_{ij}\wedge w_{j}& = & 0,\\
\label{Rel5}w_{ij}\wedge v_{i}-w_{ij}\wedge v_{j}& = & 0,\\
\label{Rel6}w_{il}\wedge w_{jl}+w_{jl}\wedge w_{ij}+w_{ji}\wedge w_{il}& = & 0.
\end{eqnarray}
The differential is given by
\[
dw_{i,j}=w_{i}\wedge v_{j}+w_{j}\wedge v_{i}, \quad dw_{i}=dv_{j}=0.
\]
\begin{lem}
	There exists a vector space decomposition
	\begin{equation}\label{holonomyconfigurationspaces}
	\mathcal{Y}_{n}={W''}\oplus {\mathcal{M}''}\oplus d{\mathcal{M}''}
	\end{equation}
	as in \eqref{vspgen} such that 
	\begin{enumerate}
		\item ${W''}^{1}\subset \mathcal{Y}_{n}$ is the vector space generated by $w_i,v_j$ for $j,i=1, \dots, n$.
		\item ${W''}^{2}\subset \mathcal{Y}_{n}$ is the vector space generated by 
		\begin{eqnarray*}
			w_{i}\wedge v_{j},\quad w_{ij}\wedge v_{i},\quad
			w_{i}\wedge w_{j},\quad  w_{ij}\wedge w_{i},\quad v_{i}\wedge v_{j},
		\end{eqnarray*}
		for $1\leq i<j\leq n$,
		and
		\begin{eqnarray*}
			&&w_{ij}\wedge v_{k}+w_{ik}\wedge v _{j}+w_{kj}\wedge v_{i}\\
			& & w_{ij}\wedge w_{k}+w_{ik}\wedge w _{j}+w_{kj}\wedge w_{i}\\
		\end{eqnarray*}
		for $1\leq i<j<k\leq n$.
		\item ${\mathcal{M}''}^{1}\subset \mathcal{Y}_{n}$ is the vector space generated by $w_{ij}$ for $j,i=1, \dots, n$.
		\item ${\mathcal{M}''}^{2}\subset \mathcal{Y}_{n}$ is the vector space generated by
		\begin{eqnarray*}
			w_{ij}\wedge w_{kl},\,\text{for  any }i<j,\,k<l,\\
			w_{ij}\wedge v_{k},\,\text{for  any }i<j,\,k<j,\\
			w_{ij}\wedge w_{k},\,\text{for  any }i<j,\,k<j,\\
		\end{eqnarray*}
	\end{enumerate}
\end{lem}
\begin{proof}
	We define ${W''}^{1}$ and ${\mathcal{M}''}^{1}$ as in point 1. and  3. resp. It is immediate to see that all the elements listed at point 2. are closed and not exact and their cohomology classes are linearly independent. It remains to prove that ${\mathcal{M}''}^{2}$, as defined above, contains no closed forms except zero. In this proof we will call the relations $\eqref{Rel1}-\eqref{Rel3}$ \emph{trivial relations}. For $i=0$ we set $v_{i}:=0$, and any letter $i,j,k,l$ is an integer between $0$ and $n$.\\
	We define the vector spaces $V_{1}, V_{2}$ and $V_{3}$ as follows.
	\begin{itemize}
		\item $V_{1}$ is the vector space generated by $w_{ij}\wedge w_{kl},\,\text{for  any }i<j,\,k<l;$
		\item $V_{2}$  is the vector space generated by $w_{ij}\wedge v_{k},\,\text{for  any }i<j,\,k<j;$
		\item $V_{3}$  is the vector space generated by $w_{ij}\wedge w_{k},\,\text{for  any }i<j,\,k<j;$
	\end{itemize}
	Note that
	\[
	V_{1}\oplus V_{2}\oplus V_{3}={\mathcal{M}''}^{2}
	\]
	and
	\[
	dV_{r}\cap d V_{s}=\left\lbrace 0\right\rbrace \text{ for }r\neq s.
	\]
	Let $a$ be a closed element of ${W''}$ of degree $2$. We write
	\[
	a=\underset{:=a_{1}}{\underbrace{\sum _{(i<j);(k<l)}\lambda_{(i<j);(k<l)}w_{ij}w_{kl}}}+\underset{:=a_{2}}{\underbrace{\sum _{(i<j);(k<j)}\mu_{(i<j);(k<j)}w_{ij}\wedge v_{k}}}+ \underset{:=a_{3}}{\underbrace{\sum _{(i<j);(k<j)}\alpha_{(i<j);(k<j)}w_{ij}\wedge w_{k}}}
	\]
	where $a_{i}\in V_{i}$ for any $i$ and we have $da=0$ if and only if $da_{1}=da_{2}=da_{3}=0$. We start with $a_{1}$. We define four vector subspaces
	\begin{itemize}
		\item $V_{1}^{0}\subset  V_{1}$ is the vector space generated by the  $w_{ij}\wedge w_{kl},\,\text{for  any }i<j,\,k<l;$ such that $\left| \left\lbrace i,j\right\rbrace \cap \left\lbrace k,l\right\rbrace\right| =0$,
		\item $V_{1}^{1}\subset  V_{1}$ is the vector space generated by the  $w_{ij}\wedge w_{kl},\,\text{for  any }i<j,\,k<l;$ such that $\left| \left\lbrace i,j\right\rbrace \cap \left\lbrace k,l\right\rbrace\right| =1$
		\item  ${V'}_{1}^{0}$ is the vector space generated by the $w_{i}\wedge v_{j}\wedge w_{kl},v_{i}\wedge w_{j}\wedge w_{kl}\,\text{for  any }i<j,\,k<l;$ such that $\left| \left\lbrace i,j\right\rbrace \cap \left\lbrace k,l\right\rbrace\right| =0$,
		\item ${V'}_{1}^{1}$ is the vector space generated by the  $w_{i}\wedge v_{j}\wedge w_{kl},v_{i}\wedge w_{j}\wedge w_{kl},\,\text{for  any }i<j,\,k<l;$ such that $\left| \left\lbrace i,j\right\rbrace \cap \left\lbrace k,l\right\rbrace\right| =1$
	\end{itemize}
	We have $V_{1}^{0}\cap V_{1}^{1}=\left\lbrace 0\right\rbrace $ and $dV_{1}^{i}\subset{V'}_{1}^{i} $ for $i=0,1$. Notice that there is no relation involving the elements of  ${V'}_{1}^{0}$. On the other hand the only relations involving elements of ${V'}_{1}^{0}$ are \eqref{Rel4} and \eqref{Rel5}. They are between elements of ${V'}_{1}^{1}$. Hence ${V'}_{1}^{0}\cap {V'}_{1}^{1}=\left\lbrace 0\right\rbrace $ and $dV_{1}^{0}\cap d   V_{1}^{1}=\left\lbrace 0\right\rbrace $. We write $a_{1}=a_{1}^{1}+a_{1}^{0}$, with $a_{1}^{i}\in V_{1}^{i}$ for $i=0,1$. Then $da_{1}=0$ if and only if $da_{1}^{i}=0$ for $i=0,1$. We have two cases.
	\begin{enumerate}
		\item We can write $a_{1}^{0}=\sum_{i<j,\,k<l, \,j<k} \lambda_{i,j,k,l} w_{i,j}\wedge w_{kl}$. Since there is no relations involving 
		\[
		v_{i}\wedge w_{j}\wedge w_{kl},w_{i}\wedge w_{j}\wedge w_{kl},\,\text{for any }i<j,\,k<l;
		\]
		inside ${V'}_{1}^{0}$, we get $da_{1}^{0}=0$ if each $\lambda_{i,j,k,l} =0$.
		\item For a $k<j<l$ we define $V^{k,j,l}_{1}$ as the vector space generated by $w_{kj}\wedge w_{kl},w_{kj}\wedge w_{jl},w_{kl}\wedge w_{jl}$. We have $V_{1}^{1}=\oplus_{k<j<l}V^{k,j,l}_{1}$, since the only relation involving elements of $V_{1}^{1}$ is \eqref{Rel6}. We define ${V'}^{k,j,l}_{1}\subset {V'}_{1}^{1}$ as the vector space generated by
		\[
		w_{s_{1}}\wedge v_{s_{2}}\wedge w_{s_{3}s_{4}},	v_{s_{1}}\wedge w_{s_{2}}\wedge w_{s_{3}s_{4}},
		\]
		where $\left(s_{1}<s_{2} ;s_{3}<s_{4} \right)\in \left\lbrace {(k<j);(k<l)} ,{(k<j);(j<l)} ,{(k<l);(j<l)} \right\rbrace$. The only non -trivial relations in this subspace are \eqref{Rel4} and \eqref{Rel5} and they imply
		\[
		{V'}_{1}^{1}=\oplus_{k<j<l}{V'}^{k,j,l}_{1}.
		\]
		We can write $a_{1}$ as 
		\[
		a_{1}^{1}=\sum_{k<j<l}\left( \lambda^{1}_{(k<j);(k<l)} w_{kj}\wedge w_{kl}+\lambda^{2}_{(k<j);(j<l)} w_{kj}\wedge w_{jl}+\lambda^{3}_{(k<l);(j<l)} w_{kl}\wedge w_{jl}\right).
		\]
		In particular since $d{V}^{k,j,l}_{1}\subset {V'}^{k,j,l}_{1}$ we have $da_{1}^{1}=0$ if and only if
		\[
		d\left( \lambda^{1}_{(k<j);(k<l)} w_{kj}\wedge w_{kl}+\lambda^{2}_{(k<j);(j<l)} w_{kj}\wedge w_{jl}+\lambda^{3}_{(k<l);(j<l)} w_{kl}\wedge w_{jl}\right)=0.
		\]
		The equation above corresponds to $\lambda^{1}_{(k<j);(k<l)}=\lambda^{3}_{(k<l);(j<l)}$ and $\lambda^{1}_{(k<j);(k<l)}=-\lambda^{2}_{(k<l);(j<l)}$. Hence 
		\[
		a_{1}^{1}=\sum_{k<j<l}\lambda_{(k<j);(k<l)}\left(  w_{kj}\wedge w_{kl}-w_{kj}\wedge w_{jl}+w_{kl}\wedge w_{jl}\right)
		\]
		i.e. $a_{1}^{1}=0$ by \eqref{Rel6}.
	\end{enumerate}
	Now assume that $da_{2}=0$. We define $V_{2}^{i,j,k}\subset V^{2}$ for any $i<j<k$ as the vector space generated by
	\[
	w_{ik}\wedge v_{j}, w_{jk}\wedge v_{i}.
	\]
	Then $V_{2}=\oplus_{i<j<k}V_{2}^{i,j,k}$ since there are no non -trivial relations involving the elements of $V_{2}$. We define ${V'}_{2}$ as the vector space generated by
	\[
	w_{i}\wedge v_{j}\wedge v_{k}, w_{j}\wedge v_{i}\wedge v_{k}, w_{k}\wedge v_{i}\wedge v_{j}.
	\] 
	for any $i<j<k$. Hence $${V'}_{2}=\oplus_{i<j<k}{V'}_{2}^{i,j,k}$$ where ${V'}_{2}^{i,j,k}$ is the vector space generated by the terms above for a fixed $i<j<k$. We can write  $$
	a_{2}=\sum_{i<j<k}\lambda_{i<j<k}^{1}w_{ik}\wedge v_{j}+\lambda_{i<j<k}^{2}w_{jk}\wedge v_{i}
	$$
	and since $dV_{2}^{i,j,k}\subset{V'}_{2}^{i,j,k} $ we have that $da_{2}=0$ if and only if $$d\left(\lambda_{i<j<k}^{1}w_{ik}\wedge v_{j}+\lambda_{i<j<k}^{2}w_{jk}\wedge v_{i} \right)=0,$$ for any $i<j<k$. This implies $\lambda_{i<j<k}^{2}=0=\lambda_{i<j<k}^{1}$, i.e $a_{2}=0$.\\
	The proof for $a_{3}$ is analogous to the one for $a_{2}$.
\end{proof}
There is a quasi-isomorphism $f\::\:A_{n}\to \mathcal{Y}_{n}$
defined via $f(\nu_{i,j}):=\nu_{i}-\nu_{j}$ for $0\leq i,j\leq n$ and 
\[
f(w_{i,j}^{(k)}):=\begin{cases} w_{i}-w_{j}& \text{ for }0\leq i,j\leq n,\, k=0\\
w_{i,j}& \text{ for }1\leq i,j\leq n,\, k=1\\
0& \text{ otherwise }.
\end{cases}
\]
(see \cite{LevBrown}).
\begin{prop}\label{prop1}
	There exists a Hodge type vector space decomposition of $A_{n}$
	\[
	A_{n}=W'\oplus\mathcal{M}'\oplus d\mathcal{M}'
	\]
	such that ${W'}^{i}=f\left( {W''}^{i}\right) $ and ${\mathcal{M}'}^{i}=f\left( {\mathcal{M}''}^{i}\right) $ for $i=0,1,2$.
\end{prop}
The proof of the proposition above follows by the following lemmas.
\begin{lem}
	Consider ${W'}^{i}\subset A_{n}^{i}$ for $i=1,2$ as above. The vector space ${W'}^{i}\subset A_{n}^{i}$ contains only closed not exact elements for $i=1,2$. 
\end{lem}
\begin{proof}
	The map $f$ above defines an isomorphism between ${W'}^{i}\subset A_{n}^{i}$  and ${W''}^{i}\subset A_{n}^{i}$ for $i=1,2$. Since it is a quasi-isomorphism, then results follow.
\end{proof}
\begin{lem}Let $V$ be the vector subspace of closed elements in $\ker(f)$.
	\begin{enumerate}
		\item Let ${V'}^{1}=0$ and ${V'}^{2}\subset \ker(f)$ be the vector space generated by 
		the elements $w(u)_{i,j}^{(p)}\wedge w(u)_{k,l}^{(q)}$, $0\leq i,j,k,l\leq n$ such that $q$ or $p>1$ and elements of the form $w(u)_{i,0}^{(1)}\wedge w(u)_{k,l}^{(q)}$, $0\leq i,k,l\leq n$. Then we have 
		\[
		\ker(f)^{2}=V^{2}\oplus {V'}^{2}.
		\]
		\item We have $V^{1}=0$ and $V^{2}$ is the vector space generated by $ \left( (\nu(u)_{i,j}+\beta(u)_{i,j})\right) \wedge w(u)_{i,j}^{(k)}$ for $0\leq i,j\leq n,\, k>1$ and $(\nu(u)_{i,0}+\beta(u)_{i,0}) \wedge w(u)_{i,0}^{(1)}$ for $0\leq i\leq n$. 
		and $\ker(f)^{2}\cap \ker(d)=V$.
	\end{enumerate} 
\end{lem}
\begin{proof}
	Point 1 is immediate. Let $a\in \ker(f)^{2}\cap \ker(d)$. If $a\in W'$ then this is a contradiction with the fact that $f$ is a quasi-isomorphism. Hence $a $ is exact. Now assume that $a\in V'$. We have  $d A_{n}^{1}\cap V'=0$, then $a=0$ by the definition of $d$.
\end{proof}
We can conclude the proof of Proposition \ref{prop1}. From the map $f\: : \: A_{n}\to \mathcal{Y}_{n}$ we have
\begin{align*}
A_{n}^{2}&\cong\mathcal{Y}_{n}^{2}\oplus \ker f^{2}\\
& \cong {W''}^{2}\oplus\left(  {\mathcal{M}''}^{2}\oplus V'\right) \oplus \left( d\left( {\mathcal{M}''}^{1}\right) \oplus V\right). \\
\end{align*}
The isomorphism above give the desired decomposition.\\
Consider the strict $C_{\infty}$-morphism $p^{1}\: :\: {B'}_{n}\to A_{n}$ defined in \eqref{morphism}. We have immediately the following lemma.
\begin{lem}
The map $p^{1}\: : \: \bar{W}\to W'$ is an isomorphism in degree $0,1,2$.
\end{lem}
Recall Proposition \ref{propquot}. Since $p^1(J)=0$ we have a well-defined strict morphism of $C_{\infty}$-algebras $p^{1}\: :\: {B}_{n}\to A_{n}$. The vector space decomposition \eqref{vspgen} 
\begin{equation}
B_{n}=W\oplus \mathcal{M}\oplus D\mathcal{M}
\end{equation}
satisfies $W'=p^{1}(W)$, $\mathcal{M}'=p^{1}(\mathcal{M})$ in degree $0,1$ and $2$.\\
We are ready for the proof of Theorem \ref{vspacedect}.\\
\begin{proof}Notice that $p^{1}\: : \:B_{n}^{i}\to A_{n}^{i}$ is an isomorphism for $i=0,1$. For $i=2$ we have
\begin{align*}
B_{n}^{2}&\cong A_{n}^{2}\oplus \left( \ker p^{1}\right)^{2}\\
& \cong {W'}^{2}\oplus\left(  {\mathcal{M}'}^{2}\oplus \left( \ker p^{1}\right)^{2}\right) \oplus \left( d\left( {\mathcal{M}'}^{1}\right) \right) \\
\end{align*}
By Lemma \ref{qisomp1} below we have that this is a Hodge type decomposition.
\end{proof} 
\subsection{Proof of Theorem \ref{theoremmodel} and Theorem \ref{1-ext}}\label{sectionmodel}
We have a commutative diagrams of $C_{\infty}$-algebras.
\[
\begin{tikzcd}
& {A}_{n}\arrow{r}{\psi^{1}}&\operatorname{Tot}_{N}A_{DR}\left(\left( \C^{n}-\mathcal{D}\right)_{\bullet}\left( \Z^{2n}\right) \right)\\
{B'}_{n}\arrow{ur}{p^{1}}\arrow{rr}{H}& &\left( \operatorname{Tot}_{N}A_{DR}\left(\left( \C^{n}-\mathcal{D}\right)_{\bullet}\left( \Z^{2n}\right) \right)\otimes \Omega(1)\right)\arrow{u}{ev^{1}}
\end{tikzcd}
\]
\begin{proof}{( of Theorem \ref{theoremmodel} )}
Consider the diagram above restricted at $\bar{W}$. The map $p^{1}\: : \: \bar{W}\to A_{n}$ is an isomorphism in degree $0$, $1$ and $2$. In \cite{LevBrown}, it is proved that $\psi^{1}$ is a quasi-isomorphism, since $\operatorname{ev}_{u}$ is a quasi-isomorphism as well, we conclude that $H|_{\bar{W}}$ is a quasi-isomorphism in degree $0$, $1$ and $2$.
\end{proof}

\begin{lem}\label{qisomp1}
	Consider $p^{1}\: :\: {B}_{n}\to A_{n}$. The graded vector space $\operatorname{Ker}\left( p^{1}\right)\subset B_{n}$ doesn't contain any closed form in dimension $1$ and $2$. In particular $p^{1}$ induces an isomorphism in the cohomology $H^{i}$ for $i=0,1,2$.
\end{lem}
\begin{proof}
	Let $a\in \operatorname{Ker}\left( p^{1}\right)^{2}$ such that $Da=0$. We can write $a$ as 
	\[
	a=\sum_{l> 2, I\in S} \lambda_{I} m_{l}(\phi^{(k)}_{i,j},\gamma(u)_{i_{1},j_{1}},\dots, \gamma(u)_{i_{l},j_{l}}),
	\]
	 Note that $a$ is a form of bidegree $(1,1)$. The element $\partial_{\Z^{2n}}a$ defines a map 
	\[
	\partial_{\Z^{2n}}a\: : \: \Z^{2n}\to  {A'}_{n},
	\] 
in particular if $Da=0$ then $\partial_{\Z^{2n}}a=0$. Let $V_{i,j}\subset {A'}_{n}^{1}$ be the vector space generated by
	\begin{equation}\label{generators}
\tilde{\gamma}(u)^{r}w(u)_{i,j}^{(k)}, \tilde{\gamma}(u)^{s}\left( \nu(u)_{i,j}+\beta(u)_{i,j}\right)w(u)_{i,j}^{(k)},\quad k,r,s\geq 0.\end{equation}
	By the definition we have $Dc\left( \Z^{2n}\right)\subset \oplus_{i,j} V_{i,j}$
	 We denote by $(\partial_{\Z^{2n}}a)_{i,j}$ the projection of $\partial_{\Z^{2n}}a$ on  $V_{i,j}$. Using the same method of Lemma \ref{modelelliptic curve}, $(\partial_{\Z^{2n}}a)_{i,j}$ induces a polynomial $P$ in variables $x_{1}$, \dots ,$x_{n}$ with coefficients in $V_{i,j}$. Moreover since $l>2$ the polynomial is not linear. Assume $\partial_{\Z^{2n}}a=0$ then the zero set of $P$ contains $\Z^{2n}$. It follows that all the coefficients of $P$ are $0$. This implies that the coefficients vanishes as well, in particular there are non -trivial linear relations between the generators \eqref{generators} which are not contained in \eqref{defAn}, hence a contradiction.\\
	Consider the Hodge type decompositions defined in the subsection above for ${A'}_{n}$ and $B_{n}$ resp. . Notice that $\left( p^{1}\right) ^{i}\: : \: {W'}^{i}\to {W}^{i}$ is an isomorphism of graded vector space for $i=1,2$. Together with the property above, we conclude that it is an isomorphism in the cohomology groups $H^{i}$, for $i=1,2$.
\end{proof} 
\begin{proof}{( of Theorem \ref{1-ext} )}
By the lemma above we have that $p^{1}$ is a quasi-isomorphism. On the other hand it is immediate to see that $I_{1,DR}=0$. The statement follows from Theorem \ref{model}.
\end{proof}
\subsection{Calculation of the p-kernels}\label{section p kernel}
We consider $B_{n}$ equipped with the $C_{\infty}$-structure $m_{\bullet}$ defined in the previous section. The Hodge type decomposition \eqref{vspgen} induces a homotopy retract diagram between chain complexes
\begin{equation}
\begin{tikzcd}
f\: : \: (B_{n}^{\bullet},D)\arrow[r, shift right, ""]&(W^{\bullet}, 0)\: : \: g\arrow[l, shift right, ""] 
\end{tikzcd}
\end{equation}
where $f$ is the projection on $W$ and $g$ is the inclusion. We define a map $h\: : \: B_{n}^{\bullet}\to B_{n}^{\bullet-1}$ as follows. Let $a\in B_{n}^{\bullet}$, the decomposition \eqref{vspgen} allows us to write $a=(a_{1}, a_{2}, Da_{3})$, where $a_{1}\in W^{\bullet}$, $a_{2}\in \mathcal{M}^{\bullet}$ and $a_{3}\in \mathcal{M}^{\bullet-1}$, then $h(a_{1}, a_{2}, Da_{3})=(0, a_{3}, 0)$. In particular $gf$ is homotopic to $Id_{B_{n}}$ via the cochain homotopy $h$. Notice that
\begin{equation}\label{sidecondition}
f\circ g=Id_{B_{n}}\quad f\circ h=0\quad h\circ g=0\quad h\circ h=0
\end{equation}
We calculate the p-kernel using Proposition \ref{explicitpkernel}. We adopt the following notation: for $v\in V_{1}\oplus V_{2}$ and $w_{1}\in V_{1}$ we say that $v= w_{1}\oplus V_{2}$ if $v=w_{1}+w$ for some $w\in V_{2}$. 
\begin{prop}
For $m=2$ the p-kernels are as follows.
\begin{enumerate}
\item $p_{2}(w(u)^{(0)}_{i,0}, \alpha(u)_{j,0})=\begin{cases}
m_{2}(w(u)^{(0)}_{i,0}, \alpha(u)_{j,0}),& i<j\\
D\left( w(u)^{(1)}_{i,0}\right), & i=j\\
D\left( w(u)^{(1)}_{i,j}-w(u)^{(1)}_{i,0}-w(u)^{(1)}_{j,0}\right) \oplus \mathcal{M}^{2},& i<j\\
\end{cases}$
\item $p_{2}(w(u)^{(0)}_{i,0}, w(u)^{(0)}_{j,0})=\begin{cases}
m_{2}(w(u)^{(0)}_{i,0}, w(u)^{(0)}_{j,0}),& i<j\\
0,& i=j\\
-m_{2}(w(u)^{(0)}_{j,0}, w(u)^{(0)}_{i,0}),& i<j\\
\end{cases}$
\item $p_{2}(w(u)^{(0)}_{i,0}, w(u)^{(0)}_{j,0})=\begin{cases}
m_{2}(w(u)^{(0)}_{i,0}, w(u)^{(0)}_{j,0}),& i<j\\
0,& i=j\\
-m_{2}(w(u)^{(0)}_{j,0}, w(u)^{(0)}_{i,0}),& i<j\\
\end{cases}$
\end{enumerate}
\end{prop}
\begin{proof}
It follows from $p_{2}=m_{2}$ and the decomposition\eqref{vspgen}.
\end{proof}
\begin{prop}
 Let $i,j,k$ be distinct. For $m=3$ the p-kernels are as follows.
\begin{enumerate}
\item $p_{3}(w(u)^{(0)}_{i,0}, \alpha(u)_{j,0},\alpha(u)_{k,0})=\begin{cases}
\left( j<i<k\right)_{1}\oplus \mathcal{M}^{2} ,& j<i<k\\
\mathcal{M}^{2} ,& \text{ otherwise }\\
\end{cases}$
\item  $p_{3}(w(u)^{(0)}_{i,0},\alpha(u)_{k,0} ,\alpha(u)_{j,0})=\begin{cases}
\left( k<i<j\right)_{1}\oplus \mathcal{M}^{2} ,& k<i<j\\
\mathcal{M}^{2} ,& \text{ otherwise }\\
\end{cases}$
\item  $p_{3}(\alpha(u)_{j,0} ,w(u)^{(0)}_{i,0},\alpha(u)_{k,0})=\begin{cases}
-\left( j<i<k\right)_{1}\oplus \mathcal{M}^{2} ,& j<i<k\\
-\left( k<i<j\right)_{1}\oplus \mathcal{M}^{2} ,& k<i<j\\
\mathcal{M}^{2} ,& \text{ otherwise }\\
\end{cases}$
\item  $p_{3}(\alpha(u)_{k,0} ,w(u)^{(0)}_{i,0},\alpha(u)_{j,0})=\begin{cases}
-\left( k<i<j\right)_{1}\oplus \mathcal{M}^{2} ,& k<i<j\\
-\left( j<i<k\right)_{1}\oplus \mathcal{M}^{2} ,& j<i<k\\
\mathcal{M}^{2} ,& \text{ otherwise }\\
\end{cases}$
\item  $p_{3}(\alpha(u)_{k,0} ,\alpha(u)_{j,0},w(u)^{(0)}_{i,0})=\begin{cases}
\left( j<i<k\right)_{1}\oplus \mathcal{M}^{2} ,& j<i<k\\
\mathcal{M}^{2} ,& \text{ otherwise }\\
\end{cases}$
\item  $p_{3}(\alpha(u)_{j,0} ,\alpha(u)_{k,0},w(u)^{(0)}_{i,0})=\begin{cases}
\left( k<i<j\right)_{1}\oplus \mathcal{M}^{2} ,& k<i<j\\
\mathcal{M}^{2} ,& \text{ otherwise }\\
\end{cases}$
\item $p_{3}(w(u)^{(0)}_{i,0}, (w(u)^{(0)}_{j,0},\alpha(u)_{k,0})=\begin{cases}
-\left( k<j<i\right)_{2}\oplus \mathcal{M}^{2} ,& j<i<k\\
\mathcal{M}^{2} ,& \text{ otherwise }\\
\end{cases}$
\item $p_{3}((w(u)^{(0)}_{j,0},w(u)^{(0)}_{i,0} ,\alpha(u)_{k,0})=\begin{cases}
-\left( k<i<j\right)_{2}\oplus \mathcal{M}^{2} ,& j<i<k\\
\mathcal{M}^{2} ,& \text{ otherwise }\\
\end{cases}$
\item  $p_{3}(w(u)^{(0)}_{i,0},\alpha(u)_{k,0} ,w(u)^{(0)}_{j,0})=\begin{cases}
\left( k<i<j\right)_{2}\oplus \mathcal{M}^{2} ,&  k<i<j\\
\left( k<j<i\right)_{2}\oplus \mathcal{M}^{2} ,&  k<j<i\\
\mathcal{M}^{2} ,& \text{ otherwise }\\
\end{cases}$
\item  $p_{3}(w(u)^{(0)}_{j,0},\alpha(u)_{k,0} ,w(u)^{(0)}_{i,0})=\begin{cases}
\left( k<i<j\right)_{2}\oplus \mathcal{M}^{2} ,&  k<i<j\\
\left( k<j<i\right)_{2}\oplus \mathcal{M}^{2} ,&  k<j<i\\
\mathcal{M}^{2} ,& \text{ otherwise }\\
\end{cases}$
\item $p_{3}(\alpha(u)_{k,0},w(u)^{(0)}_{j,0},w(u)^{(0)}_{i,0} )=\begin{cases}
-\left( k<j<i\right)_{2}\oplus \mathcal{M}^{2} ,& j<i<k\\
\mathcal{M}^{2} ,& \text{ otherwise }\\
\end{cases}$
\item $p_{3}(\alpha(u)_{k,0},w(u)^{(0)}_{i,0},w(u)^{(0)}_{j,0} )=\begin{cases}
-\left( k<i<j\right)_{2}\oplus \mathcal{M}^{2} ,& j<i<k\\
\mathcal{M}^{2} ,& \text{ otherwise }\\
\end{cases}$
\end{enumerate}
\end{prop}
\begin{proof}
Explicitly, the formula for $p_{3}$ is given by
\[
p_{3}=m_{2}\left( \left( h\otimes Id\right) \circ \left( m_{2}\otimes Id\right)\right) -m_{2}\left( \left( Id\otimes h\right) \circ \left( Id\otimes m_{2}\right)\right) +m_{3}
\]
and by Proposition \ref{generalrelation} the term $m_{3}$ vanishes. The results follows by a direct calculation applying the decomposition \eqref{vspgen}.
\end{proof}
The arguments of the proof above works also for the next two propositions.
\begin{prop}
 Let $i,j$ be distinct. For $m=3$ the p-kernels are as follows.
\begin{enumerate}
\item $p_{3}(w(u)^{(0)}_{i,0}, \alpha(u)_{j,0},\alpha(u)_{j,0})=\begin{cases}
\left( j<i,i\right)_{1}+D\left(w(u)^{(2)}_{j,0}-w(u)^{(2)}_{j,i}\right) \oplus \mathcal{M}^{2} ,& j<i\\
D\left(w(u)^{(2)}\right) \oplus \mathcal{M}^{2} ,& j=i\\
\mathcal{M}^{2} ,& \text{ otherwise }\\
\end{cases}$
\item  $p_{3}(\alpha(u)_{j,0},w(u)^{(0)}_{i,0}, \alpha(u)_{j,0})=\begin{cases}
-2\left( j<i,i\right)_{1}-2D\left(w(u)^{(2)}_{j,0}-w(u)^{(2)}_{j,i}\right) \oplus \mathcal{M}^{2} ,& j<i\\
-2D\left(w(u)^{(2)}_{j,0}\right) \oplus \mathcal{M}^{2} ,& j=i\\
\mathcal{M}^{2} ,& \text{ otherwise }\\
\end{cases}$
\item  $p_{3}( \alpha(u)_{j,0},\alpha(u)_{j,0},w(u)^{(0)}_{i,0})=\begin{cases}
\left( j<i,i\right)_{1}+D\left(w(u)^{(2)}_{j,0}-w(u)^{(2)}_{j,i}\right) \oplus \mathcal{M}^{2} ,& j<i\\
D\left(w(u)^{(2)}_{j,0}\right) \oplus \mathcal{M}^{2} ,& j=i\\
\mathcal{M}^{2} ,& \text{ otherwise }\\
\end{cases}$
\end{enumerate}
Let $i,j$ be distinct, then
\begin{enumerate}
\item $p_{3}(w(u)^{(0)}_{i,0}, \alpha(u)_{i,0},\alpha(u)_{j,0})\in\mathcal{M}^{2}$
\item  $p_{3}(w(u)^{(0)}_{i,0},\alpha(u)_{j,0}, \alpha(u)_{i,0})=\begin{cases}
\left( j<i,i\right)_{1}+D\left(w(u)^{(2)}_{i,0}\right) \oplus \mathcal{M}^{2} ,& j<i\\
\mathcal{M}^{2} ,& \text{ otherwise }\\
\end{cases}$
\item  $p_{3}(\alpha(u)_{i,0},w(u)^{(0)}_{i,0}, \alpha(u)_{j,0})=\begin{cases}
-\left( j<i,i\right)_{1}-D\left(w(u)^{(2)}_{i,0}\right) \oplus \mathcal{M}^{2} ,& j<i\\
\mathcal{M}^{2} ,& \text{ otherwise }\\
\end{cases}$
\item $p_{3}(\alpha(u)_{j,0},w(u)^{(0)}_{i,0}, \alpha(u)_{i,0})=\begin{cases}
-\left( j<i,i\right)_{1}-D\left(w(u)^{(2)}_{i,0}\right) \oplus \mathcal{M}^{2} ,& j<i\\
\mathcal{M}^{2} ,& \text{ otherwise }\\
\end{cases}$

\item  $p_{3}( \alpha(u)_{i,0},\alpha(u)_{j,0},w(u)^{(0)}_{i,0})=\begin{cases}
\left( j<i,i\right)_{1}+D\left(w(u)^{(2)}_{i,0}\right) \oplus \mathcal{M}^{2} ,& j<i\\
\mathcal{M}^{2} ,& \text{ otherwise }\\
\end{cases}$
\item $p_{3}(\alpha(u)_{j,0}, \alpha(u)_{i,0},w(u)^{(0)}_{i,0})\in\mathcal{M}^{2}$
\end{enumerate}
\end{prop}

\begin{prop}
 For $m=3$ the p-kernels are as follows.
\begin{enumerate}
\item $p_{3}(w(u)^{(0)}_{j,0}, w(u)^{(0)}_{j,0},\alpha(u)_{i,0})=\begin{cases}
-\left( i<j,j\right)_{2} \oplus \mathcal{M}^{2} ,& i<j\\
\mathcal{M}^{2} ,& \text{ otherwise }\\
\end{cases}$
\item  $p_{3}(w(u)^{(0)}_{j,0},\alpha(u)_{i,0}, w(u)^{(0)}_{j,0})=\begin{cases}
2\left( i<j,j\right)_{2} \oplus \mathcal{M}^{2} ,& i<j\\
\mathcal{M}^{2} ,& \text{ otherwise }\\
\end{cases}$
\item  $p_{3}( \alpha(u)_{i,0},w(u)^{(0)}_{j,0},w(u)^{(0)}_{j,0})=\begin{cases}
-\left( i<j,j\right)_{2}\oplus \mathcal{M}^{2} ,& i<j\\
\mathcal{M}^{2} ,& \text{ otherwise }\\
\end{cases}$
\item $p_{3}( \alpha(u)_{i,0},\alpha(u)_{j,0},w(u)^{(0)}_{i,0})=\begin{cases}
\left( i<j,j\right)_{2}\oplus \mathcal{M}^{2} ,& i<j\\
\mathcal{M}^{2} ,& \text{ otherwise }\\
\end{cases}$
\item  $p_{3}(w(u)^{(0)}_{j,0},w(u)^{(0)}_{i,0},\alpha(u)_{i,0})\in\mathcal{M}^{2}$

\item  $p_{3}(w(u)^{(0)}_{i,0},\alpha(u)_{i,0},w(u)^{(0)}_{j,0} )=\begin{cases}
\left( i<j,j\right)_{2} \oplus \mathcal{M}^{2} ,& i<j\\
\mathcal{M}^{2} ,& \text{ otherwise }\\
\end{cases}$

\item  $p_{3}(w(u)^{(0)}_{j,0},\alpha(u)_{i,0},w(u)^{(0)}_{i,0} )=\begin{cases}
\left( i<j,j\right)_{2} \oplus \mathcal{M}^{2} ,& i<j\\
\mathcal{M}^{2} ,& \text{ otherwise }\\
\end{cases}$
\item  $p_{3}(\alpha(u)_{i,0},w(u)^{(0)}_{i,0},w(u)^{(0)}_{j,0})\in\mathcal{M}^{2}$

\item  $p_{3}(\alpha(u)_{i,0},w(u)^{(0)}_{j,0},w(u)^{(0)}_{i,0} )=\begin{cases}
-\left( i<j,j\right)_{2} \oplus \mathcal{M}^{2} ,& i<j\\
\mathcal{M}^{2} ,& \text{ otherwise }\\
\end{cases}$
\end{enumerate}
\end{prop}
\begin{prop}
Let $k>2$
\begin{enumerate}
\item  $p_{k+1}(w(u)^{(0)}_{i,0},\underset{k}{\underbrace{\alpha(u)_{j,0},\dots, \alpha(u)_{j,0}}})=\begin{cases} (-1)^{k}\left( w(u)^{(k)}_{j,0}-w(u)^{(k)}_{i,j}\right)\oplus \mathcal{M}^{2} ,& j<i\\ 
\mathcal{M}^{2} ,& \text{ otherwise }\\
\end{cases}$
\item  $p_{k+1}(w(u)^{(0)}_{i,0},\underset{k}{\underbrace{\alpha(u)_{j,0},\alpha(u)_{i,0},\dots, \alpha(u)_{i,0}}})=\begin{cases} (-1)^{k} w(u)^{(k)}_{i,0}\oplus \mathcal{M}^{2} ,& j<i\\ 
\mathcal{M}^{2} ,& \text{ otherwise }\\
\end{cases}$
\item  $p_{k+1}(w(u)^{(0)}_{i,0},\underset{k}{\underbrace{\alpha(u)_{i,0},\dots, \alpha(u)_{i,0}}})=\begin{cases} (-1)^{k} w(u)^{(k)}_{i,0}\oplus \mathcal{M}^{2} ,& j<i\\ 
\mathcal{M}^{2} ,& \text{ otherwise }\\
\end{cases}$
\item Let $$x_{1}, \dots, x_{k+1}\in \left\lbrace w(u)_{i,j}^{(0)}, \alpha(u)_{i,j} \text{ for }\,i,j=0,1,\dots,n.\right\rbrace $$ such that $(x_{1}, x_{2}, \dots ,x_{k+1})$ is not a permutation of either $(w(u)^{(0)}_{i,0},\alpha(u)_{j,0},\dots, \alpha(u)_{j,0}) $\\ either $(w(u)^{(0)}_{i,0},\alpha(u)_{j,0},\alpha(u)_{i,0},\dots, \alpha(u)_{i,0})$ or $(w(u)^{(0)}_{i,0},\alpha(u)_{i,0},\dots, \alpha(u)_{i,0}) $. Then $$p_{k+1}(x_{1}, x_{2}, \dots ,x_{k+1})\in \mathcal{M}^{2}.$$
\end{enumerate}
\end{prop}
\begin{proof}
 Since $m_{k+1}\left( \left( B_{n}^{1}\right)^{\otimes k+1} \right)\subset \mathcal{M}^{2} $ for $k>2$, we conclude that $p_{T}(w(u)^{(0)}_{i,0},\alpha(u)_{j,0},\dots, \alpha(u)_{j,0})=0$ if $T$ is not binary. The condition \eqref{sidecondition} implies that the only tree such that $$p_{T}(w(u)^{(0)}_{i,0},\alpha(u)_{j,0},\dots, \alpha(u)_{j,0})\neq 0$$ is the one in Figure \ref{fig2}. A direct calculation by using the decomposition \eqref{vspgen} gives the desired result. The same argument works in the other cases as well.
\end{proof}


\begin{thebibliography}{}
\bibitem{Arnold} V. I. Arnold, \emph{The cohomology ring of the colored braid group}, Translated from Matematicheskie Zametki, Vol. 5, No. 2 (1969), pp. 227–231.
\bibitem{Be1}D. Bernard, \emph{On the Wess-Zumino-Witten Models on the Torus},  Nucl.Phys. B303 (1988), pp. 77-93.
\bibitem{Be2}D. Bernard, \emph{On the Wess-Zumino-Witten models on Riemann surfaces}, Nucl.Phys. B309 (1988), pp. 145–174.
\bibitem{LevBrown} F. C. S. Brown, Andrey Levin. \emph{ Multiple Elliptic Polylogarithms}, arXiv:1110.6917.
\bibitem{Damien} D. Calaque, B. Enriquez, P. Etingof\emph{ Universal KZB equations: the elliptic case}, in Algebra, arithmetic, and geometry: in honor of Yu. I. Manin. Vol. I, Progr. Math., 269 (2009), pp. 165–266.
\bibitem{extensionChen} K.-T. Chen, \emph{Extension of }$C^{\infty}$\emph{ function algebra by Iterated Integrals and Malcev completion of }${\pi}_{1}$, Advances in Mathematics, Volume 23, Issue 2 (1977), pp. 181-210.
\bibitem{IteratedChen} K.-T. Chen, \emph{Iterated path integrals}, Bull. Amer. Math. Soc. Volume 83, Number 5 (1977),pp. 831-879.
\bibitem{Prelie} V. Dotsenko, S. Shadrin, B. Vallette \emph{Pre-Lie deformation theory}, Moscow Mathematical Journal, Volume 16, Issue 3 (2016), pp. 505-543.
\bibitem{Dupont2} J. L. Dupont,\emph{ Simplicial De Rham cohomology and characteristic classes of Flat Bundles},  Topology 15 (1976), pp. 233-245.
\bibitem{Getz} E. Getzler, X. Z. Cheng, \emph{Transferring homotopy commutative algebraic structures}, Journal of Pure and Applied Algebra, Volume 212, Issue 11 (2008), pp. 2535-2542.
\bibitem{Hain} R. Hain, \emph{Notes on the Universal Elliptic KZB Equation}, arXiv:1309.0580
\bibitem{KadeshHuebsch}J. Huebschmann, T.V. Kadeishvili, \emph{Small models for chain algebras}, Mathematische Zeitschrift, Volume 207, Issue 1 (1991), pp. 245-280.
\bibitem{Huebsch}J. Huebschmann, \emph{On the construction of $A_{\infty}$-structures}, Festschrift in honor of T. Kadeishvili's 60-th birthday, Georgian J. Math. 17 (2010), pp. 161-202.
\bibitem{Kadesh}T.V. Kadeishvili, \emph{On the homology theory of fibre spaces}, Uspekhi Mat. Nauk 35:3 (1980), pp. 183-188.
\bibitem{KZ} V. G. Knizhnik and A. B. Zamolodchikov, \emph{Current algebra and Wess-Zumino model
	in two dimensions}, Nuclear Phys, B, 247, Volume 1 (1984), pp. 83–103.
\bibitem{kontsoibel} M. Kontsevich, Y. Soibelman, \emph{Homological mirror symmetry and torus fibrations}, Symplectic geometry and mirror symmetry. Proceedings, 4th KIAS Annual International Conference, Seoul, South Korea, August 14-18 (2000).
\bibitem{lambestaheff} L. Lambe, J. Stasheff, \emph{Application to perturbation theory to iterated fibrations}, Manuscripta Mathematica , Volume 58, Issue 3 (1987), pp. 363-376. 
\bibitem{Levrac} A. Levin, G. Racinet, \emph{Towards Multiple Elliptic Polylogarithm}, arXiv:math/0703237 
\bibitem{lodayVallette} J. Loday, B. Vallette, \emph{Algebraic Operads}, Grundlehren der mathematischen Wissenschaften, Volume 346 Springer-Verlag (2012).
\bibitem{Lothaire} M. Lothaire, \emph{Combinatorics on Words}, Cambridge University Press, (1997).
\bibitem{Markl} M. Markl, \emph{Transferring $A_{\infty}$ (strongly homotopy associative) structures}, Proceedings of the 25th Winter School ``Geometry and Physics", Publisher: Circolo Matematico di Palermo, pp. 139-151.
\bibitem{Segal} G. B. Segal, \emph{Classifying spaces and spectral sequences}, Publ. Math. IHES 34 (1968), pp. 105–112.
\bibitem{Sibilia1} C. Sibilia, \emph{$1$-minimal models for $C_{\infty}$-algebras and flat connections}, in preparation.
\bibitem{Don} D. Zagier, \emph{Periods of modular forms and Jacobi theta functions}, Invent. Math., 104 (1991), pp. 449-465.

\end{thebibliography}
\end{document}